\documentclass[11pt]{article}
\usepackage{amsmath}
\usepackage{fullpage}
\usepackage{amsthm}
\usepackage{hyperref}
\usepackage{amssymb}

\usepackage{mathtools}
\usepackage{bbm}
\usepackage{bm}

\usepackage{graphicx}

\usepackage{xcolor}

\usepackage{authblk}

\usepackage[authoryear, round]{natbib}

\graphicspath{{Fig/}}

\newtheorem{theorem}{Theorem}%
\newtheorem{proposition}{Proposition}%

\newtheorem{lemma}{Lemma}

\newtheorem{remark}{Remark}%

\newtheorem{definition}{Definition}

\newtheorem{procedure}{Procedure}
\newtheorem{assumption}{Assumption}

\newcommand{\Var}{\mathrm{Var}}
\newcommand{\rank}{\mathrm{rank}}
\newcommand{\sign}{\mathrm{sign}}

\newcommand{\E}{\mathbb{E}}
\newcommand{\Prob}{\mathbb{P}}
\newcommand{\hatED}{\hat{\mathbb E}_{\mathcal D}}
\newcommand{\EDP}{\mathbb E_{\mathcal D, P}}
\newcommand{\EP}{\mathbb E_P}
\newcommand{\VDP}{\mathcal V_{\mathcal D, P}}
\newcommand{\VP}{\mathcal V_{ P}}
\newcommand{\oP}{o_{\mathcal P}}
\newcommand{\oPW}{o_{\mathcal P,\mathcal W}}
\newcommand{\OP}{O_{\mathcal P}}
\newcommand{\sigmamin}{\sigma_{\mathrm{min}}}

\newcommand{\sumD}{\sum_{i\in \mathcal D}}

\author[1]{Cyrill Scheidegger}
\author[1,2]{Malte Londschien}
\author[1]{Peter B\"uhlmann}
\affil[1]{Seminar for Statistics, ETH Zurich}
\affil[2]{ETH AI Center, ETH Zurich}

\begin{document}
\title{ Machine-Learning-Powered Specification Testing in Linear Instrumental Variable Models}
\maketitle

\begin{abstract}
    The linear instrumental variable (IV) model is widely used in observational studies, yet its validity hinges on strong assumptions. Classical specification tests such as the Sargan--Hansen J test are limited to overidentified settings and are therefore not applicable in the common just-identified case, where the number of instruments is equal to the number of endogenous variables.
We propose a novel test for the well-specification of the linear IV model 
under the assumption that the structural error is mean independent of the instruments. This assumption enables specification testing even in the just-identified setting. 
Our approach uses the idea of residual prediction: if the two-stage least squares residuals can be predicted from the instruments better than chance, this indicates misspecification. The resulting test employs sample splitting and a user-chosen machine learning method, and we show asymptotic type I error control and consistency against a broad class of alternatives.

We further show how the proposed testing principle can be adapted to settings with weak or many instruments via an Anderson--Rubin-type inversion, thereby substantially extending the applicability.
The tests accommodate heteroskedasticity- and cluster-robust inference and are
implemented in the R package \texttt{RPIV} and the \texttt{ivmodels} software package for Python.
\end{abstract}

\section{Introduction}
The linear instrumental variable (IV) model is arguably one of the most popular statistical models for estimating causal effects in observational studies. Instrumental variable regression allows identification and inference for the causal parameter even when the explanatory variables and the structural error are correlated due to hidden confounding or omitted variables, provided that there exists an instrumental variable that is correlated with the explanatory variables and uncorrelated with the structural error.
Assuming $n$ i.i.d.\ observations of $(y_i, {\bm x}_i, {\bm z}_i)$, the linear IV model is 
\begin{equation}\label{eq_StandardLinearIV}
    y_i = {\bm x}_i^T{\bm{\beta}} + u_i, \quad \E[{\bm z}_iu_i] = 0, \quad i = 1,\ldots, n,
\end{equation}
where $y_i$ is the univariate response, ${\bm x}_i\in \mathbb R^p$ are the explanatory variables, ${\bm z}_i\in \mathbb R^d$ $(d\geq p)$ are the instruments, and $u_i\in \mathbb R$ is the univariate error with mean zero \citep[see, e.g.,][]{WooldridgeEconometricAnalysis}. Correct inference relies on three main assumptions: (i) relevance, i.e., $\rank(\E[{\bm x}_i{\bm z}_i^T])=p$, (ii) linearity, i.e., ${\bm x}_i$ affects the outcome $y_i$ linearly, and (iii) exogeneity, i.e., $\E[{\bm z}_iu_i]=0$. If relevance, linearity, and exogeneity hold, then ${\bm{\beta}}$ can be consistently estimated using the two-stage least squares estimator.
While (i)--(iii) are crucial for valid inference, they are often at most approximately satisfied in practice.

This underscores the need for simple and effective diagnostic tools that can help practitioners assess whether the linear IV model is appropriate for their analysis. Relevance can be tested using, for example, the first-stage F-statistic \citep{AndersonEstimatingLinearRestrictions, StaigerIVRegressionWithWeakInstruments}. If there are more instruments than endogenous variables ($d>p$), the assumptions (ii) and (iii) can be jointly tested using the Sargan--Hansen J test \citep{SarganEstimationEconomicRelationships, HansenLargeSampleProperties}. However, in practice, the so-called just-identified setting with the same number of instruments and explanatory variables ($d = p$) is common. In that setting, the two-stage least squares residuals are empirically uncorrelated with the instruments by construction, independently of whether (ii) and (iii) hold, and (ii) and (iii) cannot be verified from data.

When practitioners argue for the validity of (ii) and (iii), they typically appeal to the absence of any systematic relationship between the instruments and unobserved factors influencing the outcome rather than merely the absence of linear correlation \citep[see also][]{DieterleASimpleDiagnostic}. Hence, the uncorrelatedness $\E[{\bm z}_iu_i]=0$ in (iii) is often replaced by the mean independence $\E[u_i|{\bm z}_i]=0$ a.s. 
The assumption $\E[u_i|{\bm z}_i] = 0$ also opens some possibilities for specification testing, i.e., testing the assumptions (ii) and (iii), including the just-identified case. \cite{DieterleASimpleDiagnostic} consider the special case of a single explanatory variable ${x}_i$ and single instrument ${z}_i$. By adding ${z}_i^2$ as a second instrument, they make the problem overidentified, which enables them to use the Sargan--Hansen J test to test for certain forms of misspecification, that is, violations of (ii) or (iii).
In this work, we exploit the mean independence assumption $\E[u_i|{\bm z}_i]=0$ to develop a new and conceptually simple test for the well-specification of the linear IV model based on the idea of residual prediction \citep{ShahGoodnessOfFit, JankovaGOFTestingInHDGLM}.

\subsection{Our Contribution}
Our primary contribution is a new residual-prediction-based test for the well-specification of the linear instrumental variable model that is applicable even in just-identified settings under mean independence. Informally, our test seeks to answer the question ``Is a linear instrumental variable model appropriate for the data at hand?'' Formally, given data $({\bm x}_i, y_i, {\bm z}_i)_{i=1,\ldots, n}$ we test the null hypothesis
\begin{equation}
    \label{eq_H0Effective}
    H_0:\exists {\bm{\beta}}\in \mathbb R^p \text{ such that } \E[y_i-{\bm x}_i^T{\bm{\beta}}|{\bm z}_i]= 0 \text{ a.s.}
\end{equation}
If the linear instrumental variable model is well-specified with $\E[u_i|{\bm z}_i]=0$, then such a ${\bm{\beta}}$ exists by definition. However, we will see in Section \ref{sec_WhenIs} that there are also scenarios where the linear instrumental variable model is misspecified, but nevertheless, such a ${\bm{\beta}}$ exists. Hence, there exist misspecifications for which our test has no power.

Our test uses the idea of \textit{residual prediction} \citep{ShahGoodnessOfFit, JankovaGOFTestingInHDGLM}. Under $H_0$, the parameter ${\bm{\beta}}$ can be consistently estimated using two-stage least squares. Hence, the residuals $\hat r_i = y_i - {\bm x}_i^T\hat {\bm{\beta}}$ should be close to the true errors $u_i$ that satisfy $\E[u_i|{\bm z}_i]=0$. If one attempts to predict the residuals $\hat r_i$ using machine learning, there should be little systematic signal that can be learned. If, on the other hand, the machine learning algorithm can predict the residuals $\hat r_i$ well from the instruments ${\bm z}_i$, this points to
a violation of $\E[u_i|{\bm z}_i]=0$. Using sample splitting and machine learning on one part of the sample, we will construct a test statistic that is asymptotically normal with variance that can be consistently estimated.

While the proposed test is valid under standard IV asymptotics, an important additional challenge arises in settings with weak or many instruments, where two-stage least squares is unreliable.
To address this limitation, we introduce a second version of our test that does not rely on an initial estimate of the parameter.
Concretely, for a given $\bm\beta_0\in \mathbb R^p$, we test the null hypothesis
\begin{equation}\label{eq_H0Weak}
    H_0(\bm \beta_0): \E[y_i - \bm x_i^T\bm \beta_0|\bm z_i] = 0 \text{ a.s.}
\end{equation}
That is, we simultaneously test the well-specification of the model and the value of the parameter $\bm \beta_0$. By inverting the test, one can obtain the set of all $\bm \beta\in \mathbb R^p$ that are compatible with well-specification for a given significance level. This set can alternatively be interpreted as a confidence set for $\bm \beta$ under the assumption of well-specification. An additional benefit of this formulation is that it avoids the problem of selective inference, i.e., the difficulty in the interpretation of the p-values when parameter inference is done only after the specification test did not reject. Hence, this second approach provides a unified framework for assessing model specification and conducting inference that remains valid under weak identification: if the model is correctly specified, the confidence interval has asymptotically correct coverage regardless of IV strength, whereas an empty confidence set signals misspecification.

For both our main test and the weak-IV-robust extension, we provide heteroskedasticity- and cluster-robust variance estimators. A major innovation of our tests compared to existing approaches (see below) is that they employ user-chosen machine learning.
This leads to high power against a wide range of alternatives, which approaches relying on fixed transformations of instruments cannot achieve. Still, for asymptotic type I error control, we do not require the machine learning algorithm to be consistent, and for the test to have power, it is enough if the machine learning algorithm picks up some signal. This is an appealing property, since it allows to employ
a user-chosen machine learning algorithm in an off-the-shelf fashion, without worrying if the particular machine learning algorithm is consistent for the given data at hand.

Implementations of our tests can be found in the R package \texttt{RPIV}, available on CRAN and as part of the \texttt{ivmodels} software package for Python \citep{londschien2025statistician}, available at \url{https://github.com/mlondschien/ivmodels}.

\subsection{Additional Related Literature}
If one is not willing to impose any additional assumption to \eqref{eq_StandardLinearIV}, then specification testing is only possible in the overidentified case. In this setting, the J test by \cite{SarganEstimationEconomicRelationships}, later generalized by \cite{HansenLargeSampleProperties}, is extremely popular. The test also turns out to be equivalent to the idea by \cite{HausmanSpecificationTests}, where one compares the estimator that uses all instruments with an estimator using only a subset of the instruments \citep{WooldridgeEconometricAnalysis}. In addition to being applicable only in the overidentified case, it has already been noted by \cite{NeweyGMMSpecificationTesting} that there are certain directions of misspecification against which the J test does not have power.

With the stronger mean independence assumption, our approach fits into the framework of specification testing for models defined by conditional moment restrictions, i.e., models of the form $\E[g({\bm v}_i, {\bm{\beta}})|{\bm z}_i]=0$, where ${\bm v}_i$ and ${\bm z}_i$ can have overlapping components. Our setting of the linear instrumental variable model is a special case of this obtained by setting ${\bm v}_i = (y_i, {\bm x}_i)$ and $g({\bm v}_i, {\bm{\beta}}) = y_i-{\bm x}_i^T{\bm{\beta}}$. In Supplementary Appendix G, we sketch how our approach can be extended to general conditional moment restrictions.

For specification testing, by far the most attention has been given to the special case of testing the well-specification of a parametric regression function (i.e., without the need for instrumental variables), see for example \cite{GonzalesAnUpdatedReview} for a review. In this context, the idea of using nonparametric predictions of residuals is not new; see, for example, \cite{ZhengAConsistenTestOfFunctionalForm, HartNonparametricSmoothing, GuerreDataDrivenRateOptimal}. This idea has also been extended to models defined by arbitrary conditional moment restrictions, see \cite{DelgadoConsistentTestsOfConditionalMomentRestrictions}, and can hence be applied to the linear IV model. It should be noted that these approaches use kernel smoothing, which -- from a theoretical perspective -- restricts the instruments to low-dimensional continuous random variables and limits possible extensions to many instruments. Moreover, they rely on the choice of bandwidths for which no clear guidelines exist. For other works that consider general conditional moment restrictions, see, for example, \cite{WhangConsistentSpecificationTestingForCMR} for an approach based on weighting of the conditional moment restriction, \cite{LavergneAHausmanSpecificationTest} for a test using the Hausman idea of comparing two different consistent estimators, or \cite{EscancianoAGaussianProcessApproachToModelChecks} for a test using Gaussian processes and references therein.

Conceptually, our approach is most closely related to the idea of \cite{DieterleASimpleDiagnostic}, who also exploit the stronger assumption $\E[u_i | {\bm z}_i] = 0$ by additionally including the square of the instrument. However, our method differs in several important ways. Rather than relying on a fixed transformation such as ${z} \mapsto {z}^2$, we employ
user-chosen machine learning to learn a data-driven transformation of the instrument, which makes our test powerful against a wide range of alternatives and renders
our approach applicable beyond settings with a single continuous instrument.

Like our main test, all of these methods share the requirement of sufficient instrument strength or on a consistent estimator of the regression parameter. Hence, such tests are of limited use for the case of weak or many instruments. Our weak-instrument-robust extension is fundamentally related to the Anderson--Rubin test \citep{AndersonRubinTest} that jointly tests the value of the parameter and the specification. Like the J test, the Anderson--Rubin test is merely based on uncorrelatedness and not mean independence of the error term, and hence its ability to detect misspecification is limited, especially in the just-identified case. In the literature, less attention has been devoted to this setting. Notable exceptions are \cite{JunSemiparamtericTests, JunTestingUnderWeakIdentification} for approaches based on k-nearest neighbor regression (which limits the application to a small number of instruments) and \cite{AntoineIdentificationRobustNonparametricInference}, who propose a procedure based on the integrated conditional moments principle by \cite{BierensConsistentModelSpecification} and rely on kernel smoothing for the estimation of the variance, again limiting its use to low-dimensional instruments. In contrast, we employ
user-chosen machine learning to detect signal in the residuals, which naturally accommodates large-dimensional instruments.

While our tests use the residual prediction idea from \cite{ShahGoodnessOfFit} and, in particular, \cite{JankovaGOFTestingInHDGLM}, the way we analyze our test statistics is different, since their tests were aimed at high-dimensional (generalized) linear models, coming with different challenges than the linear instrumental variable model.

\subsection{Outline}
In Sections \ref{sec_MotivationMethod} and \ref{sec_TheoryIID}, we present the methodology and theory of our main test. In Section \ref{sec_ExtensionWeak}, we present the extension of the proposed testing principle to the weak- and many-IV setting. In Section \ref{sec_Experiments}, we evaluate the performance of our test on synthetic data and the well-known returns-to-education problem \citep{CardCollege}. Additional results, including cluster-robust inference as well as proofs, additional simulations, and an additional real-world dataset can be found in the Supplementary Appendix.

\subsection{Notation and Conventions}
We write $\Phi(\cdot)$ for the cumulative distribution function of a standard normal distribution. Moreover, we write $\sigmamin(\bm A)$ for the smallest singular value of a matrix $\bm A$ and $\|\bm A\|_{op}$ for the largest singular value of a matrix $\bm A$, i.e., the operator/spectral norm of $\bm A$. Finally, expressions involving conditional expectations are always to be understood almost surely, even if not stated explicitly.

\section{Motivation and Method}\label{sec_MotivationMethod}
We assume that we have data $({\bm x}_i, y_i, {\bm z}_i)_{i=1, \ldots, n}$ with responses $y_i\in \mathbb R$, explanatory variables ${\bm x}_i\in \mathbb R^p$ (potentially endogenous) and instruments ${\bm z}_i\in \mathbb R^d$ with $d\geq p$. 
As explained in the introduction, the null hypothesis \eqref{eq_H0Effective} is implied by the linear instrumental variable model with a mean independence assumption for the error term, and hence, our test will lead to a well-specification test for the linear instrumental variable model.
\begin{remark}
    Additional exogenous explanatory variables (``exogenous controls''), including an intercept, can be included by adding them to both ${\bm x}_i$ and ${\bm z}_i$.\footnote{For the weak- and many-IV extension in Section \ref{sec_ExtensionWeak}, we will introduce controls explicitly because there, low-dimensional components of $\bm \beta$ are targeted.}
\end{remark}

To test $H_0$, we use the following idea. If the model is well-specified, the two-stage least squares residuals $\hat r_i$ should not be predictable from the instruments. We therefore measure their empirical correlation with a learned nonlinear function $\hat w(\bm z_i)$ of the instruments and base a test statistic on a normalized version of $\sum_i \hat r_i \hat{w}({\bm z}_i)$.

For operationalization of this idea, we proceed as follows.
We consider splitting
our sample into two disjoint subsamples, $\{1,\ldots, n\} = \mathcal D_A\cup\mathcal D$, with the auxiliary sample $\mathcal D_A$ of size $n_A$ and the main sample $\mathcal D$ of size $n_0 = n-n_A$. The idea is to learn a function $\hat w(\cdot)$ on $\mathcal D_A$ and calculate the test statistic on $\mathcal D$. For the asymptotic results, we need that $n_0\to\infty$ as $n\to\infty$.

We first explain how to obtain the function $\hat w(\cdot)$. For the theory under $H_0$ presented below, it is not important how $\hat w$ was obtained but we only need that $|\hat w|$ is bounded by $1$ and that it is independent of the main sample $\mathcal D$ (with ``the sample $\mathcal D$'', we mean ``the sample $({\bm x}_i, y_i, {\bm z}_i)_{i\in \mathcal D}$''). For power, we will need $\hat w$ to capture some signal under alternatives. We suggest the following procedure.

\begin{procedure}[Estimation of weight function]\label{proc_EstW}
Input: auxiliary sample $({\bm x}_i, y_i, {\bm z}_i)_{i\in \mathcal D_A}$.
    \begin{enumerate}
    \item Let $\hat {\bm{\beta}}_A$ be the two-stage least squares estimator for ${\bm{\beta}}$ using the sample $\mathcal D_A$. Let $\hat r_i^A = y_i- {\bm x}_i^T\hat {\bm{\beta}}_A$, $i\in \mathcal D_A$ be the corresponding residuals.
    \item With a user-chosen machine learning method, regress $(\hat r_i^A)_{i\in \mathcal D_A}$ on $({\bm z}_i)_{i\in \mathcal D_A}$. Call the resulting function $\hat w_0(\cdot)$.
    \item For some $K>0$, define $\hat w(\cdot) = \frac{1}{K}\sign(\hat w_0(\cdot))\min(|\hat w_0(\cdot)|, K)$. Return $\hat w(\cdot)$.
\end{enumerate}
\end{procedure}

The third step ensures that $|\hat w(\cdot)|\leq 1$. In practice, we recommend to choose $K$ as some high (e.g., $80\%$- or $90\%$-) quantile of the absolute value of the in-sample fitted values $(|\hat w_0({\bm z}_i)|)_{i\in \mathcal D_A}$. In the second step, any machine learning algorithm can be used. In the applications in Section \ref{sec_Experiments}, we will use random forests \citep{BreimanRandomForest}.

In the following, we define our test statistic given the function $\hat w(\cdot)$. For this, we introduce the following notation: for any $f \colon \mathbb R^{p + 1 + d} \to V$ for some vector space $V$, define
\begin{equation}\label{eq_DefHatED}
    \hatED[f({\bm x}, y, {\bm z})] \coloneqq \frac{1}{n_0}\sum_{i\in \mathcal D}f({\bm x}_i, y_i, {\bm z}_i).
\end{equation}
That is, $\hatED$ is simply the sample average over the sample $\mathcal D$.
Then, we can write the two-stage least squares estimator with instruments $\bm z$ as \citep[see, e.g.,][Section 5.2]{WooldridgeEconometricAnalysis}
\begin{equation}
    \hat{\bm{\beta}}\coloneqq\hat{\bm M} \hatED[{\bm z}y]\label{eq_Sam2SLSGen},
\end{equation}
with 
\begin{equation}
    \hat {\bm M} \coloneqq \left(\hatED[{\bm x} {\bm z}^T]\hatED[{\bm z} {\bm z}^T]^{-1} \hatED[{\bm z} {\bm x}^T]\right)^{-1}\hatED[{\bm x} {\bm z}^T]\hatED[{\bm z} {\bm z}^T]^{-1},\label{eq_DefHatMGen}
\end{equation}
provided the inverses exist. Using $\hat{\bm{\beta}}$, we define the residuals
\begin{equation}
    \hat r_i \coloneqq y_i - {\bm x}_i^T\hat {\bm{\beta}}, \, i\in \mathcal D.\label{eq_Sam2SLSResGen}
\end{equation}
For any function $w:\mathbb R^{d}\to\mathbb R$, we can then define the quantity
\begin{equation}\label{eq_DefNStat}
    N(w) \coloneqq \frac{1}{\sqrt n_0}\sum_{i\in \mathcal D}w({\bm z}_i)\hat r_i
\end{equation}
and $N(\hat w)$ will be the numerator of our test statistic.

For a fixed bounded weight function $w$, the statistic $N(w)$ is asymptotically normal under $H_0$. Intuitively, under $H_0$, there exists $\bm \beta$ such that $y_i = \bm x_i^T\bm \beta + u_i$ with $\E[u_i|\bm z_i] = 0$, so $\E[w(\bm z_i)u_i]=0$ and a central limit theorem applies to $\sum_{i\in \mathcal D} w(\bm z_i)u_i/\sqrt{n_0}$. Replacing $u_i$ by the two-stage least squares residual $\hat r_i$ introduces an additional correction term due to estimating $\bm \beta$, which we need to account for in the variance estimator $\hat \sigma_w^2$ introduced below.

An additional complication arises since our weight $\hat w$ is learned on an auxiliary sample and is therefore random and changes with $n$. Under $H_0$, the learning step effectively fits noise, and we have no control over how the output of Procedure \ref{proc_EstW} looks like. Hence, it is possible that the implied asymptotic variance becomes very small for some realizations $\hat w$, which invalidates the Gaussian approximation. To prevent this degeneracy, we introduce an artificial lower bound $\gamma$ on the variance. This leads to the following procedure to obtain a p-value for $H_0$.

\begin{procedure}[p-value]\label{proc_Test}
    Input: Auxiliary sample $({\bm x}_i, y_i, {\bm z}_i)_{i\in \mathcal D_A}$, main sample $({\bm x}_i, y_i, {\bm z}_i)_{i\in \mathcal D}$, variance estimator $w\mapsto \hat \sigma_w^2$ for the variance of $N(w)$ defined in \eqref{eq_DefNStat} and a fixed lower bound $\gamma>0$.
\begin{enumerate}
    \item Obtain $\hat w$ on the auxiliary sample $\mathcal D_A$ using Procedure \ref{proc_EstW}.
    \item Return $p_{val}=p_{val}(\hat w) = 1-\Phi\left(N(\hat w)/\max(\hat\sigma_{\hat w}, \sqrt{\gamma})\right)$.
\end{enumerate}
\end{procedure}

Observe that we use a one-sided p-value, since $\hat w({\bm z}_i)$ is constructed to be positively correlated with the residuals. In practice, $\gamma$ should be taken proportionally to the noise level, which can be estimated by $\frac{1}{n_0}\sum_{i\in \mathcal D}\hat r_i^2$. Then, we found that in empirical examples, the lower truncation bound $\sqrt{\gamma}$ is rarely activated. Also note that if we treated the auxiliary sample $\mathcal D_A$ (and hence also $\hat w$) as fixed, the truncation with $\sqrt \gamma$ would not be needed.

If the data are i.i.d., the variance of $N(w)$ can be estimated by
\begin{equation}\label{eq_DefHatSigmaW}
    \hat\sigma_w^2\coloneqq \frac{1}{n_0}\sum_{i\in \mathcal D}(w({\bm z}_i)+\hat {\bm a}_w^T{\bm z}_i)^2\hat r_i^2 - \left(\frac{1}{n_0}\sum_{i\in \mathcal D}w({\bm z}_i)\hat r_i\right)^2, 
\end{equation}
where
\begin{equation}\label{eq_DefHatAw}
    \hat {\bm a}_w^T \coloneqq - \hatED[w({\bm z}){\bm x}^T] \hat {\bm M}
\end{equation}
is the correction term accounting for the estimation of $\bm \beta$.
We will see that $\hat \sigma_w^2$ defined in \eqref{eq_DefHatSigmaW} is a consistent estimator for the asymptotic variance of $N(w)$ both under the null and the alternative hypothesis. 
Moreover, $\hat\sigma_w^2$ is heteroskedasticity-robust \citep{WhiteAHeteroskedasticityConsistent}. If one wants to make the homoskedasticity assumption a part of the null hypothesis, one could use the variance estimator
\begin{equation}\label{eq_DefHatSigmaWHomoscedastic}
    \hat\sigma_{w, \text{hom}}^2\coloneqq \frac{1}{n_0}\sum_{i\in \mathcal D}(w({\bm z}_i)+\hat {\bm a}_w^T{\bm z}_i)^2\frac{1}{n_0}\sum_{i\in \mathcal D}\hat r_i^2,
\end{equation}
which will be consistent under the null hypothesis \eqref{eq_H0Effective} with $\E[u_i^2|{\bm z}_i]=\text{const.}$
\begin{remark}
To increase power, one may replace $\hat r_i$ in \eqref{eq_DefHatSigmaW} (and \eqref{eq_DefHatSigmaWHomoscedastic}) by the OLS residuals from regressing $\hat r_i$ on $(\hat w({\bm z}_i),{\bm z}_i)$, which is sometimes also done for the overidentifying restriction J test \citep{BasmannOnFiniteSampleDistributions, BaumIVAndGMM}. This does not affect asymptotic validity under $H_0$, and in our experiments, it yields only marginal power gains, so we do not pursue it further.
\end{remark}
Extensions to clustered data with many independent clusters and within-cluster dependence \citep{LiangLongitudinalDataAnalyisis, WhiteAsymptoticTheory} are provided in Supplementary Appendix B.

The test, as defined in this section, uses two-stage least squares residuals and therefore presumes the usual conditions under which two-stage least squares is well-behaved. Section \ref{sec_TheoryIID} establishes the formal validity theory for this baseline procedure. Section \ref{sec_ExtensionWeak} then removes the strong identification requirement by introducing an Anderson--Rubin-type extension that remains valid with weak or many instruments.

\section{Theory in the i.i.d.\ Case}\label{sec_TheoryIID}
In this section, we present the theory of our test assuming i.i.d. data. We show that, after appropriate normalization, the statistic is asymptotically normal uniformly over a collection of data-generating distributions and bounded weight functions. This yields valid specification testing under the null and consistency under alternatives. 

We use the following setup: for a distribution $P\in \mathcal P$ for $({\bm x}, y, {\bm z})\in\mathbb R^{p + 1 + d}$, we assume that $({\bm x}_i, y_i, {\bm z}_i)_{i=1,\ldots, n}$ are i.i.d.\ copies of $({\bm x}, y, {\bm z})$ under $P$.
We write $\EP[\cdot]$ for the expectation under $P$. We assume
the following conditions.

\begin{assumption}\label{ass_SigmaMinIID}
    There exists $c>0$ such that for all $P\in \mathcal P$, $\sigma_{\min}(\EP[{\bm z}{\bm z}^T])\geq c$ and $\sigma_{\min}(\EP[{\bm z}{\bm x}^T])\geq c$.
\end{assumption}

\begin{assumption}\label{ass_MomentsIID}
    There exist $\eta,C\in (0,\infty)$ such that for all $P\in \mathcal P$,
    \begin{enumerate}
        \item $\EP[\|{\bm z}\|_2^{2+\eta}]\leq C$,\label{ass_ZZIID}
        \item $\EP[\|{\bm x}\|_2^2\|{\bm z}\|_2^2]\leq C$,\label{ass_XZIID}
        \item $\EP[\|{\bm x}\|_2^2]\leq C$.\label{ass_XIID}
    \end{enumerate}
\end{assumption}
Note that assertions \ref{ass_ZZIID} and \ref{ass_XZIID} of Assumption \ref{ass_MomentsIID} imply that $\|\EP[{\bm z}{\bm z}^T]\|_{op}$ and $\|\EP[{\bm z}{\bm x}^T]\|_{op}$ are uniformly bounded in $P\in \mathcal P$.

To analyze the test statistic, we introduce population analogues of the two-stage least squares estimator \eqref{eq_Sam2SLSGen} and residual \eqref{eq_Sam2SLSResGen},
\begin{equation}\label{eq_Pop2SLS}
    {\bm{\beta}}^* \coloneqq {\bm M}\EP[{\bm z}y], \quad \epsilon \coloneqq y - {\bm x}^T{\bm{\beta}}^*
\end{equation}
with
\begin{equation}\label{eq_DefM}
    {\bm M} \coloneqq  \left(\EP[{\bm x} {\bm z}^T]\EP[{\bm z} {\bm z}^T]^{-1} \EP[{\bm z} {\bm x}^T]\right)^{-1}\EP[{\bm x} {\bm z}^T]\EP[{\bm z} {\bm z}^T]^{-1}.
\end{equation}
Moreover, define $\epsilon_i \coloneqq y_i-{\bm x}_i^T{\bm{\beta}}^*$, $i =1,\ldots, n$ accordingly.
By the law of large numbers, ${\bm{\beta}}^*$ is the almost-sure limit of $\hat{\bm{\beta}}$.
Note that ${\bm{\beta}}^*$ and $\epsilon$ are well-defined regardless of whether
the null hypothesis \eqref{eq_H0Effective} holds or not. Moreover, under $H_0$, $\epsilon$ is equal to the structural error: if $H_0$ holds, there exists ${\bm{\beta}_{H_0}}\in\mathbb R^p$ with $\E_P[y-{\bm x}^T{\bm{\beta}_{H_0}}|{\bm z}]=0$, i.e., we can write $y = \bm x^T\bm \beta_{H_0} + u$ with $\E_P[u|\bm z] = 0$. Then a straightforward calculation shows that ${\bm{\beta}}^* = {\bm{\beta}_{H_0}}$ and $\epsilon = u$. It follows that we can reformulate the null hypothesis \eqref{eq_H0Effective} as $H_0: \EP[\epsilon|{\bm z}]=0$. 
We note that ${\bm{\beta}}^*$, ${\bm M}$, and $\epsilon$ depend on $P$ even if this is not reflected in the notation.
We need the following assumption for $\epsilon$.
\begin{assumption}\label{ass_Epsilon}
    There exist $\eta, C\in (0, \infty)$ such that for all $P\in \mathcal P$ it holds that $\EP\left[|\epsilon|^{2+\eta}\right]\leq C$ and $\EP\left[\|{\bm z}\|_2^{2+\eta}|\epsilon|^{2+\eta}\right]\leq C$. 
\end{assumption}
Assumption \ref{ass_Epsilon} is needed to obtain a central limit theorem uniformly over a class of functions $w$ and over $\mathcal P$.
We define the class of (measurable) functions of $\bm z$ with absolute value bounded by $1$,
$$\mathcal W = \left\{w:\mathbb R^d\to [-1, 1]\right\}.$$
Moreover, it will be useful to define for $w\in \mathcal W$
\begin{equation}
    {\bm a}_w^T  \coloneqq -\EP[w({\bm z}){\bm x}^T]{\bm M}.\label{eq_DefAw}
\end{equation}
Finally, we define for $w\in \mathcal W$ 
\begin{align}
    \rho_w &\coloneqq \EP[w({\bm z})\epsilon]\label{eq_DefRhoW},\\
    \sigma_w^2 &\coloneqq \Var_P\left((w({\bm z}) + {\bm a}_w^T{\bm z})\epsilon\right),\label{eq_DefSigmaW2}\\
    {\bm{\tau}} &\coloneqq \EP[{\bm z}\epsilon].\label{eq_DefTau}
\end{align}
Note that ${\bm a}_w$, $\rho_w$, $\sigma_w^2$ and ${\bm{\tau}}$ depend on $P\in \mathcal P$.
Moreover, under the null hypothesis $\E_P[\epsilon|{\bm z}]=0$, both $\rho_w=0$ and ${\bm{\tau}}=0$. One may also check that $\bm \tau = 0$ automatically holds in the just-identified setting when $p = d$.

To obtain uniform results, we need a lower bound on $\sigma_w^2$. For $\zeta>0$, define the set
\begin{equation}
    \mathcal V_P(\zeta) = \left\{w\in \mathcal W \mid \sigma_w^2\geq \zeta\right\}.\label{eq_DefVZeta}
\end{equation}
The next theorem provides a uniform central limit theorem for the standardized statistic, which is the key technical result underlying the validity and power of the proposed test.
\begin{theorem}\label{thm_AsNormIID}
    Let $N(w)$ and $\hat\sigma_w^2$ be defined as in \eqref{eq_DefNStat} and \eqref{eq_DefHatSigmaW}.
    Let $\zeta>0$ be arbitrary but fixed and assume that Assumptions \ref{ass_SigmaMinIID}, \ref{ass_MomentsIID} and \ref{ass_Epsilon} hold. Then,
    \begin{align}
        \lim_{n\to\infty}\sup_{P\in\mathcal P}\sup_{w\in \VP(\zeta)}\sup_{t\in \mathbb R}\left|\Prob_P\left(\frac{N(w) - \sqrt{n_0}\left(\rho_w+(\hat {\bm a}_w- {\bm a}_w)^T{\bm{\tau}}\right)}{\hat\sigma_w}\leq t\right)-\Phi(t)\right| &= 0,\label{eq_GenTestNorm}\\
        \forall\delta>0:\lim_{n\to\infty}\sup_{P\in \mathcal P}\sup_{w\in \mathcal W}\Prob_P\left(|\sigma_w^2-\hat\sigma_w^2|>\delta\right)&=0.\label{eq_GenSigmaWCons}
    \end{align}
\end{theorem}
Under $H_0$, $\rho_w = 0$ and $\tau=0$, so the centering in \eqref{eq_GenTestNorm} vanishes.
Theorem \ref{thm_AsNormIID} follows from the more general Proposition B.1 in Supplementary Appendix B and its proof.
In the following sections, we will examine the consequences of Theorem \ref{thm_AsNormIID} for our test (Procedure \ref{proc_Test}) under the null and alternative hypotheses.

\subsection{Consequences under the Null Hypothesis $\EP[\epsilon|{z}]=0$}
We first show that the proposed procedure yields an asymptotically valid specification test under the null hypothesis (when testing on significance level $\alpha < 1/2$ -- a restriction due to the one-sidedness of the test).
\begin{theorem}\label{thm_PValNull}
    Consider a class of distributions $\mathcal P$ such that for all $P\in \mathcal P$ it holds that $\EP[\epsilon|{\bm z}]=0$. Consider the output $p_{val}(\hat w)$ of Procedure \ref{proc_Test} with fixed $\gamma >0$. If Assumptions \ref{ass_SigmaMinIID}, \ref{ass_MomentsIID} and \ref{ass_Epsilon} hold for $\mathcal P$, then for all $\alpha_0\in (0, 0.5)$,
    $$\limsup_{n\to\infty} \sup_{P\in \mathcal P}\sup_{\alpha\in (0,\alpha_0)}\left\{\Prob_P\left(p_{val}(\hat w)\leq \alpha\right)-\alpha\right\} \leq 0.$$
\end{theorem}
The proof of Theorem \ref{thm_PValNull} can be found in Supplementary Appendix C.
The type I error control in Theorem \ref{thm_PValNull} does not require consistency of the machine learning method used in the construction of the test. This remark also applies to the extension to weak instruments in Section \ref{sec_ExtensionWeak} below.

\subsection{Consequences under the Alternative Hypothesis $\EP[\epsilon|{z}]\neq 0$}\label{sec_ConsequencesHA}
We next study power under alternatives. Because the weight function is learned via machine learning on an auxiliary sample, power analysis requires high-level conditions. While, again, we do not require consistency of the machine learning step, we need to assume that at least some predictive signal is captured by the learner.
We consider a sequence of distributions $(P_n)_{n\in\mathbb N}$ from the alternative hypothesis.
\begin{assumption}\label{ass_Power}
    There exist a sequence $(\xi_n)_{n\in\mathbb N}$ with $\sqrt{n_0}\xi_n\to\infty$ and $\kappa>0$ such that the sequence of random functions $(\hat w_n)_{n\in \mathbb N}$ (output of Procedure \ref{proc_EstW}) satisfies, as $n \to \infty$,
    \begin{align}
        \Prob_{P_n}\left(\hat w_n\in \mathcal V_{P_n}( \kappa)\right)&\to 1, \label{eq_CondPower1}\\
        \Prob_{P_n}(\rho_{\hat w_n} + (\hat {\bm a}_{\hat w_n}- {\bm a}_{\hat w_n})^T{\bm{\tau}}\geq\xi_n)&\to 1.\label{eq_CondPower2}
    \end{align}
\end{assumption}

Condition \eqref{eq_CondPower1} states that $\hat w_n$ is in some sense ``non-degenerate'' with high probability.
For \eqref{eq_CondPower2}, we will prove that $\hat {\bm a}_w - {\bm a}_w\to 0$ uniformly in $w\in \mathcal W$ and $P\in \mathcal P$ (see Lemma C.1 in Supplementary Appendix C). In general, for a fixed alternative, one may expect that $\rho_{\hat w_n}$ is lower-bounded by a constant, since $\hat w_n$ is constructed to be correlated with $\epsilon$ and with large $n$, the quality of $\hat w_n$ should get better. Then, \eqref{eq_CondPower2} is satisfied.
Under these assumptions, we have the following proposition.
\begin{proposition}\label{pro_Power}
    Let $\alpha\in (0,1)$ be fixed. If Assumptions \ref{ass_SigmaMinIID}, \ref{ass_MomentsIID}, \ref{ass_Epsilon}, and \ref{ass_Power} hold for $\mathcal P = (P_n)_{n\in\mathbb N}$, then
    $\lim_{n\to\infty}\Prob_{P_n}\left(p_{val}(\hat w_n)< \alpha\right)= 1.$
\end{proposition}
Proposition \ref{pro_Power} states that if Assumptions \ref{ass_SigmaMinIID}, \ref{ass_MomentsIID}, \ref{ass_Epsilon}, and \ref{ass_Power} hold, then for fixed significance level, the power of the test converges to $1$. The proof of Proposition \ref{pro_Power} can be found in Supplementary Appendix C, and further justification for Assumption \ref{ass_Power} can be found in Supplementary Appendix A.

\subsection{When is $\EP[\epsilon|{z}]= 0$?}\label{sec_WhenIs}
The null hypothesis of the proposed test is $\E[\epsilon|\bm z] = 0$, where $\epsilon$ is the population two-stage least squares residual. To understand which misspecifications of the linear IV model our test cannot detect, we need to characterize situations where the IV model is misspecified but $\EP[\epsilon|{\bm z}]=0$.
For this, we consider a fairly general model with additive error,
\begin{equation}\label{eq_GeneralModel}
    y = f({\bm x}) + \eta.
\end{equation}
Two important examples of model misspecification fall into this setting.
\begin{enumerate}
    \item \textit{Linear model with invalid IV:} There exists ${\bm{\beta}}_0\in \mathbb R^p$ such that $f({\bm x})={\bm x}^T{\bm{\beta}}_0$ but $\E[\eta|{\bm z}]\neq 0$.
    \item \textit{Nonlinear model with valid IV:} The function $f$ is nonlinear but $\E[\eta|{\bm z}]=0$.
\end{enumerate}
From the definitions \eqref{eq_Pop2SLS} and \eqref{eq_DefM}, it follows that
$$\epsilon = y-{\bm x}^T{\bm{\beta}}^* = y-{\bm x}^T{\bm M}\EP[{\bm z}y] = f({\bm x}) + \eta - {\bm x}^T{\bm M}\EP[{\bm z}(f({\bm x})+\eta)]$$
and hence,
$$\EP[\epsilon|{\bm z}]=\EP[f({\bm x})+\eta|{\bm z}]-\EP[{\bm x}|{\bm z}]^T{\bm M}\EP[{\bm z}\EP[f({\bm x})+\eta|{\bm z}]].$$
This leads to a straightforward characterization of when $\EP[\epsilon|{\bm z}]=0.$
\begin{lemma}\label{lem_Power}
    For the model \eqref{eq_GeneralModel}, we have $\EP[\epsilon|{\bm z}]=0$ if and only if there exists ${\bm{\lambda}}\in \mathbb R^p$ such that $\EP[f({\bm x})+\eta|{\bm z}] = \EP[{\bm x}|{\bm z}]^T{\bm{\lambda}}$.
\end{lemma}
Lemma \ref{lem_Power} is proven in Supplementary Appendix C.

Hence, $\E[\epsilon|{\bm z}]= 0$ holds precisely when the conditional expectation of $f({\bm x}) + \eta$ given ${\bm z}$ is a linear function of the conditional expectation of ${\bm x}$ given ${\bm z}$. To make matters more concrete, we consider the previously mentioned special cases of \eqref{eq_GeneralModel}.
\begin{enumerate}
    \item \textit{Linear model with invalid IV:} In this case, the requirement for $\EP[\epsilon|{\bm z}]=0$ is that there exists ${\bm{\lambda}}\in  \mathbb R^p$ such that
    $\EP[{\bm x}|{\bm z}]^T{\bm{\beta}}_0 +\EP[\eta|{\bm z}]=\EP[{\bm x}^T|{\bm z}]{\bm{\lambda}},$
    i.e., there exists ${\bm{\lambda}}'\in \mathbb R^p$ such that
    $\EP[\eta|{\bm z}]=\EP[{\bm x}|{\bm z}]^T{\bm{\lambda}}'.$
    Intuitively, this means that the violation of the instrument validity should not lie in the span of the conditional expectation of the explanatory variables given the instruments. In particular, if $\EP[{\bm x}|{\bm z}]$ is linear in ${\bm z}$, i.e., there exists $\bm\Pi\in \mathbb R^{p\times d}$ such that $\EP[{\bm x}|{\bm z}] = \bm{\Pi} {\bm z}$, then a linear violation $\E[\eta|{\bm z}]={\bm z}^T\bm \mu$ of instrument validity can only be detected if $\bm\mu$ is not in the span of $\bm{\Pi}^T$. Hence, in this case, our test can only detect ``linear violations'' in the overidentified case $d>p$, as it is the case with the standard overidentifying restriction test \citep{SarganEstimationEconomicRelationships}.
    \item \textit{Nonlinear model with valid IV:} In this case, the requirement for $\EP[\epsilon|{\bm z}]=0$ is that there exists ${\bm{\lambda}} \in \mathbb R^p$ such that $\E[f({\bm x})|{\bm z}] = \E[{\bm x}|{\bm z}]^T{\bm{\lambda}}$. If we assume that ${\bm x}$ and ${\bm z}$ satisfy the completeness condition -- a common assumption in the nonlinear IV literature \citep{NeweyIVEstimationNonparametric} -- then $\E[f({\bm x})-{\bm x}^T{\bm{\lambda}}|{\bm z}]=0$ implies that $f({\bm x})-{\bm x}^ T{\bm{\lambda}} = 0.$
    Hence, under completeness, our test can detect arbitrary nonlinear misspecifications.
\end{enumerate}

\section{Extension to Weak and Many Instruments}\label{sec_ExtensionWeak}
The testing procedure in Sections \ref{sec_MotivationMethod} and \ref{sec_TheoryIID} is based on two-stage least squares residuals and therefore relies on strong identification under which two-stage least squares is well-behaved. In settings with weak or many instruments, however, two-stage least squares may be unreliable, which can invalidate the inference. To address this limitation, we introduce an extension in the spirit of the Anderson--Rubin test \citep{AndersonRubinTest} that remains valid without strong identification and allows joint specification testing and parameter inference without strong prior causal model assumptions (see below).

The underlying idea is the following: for each candidate $\bm\beta_0\in \mathbb R^p$, we consider the null hypothesis
$H_0(\bm\beta_0): \EP[y- \bm x^T\bm \beta_0|\bm z] = 0.$
Fixing $\bm\beta_0$ removes the need for consistent estimation of $\bm\beta$ and therefore does not rely on strong identification for validity. Testing $H_0(\bm\beta_0)$ amounts to testing whether the residual $r(\bm\beta_0)=y-\bm x^T\bm\beta_0$ has conditional mean zero given the instruments. We can apply the residual-prediction principle to $r(\bm\beta_0)$ and obtain a p-value $p_{val}(\bm\beta_0)$. By inverting the associated test, we can obtain a confidence set for $\bm \beta_0$ with all the parameter values that are compatible with well-specification at a given significance level.

This approach is of limited 
use when the dimension $p$ of $\bm x$ is larger than $1$ or $2$, since the confidence set for $\bm \beta_0$ becomes difficult to visualize and compute. However, in many applications, only a small number of regressors are endogenous, whereas the number of exogenous control variables may be substantially larger. In the following, we therefore focus on inference for the low-dimensional subvector of endogenous coefficients while partialling out the exogenous control variables.

\subsection{Method}
From now on, we distinguish notationally between endogenous regressors and exogenous control variables. Let $\bm c\in \mathbb R^q$ denote the additional exogenous control variables (potentially including an intercept). A linear IV model with controls corresponds to $y = \bm x^T\bm \beta + \bm c^T\bm \theta + u$ with $\EP[u|\bm z, \bm c] = 0$. Note that the setup from Sections \ref{sec_MotivationMethod} and \ref{sec_TheoryIID} can be recovered by replacing $\bm x \gets (\bm x, \bm c)$ and $\bm z \gets (\bm z, \bm c)$. We then consider the null hypothesis
\begin{equation}\label{eq_DefH0WeakPartial}
    H_0(\bm \beta_0): \exists \bm \theta \in \mathbb R^q \text{ s.t. } \EP[y - \bm x^T\bm \beta_0 - \bm c^T\bm \theta|\bm z, \bm c] = 0,
\end{equation}
which is indexed only by the endogenous coefficient $\bm \beta_0$.

Write $r(\bm\beta_0)=y-\bm x^T\bm\beta_0$ (and $r_i(\bm\beta_0)=y_i-\bm x_i^T\bm\beta_0$). Reusing the notation from Section \ref{sec_MotivationMethod}, for any $w:\mathbb R^{d}\times\mathbb R^q\to\mathbb R$, define the (sample) least squares coefficients from regressing $r(\bm \beta_0)$ and $w(\bm z, \bm c)$ on $\bm c$ on the main sample $\mathcal D$,
\begin{equation}\label{eq_DefLambdaHat}
    \hat{\bm\lambda}_r(\bm\beta_0) \coloneqq \hatED[\bm c\bm c^T]^{-1}\hatED[\bm c\,r(\bm\beta_0)],
    \quad
    \hat{\bm\lambda}_w \coloneqq \hatED[\bm c\bm c^T]^{-1}\hatED[\bm c\,w(\bm z,\bm c)].
\end{equation}
and the partialled-out quantities
\begin{equation}\label{eq_DefTildeRW}
    \tilde r_i(\bm\beta_0)\coloneqq r_i(\bm\beta_0)-\bm c_i^T\hat{\bm\lambda}_r(\bm\beta_0),
    \quad
    \tilde w_i \coloneqq w(\bm z_i,\bm c_i)-\bm c_i^T\hat{\bm\lambda}_w.
\end{equation}
The partialled-out test statistic and variance estimator are then\footnote{A variance estimator assuming homoskedasticity in the spirit of \eqref{eq_DefHatSigmaWHomoscedastic} can be found in Supplementary Appendix A.}
\begin{align}
    N(w,\bm\beta_0) &\coloneqq \frac{1}{\sqrt{n_0}}\sum_{i\in\mathcal D}\tilde w_i\,\tilde r_i(\bm\beta_0),\label{eq_WeakNPartial}\\
    \hat\sigma_w^2(\bm\beta_0) &\coloneqq 
    \frac{1}{n_0}\sum_{i\in\mathcal D}\tilde w_i^{\,2}\,\tilde r_i(\bm\beta_0)^2
    -\left(\frac{1}{n_0}\sum_{i\in\mathcal D}\tilde w_i\,\tilde r_i(\bm\beta_0)\right)^2.\label{eq_WeakVarPartial}
\end{align}
As in Procedure \ref{proc_Test}, we need to lower-bound the variance by some fixed $\gamma >0$ to define the p-value
\begin{equation}\label{eq_PValueWeak}
    p_{val}(\bm\beta_0) = p_{val}(\bm \beta_0, \hat w)\coloneqq 1-\Phi\!\left(\frac{N(\hat w,\bm\beta_0)}{\max\!\left(\hat\sigma_{\hat w}(\bm\beta_0),\sqrt\gamma\right)}\right),
\end{equation}
where $\hat w$ is learned on the auxiliary sample $\mathcal D_A$ and depends on $\bm \beta_0$ (see below).

Inverting the associated test, we obtain the confidence set
\begin{equation}\label{eq_DefCAlpha}
    \mathcal C_\alpha=\left\{\bm \beta_0\in \mathbb R^p\mid p_{val}(\bm\beta_0) \geq\alpha\right\}.
\end{equation}
In addition to being robust to weak identification, the set $\mathcal C_\alpha$ enables joint specification testing and parameter inference: if the model is well-specified, the set $\mathcal C_\alpha$ asymptotically has coverage at least $1-\alpha$. If $\mathcal C_\alpha$ is empty, well-specification can be rejected. In particular, inference using the confidence set $\mathcal C_\alpha$ avoids the problem of selective inference
that arises when one considers parameter inference only if well-specification was not rejected in an initial step. Alternatively, overall well-specification can be assessed by considering $\sup_{\bm \beta_0\in \mathbb R^p} p_{val}(\bm \beta_0)$, which is smaller than $\alpha$ if and only if $\mathcal C_\alpha$ is empty.

It remains to discuss how the function $\hat w$ should be constructed in the presence of exogenous controls.
A naive approach would regress the residuals $y_i - \bm x_i^T\bm\beta_0$ on $(\bm z_i,\bm c_i)$ using the auxiliary sample $\mathcal D_A$. However, this may lead the learned weight to capture substantial linear signal in $\bm c$, which is subsequently removed in the partialling out step of the variance estimator, potentially resulting in a very small variance estimate and undesirable truncation with $\gamma$.
To avoid this issue, we instead first partial out the linear effect of $\bm c$ from the residuals and then apply residual prediction to the resulting quantity. This leads to the following procedure.
\begin{procedure}[Estimation of weight function, weak-instrument-robust, partialling out]\label{proc_EstWPartial}
Input: auxiliary sample $({\bm x}_i, y_i, {\bm z}_i, {\bm c}_i)_{i\in \mathcal D_A}$, candidate $\bm \beta_0$.
    \begin{enumerate}
    \item Let $r_i(\bm \beta_0) = y_i - \bm x_i^T \bm\beta_0$, $i\in \mathcal D_A$ and let $\hat {\bm \lambda}_r^{A}(\bm \beta_0)$ be as in \eqref{eq_DefLambdaHat} but with the data $\mathcal D_A$ instead of $\mathcal D$.
    \item With a user-chosen machine learning method, regress $(r_i(\bm \beta_0) - \bm c_i^T\hat {\bm \lambda}_r^{A}(\bm\beta_0))_{i\in \mathcal D_A}$ on $({\bm z}_i, \bm c_i)_{i\in \mathcal D_A}$. Call the resulting function $\hat w_0(\cdot, \cdot)$.
    \item For some $K>0$, define $\hat w(\cdot, \cdot) = \frac{1}{K}\sign(\hat w_0(\cdot, \cdot))\min(|\hat w_0(\cdot, \cdot)|, K)$. Return $\hat w(\cdot, \cdot)$.
\end{enumerate}
\end{procedure}
Note that in step 2, we regress the partialled-out residuals both on $\bm z$ and on $\bm c$, so the learner can capture nonlinearities and interactions involving controls, but linear control effects are removed beforehand.
\begin{remark}
    When evaluating $H_0(\bm \beta_0)$ over a grid of values for $\bm \beta_0$ (e.g., in order to determine confidence sets), obtaining a new $\hat w$ for every candidate $\bm \beta_0$ might be computationally expensive. As a computational shortcut, we recommend tuning hyperparameters only once (for example, at the two-stage least squares residuals) and refit the machine learner at the other values of $\bm \beta$ using those hyperparameters. This leads to a significant speedup.
\end{remark}

\subsection{Theoretical Results}
We now provide validity results for the procedure introduced in the previous section under i.i.d.\ sampling. Extensions to non-i.i.d. data are provided in the Supplementary Appendix.

Let $\mathcal P$ be a class of distributions for $(\bm x,y,\bm z,\bm c)\in\mathbb R^{p+1+d+q}$, and assume that $(\bm x_i,y_i,\bm z_i,\bm c_i)_{i=1,\ldots,n}$ are i.i.d.\ copies of $(\bm x,y,\bm z,\bm c)$ under $P\in\mathcal P$. We consider distributions under the null hypothesis:
\begin{assumption}\label{ass_H0Weak}
For all $P\in\mathcal P$, there exist $\bm\beta=\bm\beta(P)\in\mathbb R^p$ and $\bm\theta=\bm\theta(P)\in\mathbb R^q$ such that
$$\EP\left[y-\bm x^T\bm\beta(P)-\bm c^T\bm\theta(P)| \bm z,\bm c\right]=0.$$
\end{assumption}
Assumption \ref{ass_H0Weak} formalizes that we consider distributions $P$, where the null hypothesis \eqref{eq_DefH0WeakPartial} holds at $\bm \beta_0 = \bm \beta(P)$.
Define $\nu\coloneqq y-\bm x^T\bm\beta(P)-\bm c^T\bm\theta(P)$. We need the following moment conditions.
\begin{assumption}\label{ass_MomentsH0Weak}
There exist $\eta,c,C\in(0,\infty)$ such that for all $P\in\mathcal P$ it holds that $\sigma_{\min}(\EP[\bm c\bm c^T])\geq c$, $\EP[|\nu|^{2+\eta}]\leq C$, $\EP[\|\bm c\|_2^{2+\eta}|\nu|^{2+\eta}]\leq C$, and $\EP[\|\bm c\|_2^{4 + \eta}]\leq C$.
\end{assumption}
Note that there are no additional moment conditions on the instruments $\bm z$ and the endogenous variables $\bm x$, in particular, no assumption on the IV strength.
For $P\in\mathcal P$ and $\zeta>0$, define
\begin{align}
\bar\sigma_w^2 &\coloneqq 
\EP\!\left[\left(w(\bm z,\bm c)-\bm c^T\EP[\bm c\bm c^T]^{-1}\EP[\bm c\,w(\bm z,\bm c)]\right)^2\nu^2\right],\label{eq_DefBarSigmaW}\\
\mathcal U_P(\zeta) &\coloneqq 
\left\{w:\mathbb R^{d}\times\mathbb R^{q}\to[-1,1]\, \vert\, \bar\sigma_w^2\ge \zeta\right\}.\label{eq_DefUZeta}
\end{align}
Then, we have the following analogue of Theorem \ref{thm_AsNormIID}, which follows from Proposition B.2 in Supplementary Appendix B and its proof.
\begin{theorem}\label{thm_WeakAsNormIID}
    Assume that Assumptions \ref{ass_H0Weak} and \ref{ass_MomentsH0Weak} hold, let $N(w, \bm \beta(P))$ and $\hat \sigma_w^2(\bm \beta(P))$ be according to \eqref{eq_WeakNPartial} and \eqref{eq_WeakVarPartial} (plugging in $\bm \beta(P)$ for $\bm \beta_0$) and let $\zeta>0$ be arbitrary but fixed. Then,
    \begin{align}
        \lim_{n\to\infty}\sup_{P\in \mathcal P}\sup_{w\in \mathcal U_P(\zeta)}\sup_{t\in \mathbb R}\left|\Prob_P\left(\frac{N(w, \bm \beta(P))}{\hat \sigma_w(\bm\beta(P))}\leq t\right)-\Phi(t)\right| &= 0, \label{eq_AsNormWeak}\\
        \forall\delta>0:\lim_{n\to\infty}\sup_{P\in \mathcal P}\sup_{|w|\leq 1}\Prob_P\left(\left|\bar\sigma_w^2-\hat\sigma_w^2(\bm \beta(P))\right| > \delta\right) & = 0.
    \end{align}
\end{theorem}
Similarly to before, the following result follows (proven in Supplementary Appendix C), stating that our proposed approach leads to asymptotically valid p-values.
\begin{theorem}\label{thm_WeakPvalNull}
Assume that the class of distributions $\mathcal P$ satisfies Assumptions \ref{ass_H0Weak} and \ref{ass_MomentsH0Weak}. Let $p_{val}(\cdot)$ be defined as in \eqref{eq_PValueWeak} with $N(\cdot, \cdot)$ and $\hat\sigma_\cdot(\cdot)$ defined in \eqref{eq_WeakNPartial} and \eqref{eq_WeakVarPartial}, $\hat w$ the output of Procedure \ref{proc_EstWPartial} and let $\gamma > 0$ be fixed. Then, it holds that for all $\alpha_0 \in (0, 0.5)$,
$$\limsup_{n\to\infty}\sup_{P\in \mathcal P}\sup_{\alpha \in (0, \alpha_0)}\left\{\Prob_P\left(p_{val}\left(\bm\beta(P)\right)\leq \alpha\right)-\alpha\right\}\leq 0.$$
\end{theorem}
\begin{remark}
    Since $\Prob_P\left(\bm \beta(P)\in \mathcal C_\alpha\right) = \Prob_P\left(p_{val}(\bm \beta(P)) \geq \alpha\right)$ it follows from Theorem \ref{thm_WeakPvalNull} that the inverted set $\mathcal C_\alpha$ has asymptotic coverage at least $1-\alpha$ for $\bm\beta(P)$, uniformly over $\mathcal P$ and $\alpha\in(0,\alpha_0)$. If the conditional moment restriction \eqref{eq_DefH0WeakPartial} fails for every $\bm\beta_0\in\mathbb R^p$, so that no parameter value is compatible with well-specification, the confidence set $\mathcal C_\alpha$ is expected to become empty as the sample size increases. A formal proof of such emptiness would require additional power conditions on the learned weight function, analogous to those discussed in Sections \ref{sec_ConsequencesHA} and \ref{sec_WhenIs}, which we do not pursue here.
\end{remark}

\section{Experiments}\label{sec_Experiments}
We apply our method to synthetic data and a real-world dataset. The method is implemented in the R package \texttt{RPIV} available on CRAN and the R code to replicate all the experiments can be found on GitHub (\url{https://github.com/cyrillsch/RPIV_Application}).
\subsection{Simulations}\label{sec_Simulations}
A detailed simulation study can be found in Supplementary Appendix D. In the following, we present a selection of the simulation results.
\subsubsection{Simulation design}
We evaluate the finite-sample performance of our procedure under a flexible data-generating process, which we sketch next. The precise data-generating process is described in Supplementary Appendix D.

Let $\bm z \in \mathbb{R}^{n_{IV}}$ denote the instruments and $\bm c \in \mathbb{R}^{n_C}$ additional exogenous controls, both generated as standard normal variables which are partially correlated with each other. The (univariate) endogenous regressor is generated as
$x = \pi \tanh\!\left(\frac{1}{\sqrt{n_{IV}}}\sum_{k=1}^{n_{IV}} z_k \right)
      + 0.3 c_1
      + \delta,$
where $\pi \geq 0$ controls instrument strength and $\delta$ is correlated with the structural error through a latent confounder, so that instrumental variables are required for identification. The outcome is given by
$y = 2 - x + 0.5 \sum_{j=1}^{n_C} c_j + \epsilon,$
where the structural error $\epsilon$ may be homoskedastic or heteroskedastic. In the heteroskedastic setting, we scale $\epsilon$ by $|z_1|.$
Under the alternative, we add a violation term $s_{viol} v$ to $y$, where $s_{viol} \geq 0$ controls the violation strength and $v$ corresponds to one of four deviations: nonlinear dependence on the instrument ($ z_1^2$ or $\text{sign}(z_1)$) or nonlinear misspecification of the structural equation ($( -x + 0.5\sum_{j = 1}^{n_C}  c_j)^2$  or $\sign ( -x + 0.5\sum_{j=1}^{n_C}  c_j)$), which we will refer to as \textit{``z squared''}, \textit{``sign(z)''},  \textit{``misspec. squared''} and \textit{``misspec. sign''} in the following.

For each configuration, we compute rejection frequencies over $1000$ Monte Carlo repetitions at level $\alpha = 0.05$.
We compare the following procedures (implementation details are given in Supplementary Appendix D):
\begin{description}
    \item[RP hom./RP het.:] Our residual prediction test from Sections \ref{sec_MotivationMethod} and \ref{sec_TheoryIID} with homoskedastic and heteroskedastic variance estimators (controls $\bm c$ and the intercept are added to both $\bm x$ and $\bm z$).
    \item[weak RP hom./weak RP het.:] The weak-identification-robust extension from Section \ref{sec_ExtensionWeak} (also with homoskedastic and heteroskedastic variance estimators), where p-values are obtained by taking the maximum $p_{val}(\beta_0)$ over a grid covering $[-10, 10]$ of candidate $\beta_0$ (note that $\beta_0$ is univariate and its true value under $H_0$ is equal to -1).
    \item[overid. J:] The classical overidentifying restrictions test if the number of instruments is larger than 1. For $n_{IV}=1$, we add $z^2$ as additional instrument, which corresponds to the approach by \citet{DieterleASimpleDiagnostic}.
    \item[smooth asymp./smooth boot.:] The smooth test of \cite{DelgadoConsistentTestsOfConditionalMomentRestrictions} applied to the linear IV model with either asymptotic or bootstrapped p-values (see Supplementary Appendix D for details).
    \item[ICM:] The integrated conditional moment (ICM) test of \cite{AntoineIdentificationRobustNonparametricInference}, where p-values are obtained by taking maximum p-values over a grid of candidate $\beta_0$ values (see Supplementary Appendix D for details).
\end{description}
For the residual prediction based procedures, we use a random forest tuned using out-of-bag error and use $n_A = \min(n/2, \exp(1)\cdot n/\log n)$ of the $n$ samples as the training set and the rest as the test set.\footnote{This choice is somewhat arbitrary, but from a theoretical point of view, it is appealing to make the size of the auxiliary sample of slightly smaller order than $n$ in order not to lose power asymptotically.}

\subsection{Size under the Null Hypothesis}
Figure \ref{fig_H0Main} reports rejection frequencies under $H_0$ (i.e., $s_{viol} = 0$) as the sample size increases. 
We fix $n_C = 2$ and $\pi = 1$ and consider $n_{IV} \in \{1,2\}$ under both homoskedastic and heteroskedastic errors.
\begin{figure}[t]
\centering
\includegraphics[width=0.9\textwidth]{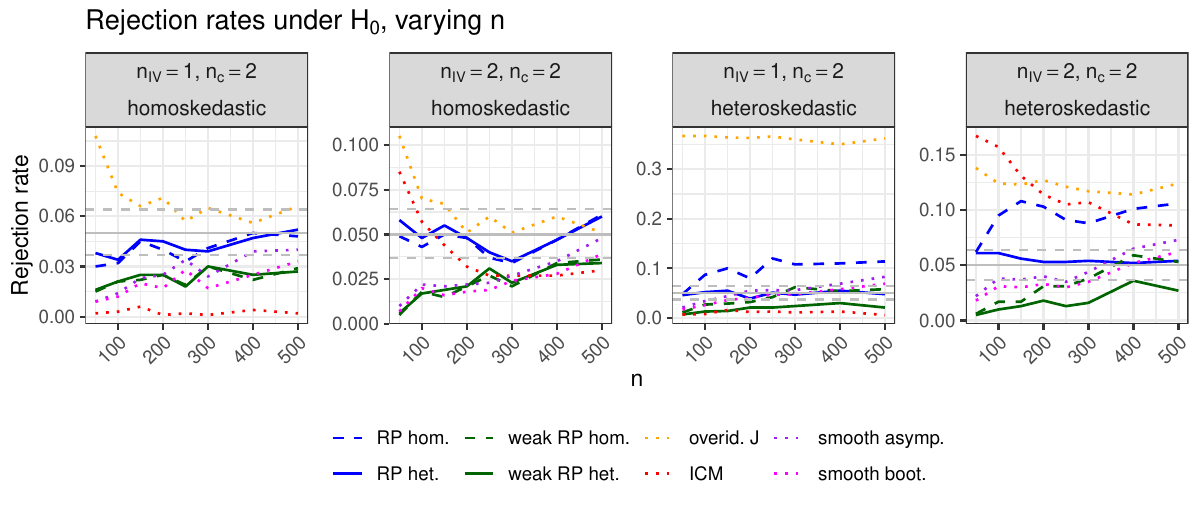}
\caption{Simulation results under $H_0$ with varying $n$. The gray solid line marks $\alpha = 0.05$, the gray dashed lines show pointwise 95\% Monte Carlo bounds for rejection rates based on 1000 replications. Dashed curves coincide with same-colored solid curves where not visible.} 
\label{fig_H0Main}
\end{figure}
Across all configurations, the rejection rates of the proposed residual prediction tests are close to (RP) or below (weak RP) the nominal level of $\alpha = 0.05$.
The heteroskedasticity-robust variants remain correctly sized under heteroskedastic errors, whereas procedures relying on homoskedasticity exhibit noticeable size distortions. Moreover, \textit{overid. J} rejects too often for both heteroskedastic settings and \textit{ICM} rejects too often for one of the heteroskedastic settings, whereas the \textit{smooth} tests seem to remain correctly sized.

In Figure \ref{fig_H0Main_vary_n_iv}, we fix $n = 300$ and $n_C = 2$ and report rejection frequencies under $H_0$ as the number of instruments $n_{IV}$ increases.
\begin{figure}[t]
\centering
\includegraphics[width=0.65\textwidth]{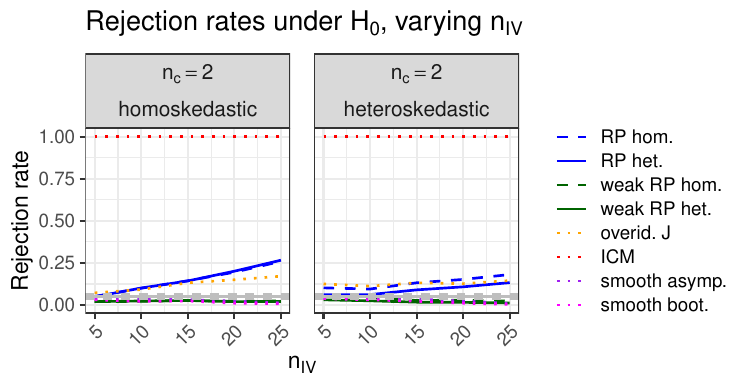}
\caption{Simulation results under $H_0$ with varying $n_{IV}$. The gray solid line marks $\alpha = 0.05$, the gray dashed lines show pointwise 95\% Monte Carlo bounds for rejection rates based on 1000 replications. Dashed curves coincide with same-colored solid curves where not visible.} 
\label{fig_H0Main_vary_n_iv}
\end{figure}
We see that \textit{weak RP hom./het.} control size well, whereas \textit{RP hom./het} reject too often for large $n_{IV}$. Moreover, also \textit{ICM} and \textit{overid. J} suffer from overrejection.

\subsection{Power under the Alternative Hypothesis}
Figures \ref{fig_HAMain} and \ref{fig_HAMain_vary_iv_strength} report rejection rates as a function of the violation strength $s_{viol}$ and the IV strength $\pi$, respectively, in the homoskedastic setting with $n = 300$, $n_C = 2$, and $n_{IV} = 1$. There is no uniformly dominant procedure across all the simulation settings. The residual prediction tests are competitive with the alternatives and yield higher power than the overidentifying J test (augmented with $z^2$) except in the \textit{z squared} setting. All procedures have limited power in the \textit{misspec. sign} setting.
In many settings, \textit{RP hom./het.} exhibit larger power than \textit{weak RP hom./het.}, but for small IV strength $\pi$ and violation \textit{sign(z)}, \textit{weak RP hom./het.} outperforms \textit{RP hom./het} in terms of power (second panel of Figure \ref{fig_HAMain_vary_iv_strength}).
\begin{figure}[t]
\centering
\includegraphics[width=0.9\textwidth]{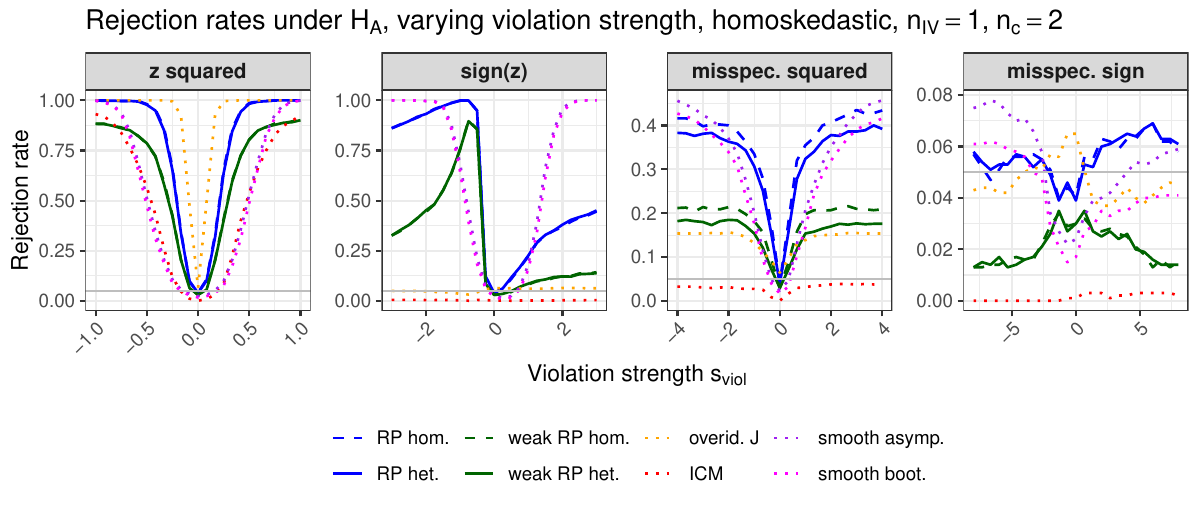}
\caption{Simulation results under $H_A$ with varying $s_{viol}$. The gray solid line marks $\alpha = 0.05$. Dashed curves coincide with same-colored solid curves where not visible.} 
\label{fig_HAMain}
\end{figure}
\begin{figure}[t]
\centering
\includegraphics[width=0.9\textwidth]{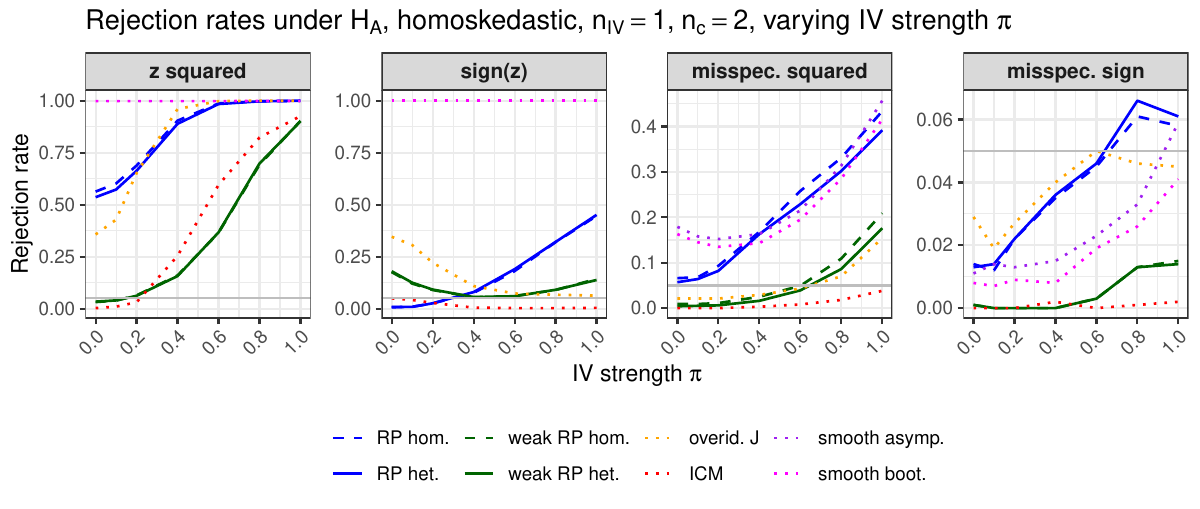}
\caption{Simulation results under $H_A$ with varying $\pi$. The gray solid line marks $\alpha = 0.05$. Dashed curves coincide with same-colored solid curves where not visible.} 
\label{fig_HAMain_vary_iv_strength}
\end{figure}

Additional simulations reported in Supplementary Appendix D investigate further simulation settings.

\subsection{Empirical Illustration -- Returns to Education}
We revisit the study of \citet{CardCollege} on the effect of education on wages. 
Following \citet{WooldridgeIntroductoryEconometrics}, the outcome is log wage, the endogenous regressor is years of education, and the instrument is an indicator for growing up near a four-year college. 
Additional exogenous controls include experience, experience squared, and demographic variables. The dataset is available in the R package \texttt{wooldridge} \citep{WooldridgeRPackage}.

To reduce variability due to sample splitting, the residual prediction procedures are repeated 100 times, and we report the doubled median p-value, which yields a conservative p-value \citep{MeinshausenPValuesForHDRegression}.

Figure \ref{fig_Card} plots p-values as a function of the candidate coefficient $\beta_0$. 
The left panel corresponds to the standard specification including experience squared, while the right panel omits this quadratic term. 
P-values are shown on a log scale. 
The residual prediction tests \textit{RP hom./het.} and the smooth tests appear as horizontal lines, whereas the weak-identification robust versions \textit{weak RP hom./het.} and the ICM procedure vary with $\beta_0$.

For the standard specification, \textit{RP hom./het.} do not reject at level $\alpha=0.05$, while the \textit{weak RP hom./het.} yields p-values smaller than $0.05$ across all candidate values of $\beta_0$. 
The first-stage $F$-statistic equals 13.26, indicating moderate instrument strength, and we have seen in the previous section that \textit{weak RP hom./het.} can sometimes outperform \textit{RP hom./het.} in terms of power, which might explain the differing behavior.
The p-values of the \textit{smooth} tests are close to the 5\% significance level, and the \textit{ICM} test produces a nonempty set of candidate coefficients compatible with well-specification. Note that the classical \textit{overid. J} test is not applicable because only one instrument is available, and the approach by \cite{DieterleASimpleDiagnostic} is not applicable, because $z = z^2$ (instrument is binary).

For the specification omitting experience squared, all procedures except the \textit{ICM} test reject at level $\alpha=0.05$, providing strong evidence against well-specification.
\begin{figure}[t]
\centering
\includegraphics[width=0.6\textwidth]{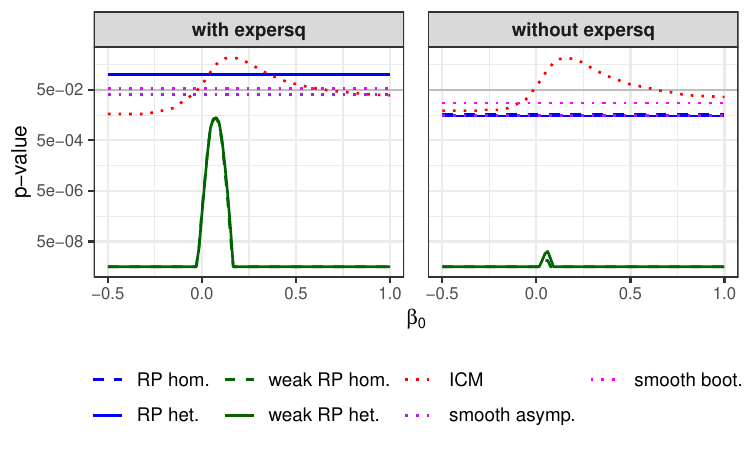}
\caption{P-values for two specifications of the Card dataset as a function of $\beta_0$. Dashed curves coincide with same-colored solid curves where not visible. The gray horizontal line indicates the significance level $\alpha = 0.05$.} 
\label{fig_Card}
\end{figure}

An additional real-world data illustration can be found in Supplementary Appendix E.

\section{Conclusion}
We introduced a new approach for specification testing in linear instrumental
variable (IV) models based on the principle of residual prediction. By leveraging modern machine learning methods to
search for predictive structure, the proposed test is able to detect a wide range of
violations of the IV assumptions.
A key advantage of our approach is that it applies even in the just-identified case, where classical tests such as the Sargan--Hansen J test are not available. We further extended the methodology with an Anderson--Rubin-type construction that remains valid in the presence of weak or many instruments, thereby addressing an important practical limitation of residual-based diagnostics relying on two-stage least squares. The resulting procedures come with asymptotic guarantees for type I error control and consistency under broad classes of alternatives, while allowing practitioners to flexibly incorporate user-chosen machine learning algorithms.
We demonstrated their performance 
through experiments on synthetic data and 
real-world datasets, illustrating their practical relevance.

An obvious limitation of our approach is that it tests only a superset of the well-specification, i.e., the existence of a vector ${\bm{\beta}}$ such that $\E[y-{\bm x}^T{\bm{\beta}}|{\bm z}]=0$. As discussed in Section \ref{sec_WhenIs}, there can be scenarios where the well-specification is violated, but such a vector ${\bm{\beta}}$ still exists. This limitation parallels concerns raised by \citet{ParenteACautionaryNote} about the classical overidentification test.

Extensions of our methodology to dependent data are, in principle, possible. In Supplementary Appendix B, we present high-level conditions that are needed for our test statistic to be asymptotically normal (a central limit theorem, a law of large numbers, and a consistent variance estimator) and use them to construct extensions of our tests to clustered data. Time series data would pose greater challenges, primarily because the tests require the auxiliary sample used to estimate the weight function to be independent of the main sample -- a condition difficult to meet under serial dependence.

Whereas we focused on the special case of the linear instrumental variable model, an extension to arbitrary models defined by conditional moment restrictions is
possible. We sketch this extension in Supplementary Appendix G.

Overall, our methodology provides flexible and theoretically well-founded tools for assessing model specification in linear IV models, with advantages compared to existing approaches due to the generality and power of machine learning. 
They are easy to use, thanks to our publicly available implementation in the R package \texttt{RPIV} and the Python package \texttt{ivmodels}.

\section*{Acknowledgments}
CS received funding from the Swiss National Science Foundation, grant no. 214865. ML was
partially supported by ETH Foundations of Data Science.

\bibliographystyle{abbrvnat}

\bibliography{reference}

\newpage
\begin{center}
    {\LARGE Supplementary Appendix for:\\
    Machine-Learning-Powered Specification Testing in Linear Instrumental Variable Models}
\end{center}

\appendix

\numberwithin{equation}{section}
\numberwithin{figure}{section}

\counterwithin{assumption}{section}
\counterwithin{theorem}{section}
\counterwithin{lemma}{section}
\counterwithin{proposition}{section}
\counterwithin{definition}{section}
\counterwithin{remark}{section}

\section{Additional Material}
\subsection{Further Motivation of Assumption 4 in the Main Text}\label{sec_FurtherPower}
In this section, we further motivate Assumption 4 in the main text. For this, consider a fixed alternative, i.e., $P_n = P$, $n\in \mathbb N$. Then, for eq. (22) in the main text, it is enough to assume that $\rho_{\hat w_n}$ is lower-bounded by a constant, since $\hat {\bm a}_w - {\bm a}_w\to 0$ uniformly in $\mathcal W$ (see Lemma \ref{lem_AwBwCons} below). This means that the function $\hat w_n$ returned by Procedure 1 in the main text should at least capture some non-vanishing signal in $\epsilon$. The following Lemma states that if $\hat w_n$ converges to some $w^*\in \mathcal W$, then $\rho_{\hat w_n}\to\rho_{w^*}$ and $\sigma_{\hat w_n}\to\sigma_{w^*}$.
\begin{lemma}\label{lem_WnWStar}
    Assume that Assumptions 1-3 in the main text hold for $\mathcal P = \{P\}$. Additionally, assume that there exists $C>0$ such that $\EP[\epsilon^2|{\bm z}]\leq C$ a.s., and there exists $w^*\in \mathcal W$ such that $\|\hat w_n-w^*\|_{L_2}=o_P(1)$, where $\|w\|_{L_2}^2 = \EP[w({\bm z})^2]$. Then,
    \begin{align}
        \rho_{\hat w_n} - \rho_{w^*} &= o_P(1),\label{eq_RhoWEst}\\
        \sigma_{\hat w_n}^2 - \sigma_{w^*}^2 &=o_P(1).\label{eq_SigmaWEst}
    \end{align}
\end{lemma}
Lemma \ref{lem_WnWStar} is proven in Section \ref{sec_ProofWnWStar} below. From Lemma \ref{lem_WnWStar}, it follows that if eq. (21) and (22) from the main text hold for $w^*$, and if $\|w^*- \hat w_n\|_{L_2}\to 0$, then Assumption 4 (from the main text) also holds.
Let $w_0(\cdot) = \EP[\epsilon|{\bm z}=\cdot]$. In view of Procedure 1, a candidate for $w^*$ is $w^*_K(\cdot) = \frac{1}{K}\sign( w_0(\cdot))\min(| w_0(\cdot)|, K)$ for some $K>0$.
Note that $\rho_{w_0} = \EP[w_0({\bm z})\epsilon]=\EP\left[\EP[\epsilon|{\bm z}]^2\right]>0$. Moreover, if $w_0({\bm z})$ is not linear in ${\bm z}$, then $\sigma_{w_0}^2 = \Var_P((w_0({\bm z}) + {\bm a}_{w_0}^T {\bm z})\epsilon)\neq 0$ (assuming that $\EP[w_0({\bm z})^2\epsilon^2]$ exists). By dominated convergence, it follows that for $K$ large enough, $\rho_{w^*_K}>0$ and $\sigma_{w^*_K}^2>0$, implying that (21) and (22) in the main text hold for $w^*_K$.

\subsection{Homoskedastic Variance Estimator for the Weak-Instrument-Robust Test}
A variant of $\hat\sigma_w^2(\bm \beta_0)$ given in eq. (28) in the main text under the assumption of homoskedastic errors is given by
\begin{equation}\label{Def_WeakSigmaHatHomo}
    \hat\sigma_{w, hom}^2(\bm \beta_0) = \frac{1}{n_0}\sum_{i\in \mathcal D} \tilde w_i^2\frac{1}{n_0}\sum_{i\in \mathcal D} \tilde r_i(\bm \beta_0)^2.
\end{equation}

\section{General Asymptotic Theory and Clustered Data}\label{sec_GeneralTheory}
We provide high-level conditions under which the residual-prediction statistics are asymptotically normal beyond the i.i.d. setting, and specialize them to clustered dependence. We treat both the baseline test (based on two-stage least squares residuals) and the weak-identification robust extension (based on testing $H_0(\bm\beta_0)$).

\subsection{Baseline Residual Prediction Test: Theory in More General Form}\label{sec_TheoryMoreGeneralForm}
Let $P\in \mathcal P$ be a probability measure that determines the law of the sequence $({\bm x}_i, y_i, {\bm z}_i)_{i\in \mathbb N}$. We do not necessarily assume that the $({\bm x}_i, y_i, {\bm z}_i)_{i\in \mathbb N}$ are identically distributed or independent. 
We consider a subset $\mathcal D\subset\{1,\ldots, n\}$ and assume that if $n\to\infty$, also $n_0 = |\mathcal D|\to \infty$. For a function $f:\mathbb R^{p+1+d}\to\mathbb R^l$, define the notation
\begin{align}
    \EDP[f({\bm x}, y, {\bm z})] &\coloneqq \EP\left[\hatED[f({\bm x}, y, {\bm z})]\right]=\frac{1}{n_0}\sum_{i\in \mathcal D}\EP[f({\bm x}_i, y_i, {\bm z}_i)],\label{eq_DefEDP}
\end{align}
where 
the notation $\hatED$ was defined in eq. (4) in the main text.
That is, $\EDP$ is the expectation under $P$ of the average over $\mathcal D$.

We can now define ${\bm{\beta}}^*$, $\bm M$, $\epsilon$, ${\bm a}_w$, $\rho_w$ and ${\bm{\tau}}$ as in (12), (13), (14), (15), and (17) in the main text by replacing all the expectations $\EP[\cdot]$ with $\EDP[\cdot]$. As an important difference to the i.i.d.\ case, we define
\begin{equation}
    \sigma_w^2\coloneqq\Var_P\left(\frac{1}{\sqrt{n_0}}\sum_{i\in \mathcal D}(w({\bm z}_i) + {\bm a}_w^T{\bm z}_i)\epsilon_i\right),\label{eq_DefSigmaWGen}
\end{equation}
which also depends on $\mathcal D$ (but notationally suppressed here). As before, let $\mathcal W=\{w:\mathbb R^d\to[-1,1]\}$ and define $\VDP(\zeta) = \left\{w\in \mathcal W|\sigma_w^2\geq \zeta\right\}$, which now also depends on $\mathcal D$ (because of the dependence of $\sigma^2_w$ on $\mathcal{D}$).
Let us introduce the following notations for uniform convergence in probability.
\begin{definition}\label{def_StochOrder}
    For a (scalar-, vector- or matrix-valued) random variable $S_\mathcal D = S\left(({\bm x}_i, y_i, {\bm z}_i)_{i\in \mathcal D}\right)$ depending on the sample $\mathcal D$, we say that 
    $$S_\mathcal D = \oP(1)$$
    if
    $$\forall\delta>0: \lim_{n\to\infty}\sup_{P\in \mathcal P}\Prob_P(\|S_\mathcal D\|>\delta)=0.$$
    If the random variable additionally depends on $w\in \mathcal W$, $S_{\mathcal D, w}  =  S\left(({\bm x}_i, y_i, {\bm z}_i)_{i\in \mathcal D}, w\right)$, we say that
    $$S_{\mathcal D, w} = \oPW(1)$$
    if
    $$\forall\delta>0: \lim_{n\to\infty}\sup_{P\in \mathcal P}\sup_{w\in \mathcal W}\Prob_P(\|S_{\mathcal D, w}\|>\delta)=0.$$ 
    We say that 
    $$S_{\mathcal D} = \OP(1)$$
    if
    $$\lim_{K\to\infty}\sup_{P\in \mathcal P}\sup_{n\in \mathbb N}\Prob_P(\|S_{\mathcal D}\|\geq K) = 0.$$
    Here, the norm $\|\cdot\|$ is the absolute value for scalars, the euclidean norm $\|\cdot\|_2$ for vectors and the operator/spectral norm $\|\cdot\|_{op}$ for matrices.
\end{definition}

We are now ready to present the high-level conditions that we need for the asymptotic normality in Theorem 1 in the main text to hold. The first one essentially states that the population two-stage least squares estimator ${\bm{\beta}}^*$ is well-defined uniformly in $n$ and $P \in \mathcal{P}$.
\begin{assumption}\label{ass_MomentsGen}
    There exist $c, C\in (0,\infty)$ such that for all $n\in \mathbb N$ and for all $P\in \mathcal P$,
    \begin{enumerate}
        \item $\sigmamin\left(\EDP[{\bm z}{\bm z}^T]\right)\geq c$,\label{ass_SigmaminZZ}
        \item $\sigmamin\left(\EDP[{\bm z}{\bm x}^T]\right)\geq c$,\label{ass_SigmaminZX}
        \item $\left\|\EDP[{\bm z}{\bm z}^T]\right\|_{op}\leq C$, \label{ass_NormZZ}
        \item $\left\|\EDP[{\bm z}{\bm x}^T]\right\|_{op}\leq C$,\label{ass_NormZX}
        \item $\EDP[\|{\bm x}\|_2]\leq C$,\label{ass_NormX}
        \item $\left\|\EDP[{\bm z}\epsilon]\right\|_{2}\leq C$.\label{ass_NormZEps}
    \end{enumerate}
\end{assumption}
Next, we need some kind of uniform central limit theorem and law of large numbers.
\begin{assumption}\label{ass_CLTGen}
It holds that
$$\lim_{n\to\infty}\sup_{P\in \mathcal P}\sup_{w\in \VDP(\zeta)}\sup_{t\in \mathbb R}\left|\Prob_P\left(\frac{1}{\sigma_w\sqrt{n_0}}\sum_{i\in \mathcal D}\left(u_i^w - \EP[u_i^w]\right)\leq t\right)-\Phi(t)\right| = 0$$
with $u_i^w = (w({\bm z}_i) + {\bm a}_w^T{\bm z}_i)\epsilon_i$.
\end{assumption}
Due to the boundedness of $w$, this assumption can often be motivated using the Lindeberg--Feller central limit Theorem \citep[e.g.,][Thm 3.4.10.]{DurrettProb}. For example, we will see that it holds, when assuming dependence within small clusters but independence between clusters (as considered in Assumption \ref{ass_Cluster} and Proposition \ref{prop_AsNormCluster} below), and uniformly bounded moments, namely for some $\eta >0$, $\sup_{P \in \mathcal{P}} \sup_{i\in \mathbb N} \EP[|\epsilon_i|^{2 + \eta}] \le C < \infty$ and $\sup_{P \in \mathcal{P}}\sup_{i\in \mathbb N} \EP[ |\epsilon_i|^{2+\eta} \|{\bm z}_i\|_2^{2 + \eta}] \leq C < \infty$ (see Section \ref{sec_Cluster} below for the details).

\begin{assumption}\label{ass_LLNGen}
It holds that
\begin{enumerate}
    \item $\sqrt{n_0}\|\hatED[{\bm z}\epsilon]-\EDP[{\bm z}\epsilon]\|_2=\OP(1),$\label{ass_LLNZEpsilon}
    \item $\|\hatED[{\bm x}{\bm z}^T]-\EDP[{\bm x}{\bm z}^T]\|_{op} = \oP(1),$\label{ass_LLNXZ}
    \item $\|\hatED[{\bm z}{\bm z}^T]-\EDP[{\bm z}{\bm z}^T]\|_{op} = \oP(1),$\label{ass_LLNZZ}
    \item $\|\hatED[w({\bm z}){\bm x}]-\EDP[w({\bm z}){\bm x}]\|_2 = \oPW(1).$\label{ass_LLNwX}
\end{enumerate}
\end{assumption}
Finally, we also need a consistent estimator for the asymptotic variance $\sigma_w^2$.
\begin{assumption}\label{ass_SigmaHatGen}
    We have an estimator $\hat\sigma_w^2$ of $\sigma_w^2$ that satisfies
    $$\sigma_w^2-\hat\sigma_w^2 = \oPW(1).$$
\end{assumption}
Under Assumptions \ref{ass_MomentsGen}, \ref{ass_CLTGen}, \ref{ass_LLNGen} and \ref{ass_SigmaHatGen}, we have uniform asymptotic normality of the test statistic $\frac{1}{\sqrt{n_0}\hat\sigma_w}\sum_{i\in \mathcal D} w({\bm z}_i) \hat r_i$.

\begin{theorem}\label{thm_AsNormGen}
Let $N(w)$ be defined in eq. (8) in the main text and let $\hat\sigma_w^2$ be according to Assumption \ref{ass_SigmaHatGen}. Let $\zeta >0$ be arbitrary but fixed and assume that Assumptions \ref{ass_MomentsGen}, \ref{ass_CLTGen}, \ref{ass_LLNGen}, and \ref{ass_SigmaHatGen} hold. Then,
\begin{equation}\label{eq_AsNormGen}
        \lim_{n\to\infty}\sup_{P\in \mathcal P}\sup_{w\in \VDP(\zeta)}\sup_{t\in \mathbb R}\left|\Prob_P\left(\frac{N(w) - \sqrt{n_0}\left(\rho_w+(\hat {\bm a}_w- {\bm a}_w)^T{\bm{\tau}}\right)}{\hat\sigma_w}\leq t\right)-\Phi(t)\right| = 0.
\end{equation}
\end{theorem}
The proof of Theorem \ref{thm_AsNormGen} can be found in Section \ref{sec_ProofAsNormGen}.

\subsection{Baseline Residual Predicton Test: Application to Clustered Data}\label{sec_Cluster}
As an application of Theorem \ref{thm_AsNormGen}, we consider clustered data. As before, we have distributions $P\in \mathcal P$ determining the law of the sequence $(y_i, {\bm x}_i, {\bm z}_i)_{i\in \mathbb N}$, which now exhibits clustered dependence. As before, $\mathcal D \subset \{1,\ldots, n\}$ and $n_0 = |\mathcal D|\to\infty$ as $n\to\infty$. We assume that we can write $\mathcal D$ as a disjoint union $\mathcal D = \dot\bigcup_{g = 1}^G I_g$, where $G$ and $(I_g)_{g=1}^G$ depend on $\mathcal D$. For $g=1,\ldots, G$, let $n_g = |I_g|$. We assume independence across clusters and bounded cluster size. In particular, the number of clusters needs to grow as fast as order $n$. However, we allow for arbitrary dependence structure within the clusters. More formally, we have the following assumption.
\begin{assumption}\label{ass_Cluster}
    The collections $({\bm x}_i, y_i, {\bm z}_i)_{i\in I_g}$, are independent across $g \in \mathbb N$. Moreover,  $\bar n\coloneqq \sup_{g\in \mathbb N} n_g<\infty$.
\end{assumption}

We need a consistent estimator for $\sigma_w^2$ defined in \eqref{eq_DefSigmaWGen}. This is slightly more subtle. Using basic manipulations, one can derive the expression 
\begin{equation}\label{eq_SigmaWCluster}
    \sigma_w^2 = \frac{1}{n_0}\sum_{g=1}^G\E_P\left[s_g(w)^2\right]-\frac{1}{n_0}\sum_{g=1}^G\EP\left[s_g(w)\right]^2
\end{equation}
with
$$s_g(w)\coloneqq \sum_{i\in I_g}(w({\bm z}_i) + {\bm a}_w^T {\bm z}_i)\epsilon_i.$$
Note that $s_g(w)$ also depends on $P$ through ${\bm a}_w$ and $\epsilon$.
The second part $\frac{1}{n_0}\sum_{g=1}^G\EP[s_g(w)]^2$ can in general only be consistently estimated when $\E_P[s_g(w)]$ does not depend on $g$. Most importantly, this is the case under $H_0: \E[\epsilon_i|{\bm z}_i]= 0$ for $i=1,\ldots, n$. That is, under the null hypothesis, the variance of our test statistic can be consistently estimated. But more generally, the assumption that $\EP[s_g(w)]$ does not depend on $g$ is also satisfied if the clusters are identically distributed.
Define the variance estimator
\begin{equation}\label{eq_SigmaHatCluster}
    \hat\sigma_w^2 \coloneqq \frac{1}{n_0}\sum_{g=1}^G\hat s_g(w)^2-\frac{n_0}{G}\left(\frac{1}{n_0}\sum_{g=1}^G\sum_{i\in I_g}w({\bm z}_i)\hat r_i)\right)^2
\end{equation}
with
$$\hat s_g(w) = \sum_{i\in I_g}(w({\bm z}_i) + \hat {\bm a}_w^T {\bm z}_i)\hat r_i.$$
We now give the assumptions that are needed for Theorem \ref{thm_AsNormGen} to hold in the clustered data setting.
\begin{assumption}\label{ass_SigmaMinCluster}
    There exists $c>0$ such that for all $n\in \mathbb N$ and for all $P\in \mathcal P$, $\sigma_{\min}(\EDP[{\bm z}{\bm z}^T])\geq c$ and $\sigma_{\min}(\EDP[{\bm z}{\bm x}^T])\geq c$.
\end{assumption}

\begin{assumption}\label{ass_MomentsCluster}
    There exist $\eta, C\in (0,\infty)$ such that for all $i\in \mathbb N$ and $P\in \mathcal P$,
    \begin{enumerate}
        \item $\E_P[|\epsilon_i|^{2+\eta}]\leq C$, \label{ass_EpsilonCluster}
        \item $\E_P[\|{\bm z}_i\|_2^{2+\eta}|\epsilon_i|^{2+\eta}]\leq C$, \label{ass_EpsilonZCluster}
        \item $\EP[\|{\bm z}_i\|_2^{2+\eta}]\leq C$,\label{ass_ZZCluster}
        \item $\EP[\|{\bm x}_i\|_2^2\|{\bm z}_i\|_2^2]<C$,\label{ass_XZCluster}
        \item $\EP[\|{\bm x}_i\|_2^2]<C$.\label{ass_XCluster}
    \end{enumerate}
\end{assumption}
\begin{proposition}\label{prop_AsNormCluster}
    If Assumptions \ref{ass_Cluster}, \ref{ass_SigmaMinCluster} and \ref{ass_MomentsCluster} hold and if for all $w\in\mathcal W$ and $P\in \mathcal P$, $\EP[s_1(w)]=\ldots=\EP[s_G(w)]$, then \eqref{eq_AsNormGen}, the statement of Theorem \ref{thm_AsNormGen}, holds. In particular, \eqref{eq_AsNormGen} holds if $\E[\epsilon_i|{\bm z}_i] = 0$ for all $i\in \mathbb N$ which is the null-hypothesis of interest.
\end{proposition}
\begin{remark}
    It is important to emphasize that the validity of p-values in our procedure relies on the independence between the auxiliary sample used to estimate the weight function and the main sample. In the case of clustered data, this requirement is straightforward to satisfy by respecting the cluster structure during sample splitting, i.e., taking a random subset of entire clusters as the auxiliary sample. However, applying the general theory from Section \ref{sec_TheoryMoreGeneralForm} to other forms of dependent data, such as time series, can present challenges. In such settings, it may not be feasible to split the data into two genuinely independent subsamples, potentially compromising the validity of the test.
\end{remark}

\subsection{Weak-IV-Robust Extension: Theory in More General Form}
We consider the analogous setup as in Section \ref{sec_TheoryMoreGeneralForm} with the difference that we now consider distributions $P\in \mathcal P$ determining the law of the sequence $(\bm x_i, y_i, \bm z_i, \bm c_i)_{i\in \mathbb N}$. Definition \ref{def_StochOrder} and \eqref{eq_DefEDP} can be analogously formulated for this setup. The following assumption ensures that distributions in $\mathcal P$ satisfy the null hypothesis $H_0(\bm \beta_0)$ for some $\bm \beta_0$.
\begin{assumption}\label{ass_BetaP}
For all $P\in \mathcal P$, there exist $\bm \beta = \bm \beta(P)\in \mathbb R^p$ and $\bm \theta = \bm \theta(P)\in \mathbb R^q$ such that
$$\E_P\left[y_i - \bm x_i^T\bm \beta(P) - \bm c_i^T\bm \theta(P)|\bm z_i, \bm c_i\right] = 0,\, i \in \mathcal D.$$
\end{assumption}
Define $\nu_i \coloneqq y_i - \bm x_i^T\bm \beta(P) - \bm c_i^T\bm \theta(P)$. Similarly to (25) and (26) in the main text, define for any $w:\mathbb R^d\times \mathbb R^q\to \mathbb R$,
\begin{equation}\label{eq_DefLambdaW}
    \bm \lambda_w\coloneqq\EDP[\bm c\bm c^T]^{-1}\EDP[\bm c w(\bm z, \bm c)], \quad \bar w_i = w(\bm z_i, \bm c_i) - \bm c_i^T\bm \lambda_w.
\end{equation}
Generalizing (31) and (32) in the main text, let
\begin{align}
    \bar\sigma_w^2 &\coloneqq \Var\left(\frac{1}{\sqrt n_0}\sum_{i\in \mathcal D} \nu_i \bar w_i\right),\\
    \mathcal U_{\mathcal D, P}(\zeta)&\coloneqq \left\{w:\mathbb R^d\times \mathbb R^q\to [-1,1]\,\vert\, \bar \sigma_w^2\geq \zeta\right\}.
\end{align}
We need similar assumptions to \ref{sec_TheoryMoreGeneralForm}.
\begin{assumption}\label{ass_GenMomentsWeak}
    There exist $c, C\in (0, \infty)$ such that for all $n\in \mathbb N$ and for all $P\in \mathcal P$,
    \begin{enumerate}
        \item $\sigma_{\min}(\EDP[\bm c\bm c^T])\geq c,$
        \item $\EDP[\|\bm c\|_2]\leq C$.
    \end{enumerate}
\end{assumption}
\begin{assumption}\label{ass_GenNormWeak}
    It holds that
$$\lim_{n\to\infty}\sup_{P\in \mathcal P}\sup_{w\in \mathcal U_{\mathcal D, P}(\zeta)}\sup_{t\in \mathbb R}\left|\Prob_P\left(\frac{1}{\bar\sigma_w\sqrt{n_0}}\sum_{i\in \mathcal D}\nu_i \bar w_i\leq t\right)-\Phi(t)\right| = 0.$$

\end{assumption}
\begin{assumption}\label{ass_GenULLNWeak}
    It holds that
    \begin{enumerate}
        \item $\sqrt{n_0}\hatED[\bm c \nu] = \OP(1)$,\label{ass_GenULLNWeak1}
        \item $\|\hatED[\bm c\bm c^T]-\EDP[\bm c\bm c^T]\|_{op} = \oP(1)$,\label{ass_GenULLNWeak2}
        \item $\|\hatED[w(\bm z, \bm c)\bm c]-\EDP[w(\bm z, \bm c) \bm c]\|_2 = \oPW(1)$.\label{ass_GenULLNWeak3}
    \end{enumerate}
\end{assumption}
For assertion \ref{ass_GenULLNWeak1}, recall that $\EDP[\bm c \nu] = 0$ by Assumption \ref{ass_BetaP}.
\begin{assumption}\label{ass_GenVarWeak}
    We have a variance estimator function $\hat\sigma_w(\cdot)$ that satisfies
    $$\bar\sigma_w^2 - \hat\sigma_w^2(\bm \beta(P))=\oPW(1).$$
\end{assumption}
With these assumptions, we get uniform asymptotic normality.
\begin{theorem}\label{thm_GenAsympWeak}
    Let $N(w, \bm \beta(P))$ be defined as in eq. (27) in the main text (plugging in $\bm \beta(P)$ for $\bm \beta_0$, and let $\hat \sigma_w^2(\cdot)$ be defined according to Assumption \ref{ass_GenVarWeak}. Let $\zeta >0$ be arbitrary but fixed and assume that Assumptions \ref{ass_GenMomentsWeak}, \ref{ass_GenNormWeak}, \ref{ass_GenULLNWeak}, and \ref{ass_GenVarWeak} hold. Then,
    \begin{equation}\label{eq_GenAsympWeak}
        \lim_{n\to\infty}\sup_{P\in \mathcal P}\sup_{w\in \mathcal U_{\mathcal D, P}(\zeta)}\sup_{t\in \mathbb R}\left|\Prob_P\left(\frac{N(w, \bm \beta(P))}{\hat\sigma_w(\bm \beta(P))}\leq t\right)-\Phi(t)\right| = 0.
    \end{equation}
\end{theorem}

\subsection{Weak-IV-Robust Extension: Application to Clustered Data}
We now specialize the general weak-identification robust theory to the case of clustered data. As in Section \ref{sec_Cluster}, let $P\in \mathcal P$ determine the law of the sequence $(\bm x_i, y_i, \bm z_i, \bm c_i)_{i\in \mathbb N}$, and let $\mathcal D \subset \{1,\ldots,n\}$ with $n_0 = |\mathcal D|\to\infty$ as $n\to\infty$. We assume that $\mathcal D$ can be written as a disjoint union
$$
\mathcal D = \dot\bigcup_{g=1}^G I_g,
$$
where $G$ and $(I_g)_{g=1}^G$ may depend on $\mathcal D$, and $n_g = |I_g|$. As before, we assume independence across clusters and bounded cluster size.

\begin{assumption}\label{ass_ClusterWeak}
    The collections $(\bm x_i, y_i, \bm z_i, \bm c_i)_{i\in I_g}$ are independent across $g\in\mathbb N$. Moreover, $\bar n \coloneqq \sup_{g\in\mathbb N} n_g < \infty$.
\end{assumption}

Throughout this section, we consider distributions $P\in\mathcal P$ satisfying Assumption \ref{ass_BetaP}. That is, for each $P\in\mathcal P$, there exist $\bm \beta(P)\in\mathbb R^p$ and $\bm \theta(P)\in\mathbb R^q$ such that $\E_P\!\left[y_i - \bm x_i^T\bm \beta(P) - \bm c_i^T\bm \theta(P)\,\middle|\, \bm z_i, \bm c_i\right] = 0$
for all $i\in\mathcal D$. Define $\nu_i \coloneqq y_i - \bm x_i^T\bm \beta(P) - \bm c_i^T\bm \theta(P)$.

Similarly to \eqref{eq_SigmaWCluster}, we can write
\begin{equation}\label{eq_SigmaWeakCluster}
    \bar\sigma_w^2 
    = \frac{1}{n_0}\sum_{g=1}^G \EDP\!\left[\bar s_g(w)^2\right]
\end{equation}
where
$$
\bar s_g(w) \coloneqq \sum_{i\in I_g} \nu_i \bar w_i.
$$

We use the cluster robust variance estimator
\begin{equation}\label{eq_HatSigmaWeakCluster}
    \hat\sigma_w^2(\bm \beta_0) = \frac{1}{n_0}\sum_{g = 1}^G\tilde s_g(w, \bm \beta_0)^2- \frac{n_0}{G}\left(\frac{1}{n_0}\sum_{g = 1}^G\tilde s_g(w, \bm \beta_0)\right)^2,
\end{equation}
where
$$\tilde s_g(w, \bm \beta_0)  \coloneqq \sum_{i\in I_g} \tilde w_i\tilde r_i(\bm \beta_0)$$
and $\tilde w_i$ and $\tilde r_i(\bm \beta_0)$ are defined in eq. (26) in the main text.

We impose the following moment conditions.

\begin{assumption}\label{ass_MomentsWeakCluster}
    There exist $\eta, c, C\in (0,\infty)$ such that for all $i\in\mathbb N$ and $P\in\mathcal P$,
    \begin{enumerate}
        \item $\sigma_{\min}(\EDP[\bm c\bm c^T])\geq c,$
        \item $\E_P[|\nu_i|^{2+\eta}] \le C,$
        \item $\E_P[\|\bm c_i\|_2^{2+\eta}|\nu_i|^{2+\eta}] \le C,$
        \item $\E_P[\|\bm c_i\|_2^{4+\eta}] \le C.$
    \end{enumerate}
\end{assumption}

\begin{proposition}\label{prop_AsNormWeakCluster}
    If Assumptions \ref{ass_ClusterWeak} and \ref{ass_MomentsWeakCluster} hold and if Assumption \ref{ass_BetaP} is satisfied, then \eqref{eq_GenAsympWeak}, the statement of Theorem \ref{thm_GenAsympWeak} holds.
\end{proposition}

As in the baseline case, independence between the auxiliary sample and the main sample is essential for the validity of the p-values. In the clustered setting, this can be ensured by performing sample splitting at the cluster level, i.e., assigning entire clusters either to the auxiliary or to the main sample.

\section{Proofs}
\subsection{Proof of Theorem \ref{thm_AsNormGen}}\label{sec_ProofAsNormGen}
By definition, $y = \bm x^T\bm \beta^* + \epsilon$ and we have that $\hat{\bm{\beta}} = \hat {\bm M}\hatED[{\bm z}y] = {\bm{\beta}^*} +\hat {\bm M} \hatED[{\bm z}\epsilon]$. It follows that $\hat r_i =\epsilon_i +{\bm x}_i^T({\bm{\beta}^*}-\hat{\bm{\beta}}) = \epsilon_i - {\bm x}_i^T\hat {\bm M}\hatED[{\bm z}\epsilon]$. Hence, we can rewrite
\begin{align}
    \frac{1}{\sqrt{n_0}}\sum_{i\in \mathcal D}w({\bm z}_i)\hat r_i 
    &= \frac{1}{\sqrt{n_0}}\sum_{i\in \mathcal D}w({\bm z}_i)\epsilon_i - \frac{1}{\sqrt{n_0}}\sum_{i\in \mathcal D} w({\bm z}_i) {\bm x}_i^T \hat {\bm M}\hat \E_\mathcal D[{\bm z}\epsilon]\nonumber\\
    &=\frac{1}{\sqrt{n_0}}\sum_{i\in \mathcal D}w({\bm z}_i)\epsilon_i - \hat\E_\mathcal D[w({\bm z}){\bm x}^T] \hat {\bm M}\frac{1}{\sqrt{n_0}}\sum_{i\in \mathcal D}{\bm z}_i\epsilon_i\nonumber\\
    &=\frac{1}{\sqrt{n_0}}\sum_{i\in \mathcal D} (w({\bm z}_i)+ \hat {\bm a}_w^T {\bm z}_i)\epsilon_i\label{eq_NEquiv}
\end{align}
with $\hat {\bm a}_w$ defined in eq. (10) in the main text.

Next, observe that from the definitions of $\epsilon$, ${\bm{\beta}}^*$ and ${\bm M}$,
\begin{align}
    {\bm M} \EDP[{\bm z}\epsilon]
    &={\bm M}\EDP[{\bm z}(y-{\bm x}^T{\bm{\beta}}^*)]\nonumber\\
    &={\bm M}\EDP[{\bm z}(y-{\bm x}^T{\bm M}\EDP[{\bm z}y])]\nonumber\\
    &={\bm M}(I-\EDP[{\bm z}{\bm x}^T]{\bm M})\EDP[{\bm z}y]\nonumber\\
    &=({\bm M}- {\bm M}\EDP[{\bm z}{\bm x}^T]{\bm M})\EDP[{\bm z}y]\nonumber\\
    &=0,\label{eq_MEZEpsilon}
\end{align}
since ${\bm M}\EDP[{\bm z}{\bm x}^T] = I_p$.
Using the definition  of ${\bm a}_w$, we have that
$$\EDP[{\bm a}_w^T{\bm z}\epsilon]=-\EDP[w(\bm z) \bm x^T]{\bm M}\EDP[{\bm z}\epsilon]=0$$
and
\begin{equation}\label{eq_IdentRhoW}
    \frac{1}{n_0}\sum_{i\in \mathcal D}\E_P[(w({\bm z}_i)+{\bm a}_w^T{\bm z}_i)\epsilon_i] = \EDP[(w({\bm z})+{\bm a}_w^T {\bm z})\epsilon]=\EDP[w({\bm z})\epsilon]=\rho_w.
\end{equation}

With \eqref{eq_NEquiv}, we can write the numerator of the quantity from  Theorem \ref{thm_AsNormGen} as
\begin{align*}
    &\frac{1}{\sqrt{n_0}}\sumD w({\bm z}_i) \hat r_i - \sqrt{n_0}\left(\rho_w + (\hat {\bm a}_w - {\bm a}_w)^T {\bm{\tau}}\right)\\
    &=\frac{1}{\sqrt{n_0}}\sumD (w({\bm z}_i) + \hat {\bm a}_w^T{\bm z}_i)\epsilon_i - \sqrt{n_0}\left(\rho_w + (\hat {\bm a}_w - {\bm a}_w)^T {\bm{\tau}}\right)\\
    & = \underbrace{\frac{1}{\sqrt{n_0}}\sumD (w({\bm z}_i) + {\bm a}_w^T {\bm z}_i)\epsilon_i - \sqrt{n_0}\rho_w}_{(*)}+(\hat {\bm a}_w- {\bm a}_w)^T\left(\frac{1}{\sqrt {n_0}}\sumD  {\bm z}_i\epsilon_i-\sqrt{n_0}{\bm{\tau}}\right).
\end{align*}
From \eqref{eq_IdentRhoW}, we know that $(*)$ is equal to $\sigma_w$ times the quantity from Assumption \ref{ass_CLTGen}.
Since for $w\in \VDP(\zeta)$, we have $\sigma_w^2\geq \zeta$, we can use Lemma \ref{lem_USlutsky} below (uniform version of Slutsky's lemma) together with Assumption \ref{ass_SigmaHatGen}, and it is enough to show that
\begin{equation}
    (\hat {\bm a}_w-{\bm a}_w)^T\left(\frac{1}{\sqrt {n_0}}\sumD  {\bm z}_i\epsilon_i-\sqrt{n_0}{\bm{\tau}}\right) = \oPW(1). \label{eq_RestTerm}
\end{equation}
By assertion \ref{ass_LLNZEpsilon} of Assumption \ref{ass_LLNGen} and the definition of ${\bm{\tau}}$, we have that
$$\frac{1}{\sqrt {n_0}}\sumD  {\bm z}_i\epsilon_i-\sqrt{n_0}{\bm{\tau}}=\sqrt{n_0}(\hatED[{\bm z}\epsilon]-\EDP[{\bm z}\epsilon])=\OP(1).$$
Using Lemma \ref{lem_ProdOrder} and \eqref{eq_ConsAw} from Lemma \ref{lem_AwBwCons} below, \eqref{eq_RestTerm} follows, which concludes the proof.
\begin{lemma}\label{lem_AwBwCons}
    Under Assumptions \ref{ass_MomentsGen} and \ref{ass_LLNGen}, we have that $\|\bm M\|_{op}$ and $\|{\bm a}_w\|_2$ are uniformly bounded in $n\in \mathbb N$, $P\in \mathcal P$ and $w\in \mathcal W$. Moreover
    \begin{align}
        \|\hat {\bm M}- \bm M\|_{op}&=\oP(1),\label{eq_ConsM}\\
        \|\hat {\bm a}_w - {\bm a}_w\|_2&=\oPW(1)\label{eq_ConsAw},\\
        \|\hat {\bm{\beta}} - {\bm{\beta}}^*\|_2 &=\oP(1)\label{eq_ConsBeta},
    \end{align}
     where the corresponding norm is defined to be $\infty$ if $\hat {\bm M}$, or $\hat {\bm a}_w$ does not exist.
\end{lemma}
\begin{proof}
    For \eqref{eq_ConsM}, we use Assumption \ref{ass_MomentsGen} to observe that
    $$\|\EDP[{\bm z}{\bm z}^T]^{-1}\|_{op}= \sigmamin(\EDP[{\bm z}{\bm z}^T])^{-1}\leq 1/c$$
    is uniformly bounded in $n\in \mathbb N$ and $P\in \mathcal P$. Moreover, 
    \begin{align*}
        &\left\|\left[\EDP[{\bm x} {\bm z}^T]\EDP[{\bm z} {\bm z}^T]^{-1} \EDP[{\bm z} {\bm x}^T]\right]^{-1}\right\|_{op}\\
        &= \sigmamin\left(\EDP[{\bm x} {\bm z}^T]\EDP[{\bm z} {\bm z}^T]^{-1} \EDP[{\bm z} {\bm x}^T]\right)^{-1}\\
        &\leq \sigmamin(\EDP[{\bm z}{\bm x}^T])^{-2}\|\EDP[{\bm z}{\bm z}^T]\|_{op}\\
        &\leq C/c^2.
    \end{align*}
    Hence, \eqref{eq_ConsM} follows from the definitions $\bm M$ and $\hat {\bm M}$, Assumptions \ref{ass_MomentsGen} and \ref{ass_LLNGen} and successive applications of Lemma \ref{lem_ProdMatrix} below.
    For \eqref{eq_ConsAw}, note that for $w\in \mathcal W$ by Assumption \ref{ass_MomentsGen}
    $$\|\EDP[w(\bm z) \bm x^T]\|_2\leq\EDP[\|\bm x\|_2]\leq C$$
    and hence, \eqref{eq_ConsAw} follows from \eqref{eq_ConsM} and Assumption \ref{ass_LLNGen} with Lemma \ref{lem_ProdMatrix} below.

   For \eqref{eq_ConsBeta}, observe that from the definition of $\hat {\bm{\beta}}$ and using \eqref{eq_MEZEpsilon}
   \begin{equation*}
    \|{\bm{\beta}}^*-\hat{\bm{\beta}}\|_2= \|\hat {\bm M}\hat\E_\mathcal D[{\bm z}\epsilon]\|_2=\|\hat {\bm M}\hat\E_\mathcal D[{\bm z}\epsilon]-{\bm M}\EDP [{\bm z}_i\epsilon_i]\|_2,
\end{equation*}
which is $\oP(1)$ by \eqref{eq_ConsM}, Assumption \ref{ass_MomentsGen}, Assumption \ref{ass_LLNGen}, and Lemma \ref{lem_ProdMatrix} below.
\end{proof}
\begin{lemma}\label{lem_ProdMatrix}
Let ${\bm A_n}, {\bm B_n}$ be two sequences of fixed matrices that depend on $P\in \mathcal P$ and let $\bar {\bm A}_n, \bar {\bm B}_n$ be sequences of random matrices that also depend on $P\in \mathcal P$. Assume that $\|\bar {\bm A}_n - {\bm A_n}\|_{op} = \oP(1)$ and $\|\bar {\bm B}_n - {\bm B_n} \|_{op} = \oP(1)$.
\begin{enumerate}
    \item Assume that the matrices have conformable dimensions, $\sup_{n\in \mathbb N}\sup_{P\in \mathcal P}\|{\bm A_n}\|_{op}<\infty$ and $\sup_{n\in \mathbb N}\sup_{P\in \mathcal P}\|{\bm B_n}\|_{op}<\infty$. Then, it holds that $\|\bar {\bm A}_n\bar {\bm B}_n- {\bm A_n}{\bm B_n}\|_{op} = \oP(1)$.\label{sublem_Prod}
    \item Assume that ${\bm A}_n^{-1}$ exists for all $n\in \mathbb N$ and $P\in \mathcal P$ and that $\sup_{n\in \mathbb N}\sup_{P\in \mathcal P}\|{\bm A}_n^{-1}\|_{op}<\infty$. Then, it holds that $\|\bar {\bm A}_n^{-1} - {\bm A}_n^{-1}\|_{op} = \oP(1)$ (where we set $\|\bar {\bm A}_n^{-1} - {\bm A}_n^{-1}\|_{op}=\infty$ if $\bar {\bm A}_n$ is not invertible).\label{sublem_Inv}

\end{enumerate}
\end{lemma}
\begin{proof}
    For ease of notation, we sometimes omit the dependence on $n$ in the following.
    For \ref{sublem_Prod}, note that by the triangle inequality and the submultiplicativity of the operator norm,
    $$\|{\bm A}{\bm B}-\bar {\bm A}\bar {\bm B}\|_{op}\leq \|{\bm A}\|_{op}\|{\bm B}-\bar {\bm B}\|_{op} + \|{\bm A}-\bar {\bm A}\|_{op}\|{\bm B}\|_{op} + \|{\bm A}-\bar {\bm A}\|_{op}\|\bar {\bm B}- {\bm B}\|_{op},$$
    from which we can easily conclude using the uniform boundedness of $\|{ {\bm A}_n}\|_{op}$ and $\|{{\bm B}_n}\|_{op}.$

    For \ref{sublem_Inv}, note that $\bar {\bm A} = {\bm A}(I+{\bm A}^{-1}(\bar {\bm A}-{\bm A}))$. That means if $\|{\bm A}^{-1}\|_{op}\|\bar {\bm A}-{\bm A}\|_{op}<1$, then $\bar {\bm A}$ is invertible using the Neumann series. Moreover,
    $$\|\bar {\bm A}^{-1}\|_{op} = \|\bar {\bm A}^{-1}(\bar {\bm A}-{\bm A}){\bm A}^{-1} - {\bm A}^{-1}\|_{op}\leq\|\bar {\bm A}^{-1}\|_{op}\|\bar {\bm A}-{\bm A}\|_{op}\|{\bm A}^{-1}\|_{op} + \|{\bm A}^{-1}\|_{op},$$
    hence if $\|{\bm A}^{-1}\|_{op}\|\bar {\bm A}-{\bm A}\|_{op}<1$, then
    $$\|\bar {\bm A}^{-1}\|_{op}\leq \frac{\|{\bm A}^{-1}\|_{op}}{1-\|{\bm A}^{-1}\|_{op}\|\bar {\bm A}-{\bm A}\|_{op}}.$$
    Consequently,
    \begin{align*}
       \|\bar {\bm A}^{-1}-{\bm A}^{-1}\|_{op} &= \|\bar {\bm A}^{-1}(\bar {\bm A}-{\bm A}){\bm A}^{-1}\|_{op}\\
       &\leq \|\bar {\bm A}^{-1}\|_{op}\|\bar {\bm A}-{\bm A}\|_{op}\|{\bm A}^{-1}\|_{op}\\
       &\leq \frac{\|{\bm A}^{-1}\|_{op}^2\|\bar {\bm A}-{\bm A}\|_{op}}{1-\|{\bm A}^{-1}\|_{op}\|\bar {\bm A}-{\bm A}\|_{op}}. 
    \end{align*}
    Since $\sup_{n\in \mathbb N}\sup_{P\in \mathcal P}\|{\bm A}_n^{-1}\|_{op}<\infty$ and $\|\bar {\bm A}_n-{\bm A_n}\|_{op} = \oP(1)$, the result follows.
\end{proof}
\subsection{Proof of Proposition \ref{prop_AsNormCluster}}
We can apply Theorem \ref{thm_AsNormGen} and need to verify Assumptions \ref{ass_MomentsGen}-\ref{ass_SigmaHatGen}.
We will repeatedly use the following fact.
\begin{lemma}\label{lem_BoundEp}
    Under Assumption \ref{ass_MomentsGen}, there exists a constant $C_1>0$ such that for all $s>1$, for all $w\in \mathcal W$ and all $P\in \mathcal P$,
    $$\EP\left[\left|(w({\bm z}_i)+{\bm a}_w^T{\bm z}_i)\epsilon_i-\E_P[(w({\bm z}_i) + {\bm a}_w^T{\bm z}_i)\epsilon_i]\right|^s\right]\leq C_1^s \left(\EP[|\epsilon_i|^{s}] +\EP[\|{\bm z}_i\|_2^s|\epsilon_i|^s]\right).$$
\end{lemma}
\begin{proof}
    Follows from the Cauchy--Schwarz inequality, the triangle inequality, and the fact that ${\bm a}_w$ is uniformly bounded (see  Lemma \ref{lem_AwBwCons}).
\end{proof}

\subsubsection{Verification of Assumption \ref{ass_MomentsGen}}
Assertions \ref{ass_SigmaminZZ} and \ref{ass_SigmaminZX} are identical to Assumption \ref{ass_SigmaMinCluster}. For assertions \ref{ass_NormZZ}, \ref{ass_NormZX}, \ref{ass_NormX}, and \ref{ass_NormZEps} we can use assertions \ref{ass_ZZCluster}, \ref{ass_XZCluster}, \ref{ass_XCluster}, and \ref{ass_EpsilonZCluster} of Assumption \ref{ass_MomentsCluster} and use linearity of expectation and the triangle inequality.

\subsubsection{Verification of Assumption \ref{ass_CLTGen}}\label{sec_VerifyAssCLTGen}
Define
    $$r_g(w) = \frac{1}{\sigma_w\sqrt{n_0}}\sum_{i\in I_g}\left\{(w({\bm z}_i) + {\bm a}_w^T{\bm z}_i)\epsilon_i - \E_P[(w({\bm z}_i) + {\bm a}_w^T{\bm z}_i)\epsilon_i]\right\}.$$
    Since $\mathcal D = \dot\bigcup_{g = 1}^G I_g$, we have to show that
    $$\lim_{n\to\infty}\sup_{P\in \mathcal P}\sup_{w\in \VDP(\gamma)}\sup_{t\in \mathbb R}\left|\Prob_P\left(\sum_{g=1}^G r_g(w)\leq t\right)-\Phi(t)\right|=0.$$
     Let $P_n\in \mathcal P$ and $w_n\in \mathcal V_{\mathcal D, P_n}(\gamma)$ be such that
     \begin{align*}
         &\sup_{P\in \mathcal P}\sup_{w\in \VDP(\gamma)}\sup_{t\in \mathbb R}\left|\Prob_P\left(\sum_{g=1}^G r_g(w)\leq t\right)-\Phi(t)\right|\\
         &\leq \sup_{t\in \mathbb R} \left|\Prob_{P_n}\left(\sum_{g=1}^G r_g(w_n)\leq t\right)-\Phi(t)\right|+ \frac{1}{n}.
     \end{align*}
    Since $\EP[r_g(w_n)] = 0$ and by the definition of $\sigma_w^2$, $\sum_{g=1}^G\E_{P_n}[r_g(w_n)^2] = 1$, we can use the Lindeberg--Feller Theorem \citep[see for example][Thm. 3.4.10.]{DurrettProb}, that is, it is enough to show that for all $\delta >0$,
    $ \lim_{n\to\infty}\sum_{g=1}^G\E_{P_n}\left[r_g(w_n)^2; |r_g(w_n)|>\delta\right]=0.$
    This, in turn, follows from
    \begin{equation}\label{eq_Lindeberg}
        \lim_{n\to\infty}\sup_{P\in \mathcal P}\sup_{w\in \VDP(\gamma)}\sum_{g=1}^G\E_{P}\left[|r_g(w)|^{2+\eta}\right] = 0.
    \end{equation}
    To verify \eqref{eq_Lindeberg}, recall that $n_g = |I_g|$. From H\"older's inequality and Jensen's inequality, the fact that $w\in \VDP(\zeta)$ and Lemma \ref{lem_BoundEp},
    \begin{align*}
        \EP\left[|r_g(w)|^{2+\eta}\right]&\leq\frac{1}{\sigma_w^{2+\eta}n_0^{1+\eta/2}}\left|\sum_{i\in I_g}\left\{(w({\bm z}_i) + {\bm a}_w^T{\bm z}_i)\epsilon_i - \E_P[(w({\bm z}_i) + {\bm a}_w^T{\bm z}_i)\epsilon_i]\right\}\right|^{2+\eta}\\
        &\leq\frac{n_g^{1+\eta}}{\sigma_w^{2+\eta}n_0^{1+\eta/2}}\sum_{i\in I_g}\left|(w({\bm z}_i) + {\bm a}_w^T{\bm z}_i)\epsilon_i - \E_P[(w({\bm z}_i) + {\bm a}_w^T{\bm z}_i)\epsilon_i]\right|^{2+\eta}\\
        &\leq \frac{n_g^{1+\eta}}{\zeta^{1+\eta/2}n_0^{1+\eta/2}}C_1^{2+\eta}\sum_{i\in I_g}\left(\EP[|\epsilon_i|^{2+\eta}] +\EP[\|{\bm z}_i\|_2^{2+\eta}|\epsilon_i|^{2+\eta}]\right)\\
        &\leq \frac{2n_g^{2+\eta}}{\zeta^{1+\eta/2}n_0^{1+\eta/2}}C_1^{2+\eta}C,
    \end{align*}
    where we used assertions \ref{ass_EpsilonCluster} and \ref{ass_EpsilonZCluster} of Assumption \ref{ass_MomentsCluster}. Hence, \eqref{eq_Lindeberg} follows since
    $$\sum_{g=1}^G \frac{n_g^{2+\eta}}{n_0^{1+\eta/2}}\leq \frac{\bar n^{1+\eta}}{n_0^{\eta/2}}\to 0.$$
    by Assumption \ref{ass_Cluster}.

\subsubsection{Verification of Assumption \ref{ass_LLNGen}}\label{sec_VeriAssLLNGen}
We start with assertion \ref{ass_LLNZEpsilon}. For simplicity, assume that $d=1$ (i.e., ${\bm z}$ is univariate). Otherwise, one can apply the same argument across the components of ${\bm z}$.
    Using Markov's inequality, it is enough to show that 
    $$\sup_{P\in \mathcal P}\E_P\left[n_0\left(\hatED[{\bm z}\epsilon]-\EDP[{\bm z}\epsilon]\right)^2\right]<\infty.$$
    Since by Assumption \ref{ass_Cluster}, the clusters are independent, and using the Cauchy--Schwarz inequality,
    \begin{align*}
        \E_P\left[n_0\left(\hatED[{\bm z}\epsilon]-\EDP[{\bm z}\epsilon]\right)^2\right] &= \frac{1}{n_0}\E_P\left[\left(\sum_{g=1}^G\sum_{i\in I_g} ({\bm z}_i\epsilon_i-\EP[{\bm z}_i\epsilon_i])\right)^2\right]\\
        &=\frac{1}{n_0}\sum_{g=1}^G\EP\left[\left(\sum_{i\in I_g}({\bm z}_i\epsilon_i-\EP[{\bm z}_i\epsilon_i])\right)^2\right]\\
        &\leq \frac{1}{n_0}\sum_{g=1}^Gn_g\sum_{i\in I_g}\EP[{\bm z}_i^2\epsilon_i^2]\\
        &\leq \bar n\sup_{P\in \mathcal P}\sup_{i\in \mathbb N} \E_P[{\bm z}_i^2\epsilon_i^2],
    \end{align*}
    which is bounded by assertion \ref{ass_EpsilonZCluster} of Assumption \ref{ass_MomentsCluster} and Assumption \ref{ass_Cluster}.
    This completes the proof of assertion \ref{ass_LLNZEpsilon}.

    For assertion \ref{ass_LLNXZ}, also without loss of generality, assume that ${\bm x}$ and ${\bm z}$ are both univariate. Note that 
    $$\hatED[{\bm x}{\bm z}]- \EDP[{\bm x}{\bm z}]=\frac{1}{G}\sum_{g=1}^G\frac{G}{n_0}\sum_{i\in I_g}({\bm x}_i{\bm z}_i-\EP[{\bm x}_i {\bm z}_i]).$$
    Since the clusters are independent, we apply Lemma \ref{lem_ULLN}, and it is enough to show that
    $$\sup_{P\in \mathcal P}\sup_{g\in \mathbb N}\EP\left[\left|\sum_{i\in I_g}{\bm x}_i{\bm z}_i\right|^{1+\eta}\right]<\infty,$$
    which directly follows from H\"older's inequality, assertion \ref{ass_XZCluster} of Assumption \ref{ass_MomentsCluster} and Assumption \ref{ass_Cluster}.

    Assertion \ref{ass_LLNZZ} follows similarly with assertion \ref{ass_ZZCluster} of Assumption \ref{ass_MomentsCluster}.

    For assertion \ref{ass_LLNwX}, one can use the same reasoning since $|w|<1$ for $w\in \mathcal W$ and using assertion \ref{ass_XCluster} of Assumption \ref{ass_MomentsCluster}.

\subsubsection{Verification of Assumption \ref{ass_SigmaHatGen}}\label{sec_VeriAssSigmaHatGen}
Recall the definitions of $s_g(w)$ and $\hat s_g(w)$ from Section \ref{sec_Cluster}.
Since by assumption, $\EP[s_g(w)]$ does not depend on $g$, we have that
$$
    \EP\left[s_g(w)\right]=\frac{1}{G}\sum_{g=1}^G\EP[s_g(w)]
    =\frac{n_0}{G}\EDP[(w({\bm z})+{\bm a}_w^T{\bm z})\epsilon].
$$
Note that $\EDP[{\bm a}_w^T{\bm z}\epsilon]=-\EDP[w(\bm z) \bm x^T]\bm M\EDP[{\bm z}\epsilon] = 0$ by \eqref{eq_MEZEpsilon}, hence
$$
    \EP\left[s_g(w)\right] = \frac{n_0}{G}\EDP[w({\bm z})\epsilon]=\frac{1}{G}\sum_{g=1}^G\sum_{i\in I_g}\EP[w({\bm z}_i)\epsilon_i].
$$
Hence, for the second term in \eqref{eq_SigmaWCluster},
$$\frac{1}{n_0}\sum_{g=1}^G\EP[s_g(w)]^2 = \frac{n_0}{G}\left(\frac{1}{n_0}\sum_{g=1}^G\sum_{i\in I_g}\EP[w({\bm z}_i)\epsilon_i]\right)^2.$$
By Assumption \ref{ass_Cluster}, $n_0/G\leq \bar n$. Hence, it is enough to show
\begin{align}
    \frac{1}{n_0}\sum_{g=1}^G\sum_{i\in I_g} \left(w({\bm z}_i)\hat r_i - \EP[w({\bm z}_i)\epsilon_i]\right) = \oPW(1),\label{eq_2ndSigmaW}\\
    \frac{1}{n_0}\sum_{g=1}^G\left(\hat s_g(w)^2-\EP[s_g(w)^2]\right) = \oPW(1). \label{eq_1stSigmaW}
\end{align}
We start with \eqref{eq_2ndSigmaW}. Note that $\hat r_i=\epsilon_i + {\bm x}_i^T({\bm{\beta}}^* - \hat{\bm{\beta}})$.
On one hand, since $\bar n<\infty$ and $|w|\leq 1$, we have $\sup_{P\in \mathcal P}\sup_{g\in \mathbb N}\sup_{w\in \mathcal W}\EP[|\sum_{i\in I_g}w({\bm z}_i)\epsilon_i|^2]<\infty$ using assertion \ref{ass_EpsilonCluster} of Assumption \ref{ass_MomentsCluster}. With Lemma \ref{lem_ULLN}, it holds that 
$\frac{1}{n_0}\sum_{g=1}^G \sum_{i\in I_g} (w({\bm z}_i)\epsilon_i - \EP[w({\bm z}_i)\epsilon_i]) = \oPW(1).$

Moreover
$$\left|\frac{1}{n_0}\sum_{g=1}^G\sum_{i\in I_g}w({\bm z}_i){\bm x}_i^T({\bm{\beta}}^*-\hat {\bm{\beta}})\right|\leq \|{\bm{\beta}}^*-\hat{\bm{\beta}}\|_2\frac{1}{n_0}\sum_{g=1}^G\sum_{i\in I_g} \|{\bm x}_i\|_2.$$
This bound is independent of $w\in \mathcal W$ and the second factor is $\OP(1)$ by Lemma \ref{lem_UnifOP1}. Moreover, $\|{\bm{\beta}}^* - \hat {\bm{\beta}}\|_2 = \oP(1)$ by Lemma \ref{lem_AwBwCons}.
Hence, \eqref{eq_2ndSigmaW} follows using Lemma \ref{lem_ProdOrder}.

It remains to prove \eqref{eq_1stSigmaW},
which reduces to proving
\begin{align}
    \frac{1}{n_0}\sum_{g=1}^G\left(s_g(w)^2- \EP[s_g(w)^2]\right)=\oPW(1),\label{eq_OracleLLNCluster}\\
    \frac{1}{n_0}\sum_{g=1}^G\left(\hat s_g(w)^2 - s_g(w)^2\right)=\oPW(1).\label{eq_DiffOracleCluster}
\end{align}
For \eqref{eq_OracleLLNCluster}, we can use Lemma \ref{lem_ULLN} provided that $\sup_{P\in \mathcal P}\sup_{w\in \mathcal W}\sup_{g\in \mathbb N}\EP[|s_g(w)|^{2+\eta}]<\infty.$ But this follows similarly to before using that $\bar n<\infty$, $|w|\leq 1$ and assertions \ref{ass_EpsilonCluster} and \ref{ass_EpsilonZCluster} of Assumption \ref{ass_MomentsCluster}.

For \eqref{eq_DiffOracleCluster}, we use Lemma \ref{lem_SumDiffSquares}, and it is enough to prove
\begin{align}
    \sup_{w\in \mathcal W}\frac{1}{n_0}\sum_{g=1}^G s_g(w)^2 &= \OP(1),\label{eq_WEpsO1Cluster}\\
    \frac{1}{n_0}\sum_{g=1}^G(\hat s_g(w)- s_g(w))^2 &= \oPW(1).\label{eq_WEpsDiffCluster}
\end{align}
For \eqref{eq_WEpsO1Cluster}, note that using similar arguments as before,
$$\frac{1}{n_0}\sum_{g=1}^Gs_g(w)^2\leq \frac{2\bar n}{n_0}\sum_{g=1}^G\sum_{i\in I_g}\left(\epsilon_i^2 + \|{\bm a}_w\|_2^2\|{\bm z}_i\|_2^2\epsilon_i^2\right).$$
By Lemma \ref{lem_AwBwCons}, $\|{\bm a}_w\|_2^2$ is uniformly bounded in $w\in \mathcal W$, $P\in \mathcal P$ and $n\in \mathbb N$. Hence, \eqref{eq_WEpsO1Cluster} follows using Lemma \ref{lem_UnifOP1} and assertions \ref{ass_EpsilonCluster} and \ref{ass_EpsilonZCluster} of Assumption \ref{ass_MomentsCluster}.

For \eqref{eq_WEpsDiffCluster}, observe that
\begin{align*}
    &\frac{1}{n_0}\sum_{g=1}^G(\hat s_g(w)- s_g(w))^2 \\&\leq \frac{1}{n_0}\sum_{g=1}^Gn_g\sum_{i\in I_g}\left[(w({\bm z}_i)+\hat {\bm a}_w^T{\bm z}_i)\hat r_i-(w({\bm z}_i)+{\bm a}_w^T {\bm z}_i)\epsilon_i\right]^2\\
    &\leq 2 \underbrace{\frac{\bar n}{n_0}\sum_{g=1}^G\sum_{i\in I_g}\left((w({\bm z}_i) + \hat {\bm a}_w^T {\bm z}_i)(\hat r_i - \epsilon_i)\right)^2}_{(I)} + 2\underbrace{\frac{\bar n}{n_0}\sum_{g=1}^G\sum_{i\in I_g}(\hat {\bm a}_w^T {\bm z}_i - {\bm a}_w^T {\bm z}_i)^2 \epsilon_i^2}_{(II)}
\end{align*}

For $(I)$, observe that $\hat r_i - \epsilon_i = {\bm x}_i^T({\bm{\beta}}^*-\hat {\bm{\beta}})$. Hence, using similar arguments as before,
\begin{align}
    (I)&= \frac{\bar n}{n_0}\sum_{g=1}^G\sum_{i\in I_g}\left((w({\bm z}_i) + \hat {\bm a}_w^T {\bm z}_i){\bm x}_i^T({\bm{\beta}}^*-\hat {\bm{\beta}})\right)^2\nonumber\\
    &\leq\|{\bm{\beta}}^*-\hat{\bm{\beta}}\|_2^2\frac{\bar n}{n_0}\sum_{g=1}^G\sum_{i\in I_g}\|(w({\bm z}_i) + \hat {\bm a}_w^T{\bm z}_i){\bm x}_i^T\|_2^2\nonumber\\
    &\leq\|{\bm{\beta}}^*-\hat {\bm{\beta}}\|_2^2\frac{2\bar n}{n_0}\sum_{g=1}^G\sum_{i\in I_g} \left(\|{\bm x}_i\|_2^2 + \|\hat {\bm a}_w\|_2^2\|{\bm z}_i\|_2^2\|{\bm x}_i\|_2^2\right).\nonumber\\
    &\leq \|{\bm{\beta}}^*-\hat{\bm{\beta}}\|_2^2\left(\frac{2\bar n}{n_0}\sum_{g=1}^G\sum_{i\in I_g}\|{\bm x}_i\|_2^2 +\|\hat {\bm a}_w\|_2^2\frac{2\bar n}{n_0}\sum_{g=1}^G\sum_{i\in I_g}\|{\bm z}_i\|_2^2\|{\bm x}_i\|_2^2\right)\label{eq_BoundTermI}
\end{align}
Now, $\|{\bm{\beta}}^*-\hat {\bm{\beta}}\|_2 = \oP(1)$ by Lemma \ref{lem_AwBwCons}. Moreover, $\frac{2\bar n}{n_0}\sum_{g=1}^G\sum_{i\in I_g}\|{\bm x}_i\|_2^2=\OP(1)$ and $$\frac{2\bar n}{n_0}\sum_{g=1}^G\sum_{i\in I_g}\|{\bm z}_i\|_2^2\|{\bm x}_i\|_2^2=\OP(1)$$
by Lemma \ref{lem_UnifOP1} and assertions \ref{ass_XCluster} and \ref{ass_XZCluster} of Assumption \ref{ass_MomentsCluster}. Finally,
$$\sup_{w\in \mathcal W} \|\hat {\bm a}_w\|_2 \leq \hatED[\|{\bm x}_i\|_2]\|\hat {\bm M}\|_{op} = \OP(1),$$
since $\hatED[\|{\bm x}_i\|_2]=\OP(1)$ by Lemma \ref{lem_UnifOP1} and assertion \ref{ass_XCluster} of Assumption \ref{ass_MomentsCluster} and $\|\hat {\bm M}\|_{op} = \|{\bm M}\|_{op}+\|\hat {\bm M} - {\bm M}\|_{op} = O(1) + \oP(1)$ by Lemma \ref{lem_AwBwCons}. We get that  $(I) = \oPW(1)$.

For $(II)$, observe that
$$(II)\leq \|\hat {\bm a}_w - {\bm a}_w\|_2^2\frac{1}{n_0}\sum_{g=1}^G\sum_{i\in I_g}\|{\bm z}_i\|_2^2\epsilon_i^2.$$
Hence, we can conclude that $(II)=\oPW(1)$ using
Lemma \ref{lem_AwBwCons}, Lemma \ref{lem_UnifOP1} and assertion \ref{ass_EpsilonZCluster} of Assumption \ref{ass_MomentsCluster}. 
This concludes the proof of \eqref{eq_WEpsDiffCluster}.

\subsection{Proof of Theorem 2 in the Main Text}\label{sec_ProofPValNull}
First, consider a fixed $w\in \mathcal W$. Define $q_{1-\alpha} = \Phi^{-1}(1-\alpha)$, that is, the $(1-\alpha)$ quantile of a standard normal distribution. Fix $\alpha_0\in (0,0.5)$. Let $p_{val}(w) = 1-\Phi\left(N(w)/\max(\sqrt \gamma, \hat\sigma_w)\right)$. Then, it holds that
$$\Prob_P\left(p_{val}(w)\leq \alpha\right)=\Prob_P\left(\frac{N(w)}{\max(\sqrt \gamma, \hat\sigma_w)}\geq q_{1-\alpha}\right).$$
Now, consider $K\geq 2$. Then,
\begin{align}
    \Prob_P\left(p_{val}(w)\leq \alpha\right)
    &=\Prob_P\left(\frac{N(w)}{\max(\sqrt \gamma, \hat\sigma_w)}\geq q_{1-\alpha},\,\sigma_w\geq \frac{\sqrt \gamma}{K}\right)\nonumber\\
    &\quad+ \Prob_P\left(\frac{N(w)}{\max(\sqrt \gamma, \hat\sigma_w)}\geq q_{1-\alpha},\, \sigma_w< \frac{\sqrt \gamma}{K}, \, \hat \sigma_w^2 \geq \gamma\right)\nonumber\\
    &\quad +\Prob_P\left(\frac{N(w)}{\max(\sqrt \gamma, \hat\sigma_w)}\geq q_{1-\alpha},\, \sigma_w< \frac{\sqrt \gamma}{K}, \, \hat \sigma_w^2 < \gamma\right). \label{eq_UnionK}
\end{align}
Since $\max(\sqrt \gamma, \hat\sigma_w)>\hat\sigma_w$ and applying Theorem 1 in the main text with $\zeta = \gamma/K^2$ (and noting that $\rho_w= 0$ and $\bm \tau = 0$ under the null hypothesis), we have for the first term of \eqref{eq_UnionK} that
\begin{align}
    &\sup_{P\in \mathcal P}\sup_{w\in \mathcal W}\sup_{\alpha\in (0, \alpha_0)}\left\{\Prob_P\left(\frac{N(w)}{\max(\sqrt \gamma, \hat\sigma_w)}\geq q_{1-\alpha},\,\sigma_w\geq \frac{\sqrt \gamma}{K}\right)-\alpha\right\}\nonumber\\
    &\leq \sup_{P\in \mathcal P}\sup_{w\in \VP(\gamma/K^2)}\sup_{\alpha\in (0, \alpha_0)}\left\{\Prob_P\left(\frac{N(w)}{\hat\sigma_w}\geq q_{1-\alpha}\right)-\alpha\right\}\nonumber\\
    &\to 0\text{, as } n\to\infty.\label{eq_K1}
\end{align}
For the second term in \eqref{eq_UnionK}, we use that $K\geq 2$ and the second part of Theorem 1 in the main text and get
\begin{align}
 &\sup_{P\in \mathcal P}\sup_{w\in \mathcal W}\sup_{\alpha\in (0, \alpha_0)} \Prob_P\left(\frac{N(w)}{\max(\sqrt \gamma, \hat\sigma_w)}\geq q_{1-\alpha},\, \sigma_w< \frac{\sqrt \gamma}{K}, \, \hat \sigma_w^2 \geq \gamma\right)\nonumber\\
 &\leq \sup_{P\in \mathcal P}\sup_{w\in \mathcal W}\Prob_P\left(|\sigma_w^2 - \hat\sigma_w^2|\geq 3\gamma/4\right)\nonumber\\
 &\to 0\text{, as } n\to\infty. \label{eq_K2}
\end{align}
For the third term in \eqref{eq_UnionK}, we use the restriction on $\hat\sigma_w$ and the union bound to get
\begin{align}
    &\Prob_P\left(\frac{N(w)}{\max(\sqrt \gamma, \hat\sigma_w)}\geq q_{1-\alpha},\, \sigma_w< \frac{\sqrt \gamma}{K}, \, \hat \sigma_w^2 < \gamma\right)\nonumber\\
    &\leq \Prob_P\left(\frac{N(w)}{\sqrt \gamma}\geq q_{1-\alpha},\, \sigma_w< \frac{\sqrt \gamma}{K}\right)\nonumber\\
    &\leq \Prob_P\left(\frac{N^*(w)}{\sqrt \gamma}\geq \frac{q_{1-\alpha}}{2},\, \sigma_w< \frac{\sqrt \gamma}{K}\right) + \Prob_P\left(\frac{N'(w)}{\sqrt \gamma}\geq \frac{q_{1-\alpha}}{2}\right)\label{eq_K3}
\end{align}
with $N^*(w) = \frac{1}{\sqrt{n_0}}\sum_{i\in \mathcal D}(w({\bm z}_i) + {\bm a}_w^T{\bm z}_i)\epsilon_i$ and $N'(w) = (\hat {\bm a}_w-{\bm a}_w)^T\frac{1}{\sqrt{n_0}}\sum_{i\in \mathcal D}{\bm z}_i\epsilon_i$. For the first term in \eqref{eq_K3}, note that $\E[N^*(w)^2] = \sigma_w^2$. Moreover, note that $q_{1-\alpha}\geq q_{1-\alpha_0}>0$. Hence, we can use Markov's inequality and obtain
\begin{align}
    &\sup_{P\in \mathcal P}\sup_{w\in \mathcal W}\sup_{\alpha\in (0, \alpha_0)}\Prob_P\left(\frac{N^*(w)}{\sqrt \gamma}\geq \frac{q_{1-\alpha}}{2},\, \sigma_w< \frac{\sqrt \gamma}{K}\right)\nonumber\\
    &\leq \frac{4 \sigma_w^2}{\gamma q_{1-\alpha_0}^2}1\{\sigma_w^2\leq \gamma/K^2\}\nonumber\\
    &\leq\frac{4}{K^2q_{1-\alpha_0}^2}.\label{eq_K3A}
\end{align}
For the second term in \eqref{eq_K3}, we use that it holds by \eqref{eq_RestTerm} that
\begin{equation}
    \lim_{n\to\infty}\sup_{P\in \mathcal P}\sup_{w\in \mathcal W}\sup_{\alpha\in (0, \alpha_0)}\Prob_P\left(\frac{N'(w)}{\sqrt \gamma}\geq \frac{q_{1-\alpha}}{2}\right) = 0.\label{eq_K3B}
\end{equation}
From \eqref{eq_UnionK}, \eqref{eq_K1}, \eqref{eq_K2}, \eqref{eq_K3}, \eqref{eq_K3A} and \eqref{eq_K3B}, it follows that
$$\limsup_{n\to\infty}\sup_{P\in \mathcal P}\sup_{w\in \mathcal W}\sup_{\alpha\in (0, \alpha_0)}\left\{\Prob_P(p_{val}(w)\leq \alpha) - \alpha)\right\}\leq \frac{4}{K^2q_{1-\alpha_0}^2}.$$
Since this holds for all $K>2$, we have that 
$$\limsup_{n\to\infty}\sup_{P\in \mathcal P}\sup_{w\in \mathcal W}\sup_{\alpha\in (0, \alpha_0)}\left\{\Prob_P(p_{val}(w)\leq \alpha) - \alpha\right\}\leq 0.$$
Now, we consider $\hat w$. Note that the samples $\mathcal D_A$ and $\mathcal D$ are independent. Hence, $\Prob_P(p_{val}(\hat w)\leq \alpha) = \E_P[\Prob_P(p_{val}(\hat w)\leq \alpha|\mathcal D_A)] = \E_P[\Prob_P(p_{val}(w)\leq \alpha)\bigr\vert_{w = \hat w}]$. It follows that
\begin{align*}
    &\limsup_{n\to\infty}\sup_{P\in \mathcal P}\sup_{\alpha\in (0,\alpha_0)}\left\{\Prob_P(p_{val}(\hat w)\leq \alpha) - \alpha\right\}\\
    &\leq \limsup_{n\to\infty}\sup_{P\in \mathcal P}\sup_{w\in \mathcal W}\sup_{\alpha\in (0,\alpha_0)}\left\{\Prob_P(p_{val}(w)\leq \alpha)-\alpha\right\}\leq 0,
\end{align*}
which completes the proof.

\subsection{Proof of Proposition 1 in the Main Text}\label{sec_ProofPower}
Let $w\in \mathcal W$ and define $q_{1-\alpha}=\Phi^{-1}(1-\alpha).$ Then, it holds that
\begin{align}
    \left\{p_{val}(w)\geq \alpha\right\}&=\left\{\frac{N(w)}{\max(\sqrt \gamma, \hat\sigma_w)}\leq q_{1-\alpha}\right\}.
    \label{eq_Q1}
\end{align}
Now, define $N_1(w) = N(w) - \sqrt{n_0}\left(\rho_w + (\hat {\bm a}_w-{\bm a}_w)^T{\bm{\tau}}\right)$ (that is, the numerator of the quantity in eq. (19) in the main text). Then, it holds that
\begin{align}
    \left\{\frac{N(w)}{\max(\sqrt \gamma, \hat \sigma_w)}\leq q_{1-\alpha}\right\}&=\left\{\frac{N_1(w)+\sqrt{n_0}\left(\rho_w + (\hat {\bm a}_w-{\bm a}_w)^T{\bm{\tau}}\right)}{\max(\sqrt \gamma, \hat \sigma_w)}\leq q_{1-\alpha}\right\}\nonumber\\
    &\subseteq\left\{\frac{N_1(w)+ \sqrt{n_0}\xi_n}{\max(\sqrt \gamma, \hat \sigma_w)}\leq q_{1-\alpha}\right\}\cup\left\{\rho_w + (\hat {\bm a}_w-{\bm a}_w)^T{\bm{\tau}}< \xi_n\right\}\label{eq_Q2}
\end{align}
Now, define $K_n = \sqrt{\sqrt{n_0}\xi_n}$. Since $\xi_n\gg \frac{1}{\sqrt{n_0}}$ by assumption, $K_n\to\infty$ as $n\to\infty$. It follows that
\begin{align}
    \left\{\frac{N_1(w)+ \sqrt{n_0}\xi_n}{\max(\sqrt \gamma, \hat \sigma_w)}\leq q_{1-\alpha}\right\}&\subseteq\left\{\frac{-K_n + \sqrt{n_0}\xi_n}{\max(\sqrt \gamma, \hat \sigma_w)} \leq q_{1-\alpha}\right\}\cup\left\{-K_n> N_1(w)\right\}.\label{eq_Q3}
\end{align}
Using similar arguments as before, there exists $\bar \sigma\in \mathbb R$ such that $\sigma_w\leq \bar \sigma$ for all $w\in \mathcal W$ and $P\in \{P_k\}_{k\in\mathbb N}$.
By Theorem 1 in the main text, $\lim_{n\to\infty}\sup_{P\in \{P_k\}_{k\in\mathbb N}}\sup_{w\in\mathcal W}\Prob_P(\hat\sigma_w> 2 \bar \sigma) = 0.$
Hence, for the first set in \eqref{eq_Q3}, since $-K_n + \sqrt{n_0}\xi_n\to\infty$,
\begin{equation}
    \lim_{n\to\infty}\sup_{P\in \{P_k\}_{k\in\mathbb N}}\sup_{w\in\mathcal W}\Prob_P\left(\frac{-K_n + \sqrt{n_0}\xi_n}{\max(\sqrt \gamma, \hat \sigma_w)} \leq q_{1-\alpha}\right)=0.\label{eq_Q3A}
\end{equation}

For the second set in \eqref{eq_Q3},
\begin{align}
    \left\{-K_n>N_1(w)\right\}\subseteq \left\{w\notin\mathcal V_{P_n}(\kappa)\right\}\cup \left\{-K_n>N_1(w), \, w\in \mathcal V_{P_n}(\kappa)\right\}\label{eq_Q4}
\end{align}
From \eqref{eq_Q1}, \eqref{eq_Q2}, \eqref{eq_Q3}, and \eqref{eq_Q4} and the fact that $\hat w_n$ was estimated on the auxiliary sample $\mathcal D_A$ independent of $\mathcal D$, it follows that,
\begin{align*}
    \Prob_{P_n}\left(p_{val}(\hat w_n)\geq \alpha\right)&\leq \Prob_{P_n}\left(\rho_{\hat w_n} + (\hat {\bm a}_{\hat w_n}-{\bm a}_{\hat w_n})^T {\bm{\tau}}<\xi_n\right) + \Prob_{P_n}\left(\hat w_n\notin \mathcal V_{P_n}(\kappa)\right)\\
    &\quad +\Prob_{P_n}\left(\frac{-K_n + \sqrt{n_0}\xi_n}{\max(\sqrt \gamma, \hat \sigma_w)} \leq q_{1-\alpha}\right)\\
    &\quad+ \sup_{P\in \{P_k\}_{k\in \mathbb N}}\sup_{w\in \mathcal V_{P}(\kappa)}\Prob_P\left(-K_n> N_1(w)\right).
\end{align*}
We can use Assumption 4 in the main text and \eqref{eq_Q3A}, and it remains to show that
\begin{equation}
    \lim_{n\to\infty}\sup_{P\in \{P_k\}_{k\in \mathbb N}}\sup_{w\in \mathcal V_{P}(\kappa)}\Prob_P\left(-K_n> N_1(w)\right)=0.\label{eq_Q5}
\end{equation}
For this, note that for $w\in \mathcal V_P(\kappa)$, it holds that $\sigma_w^2\geq \kappa$ and 
$$
    \Prob_P\left(-K_n>N_1(w)\right)=\Prob_P\left(\frac{-K_n}{\sigma_w}\geq \frac{N_1(w)}{\sigma_w}\right)\leq\Prob_P\left(\frac{-K_n}{\sqrt \kappa}\geq \frac{N_1(w)}{\sigma_w}\right).
$$
From Theorem 1 in the main text with $\zeta = \kappa$ and $\mathcal P = \{P_k\}_{k\in \mathbb N}$, we know that $\frac{N_1(w)}{\sigma_w}\to\mathcal N(0,1)$ uniformly in $P\in \{P_k\}_{k\in \mathbb N}$ and $w\in \mathcal V_{P}(\kappa)$. Since $-K_n/\sqrt \kappa\to -\infty$, this proves \eqref{eq_Q5} and concludes the proof.

\subsection{Proof of Theorem \ref{thm_GenAsympWeak}}
In the following, we omit the dependence of $\bm\beta = \bm \beta(P)$ and $\bm \theta = \bm \theta(P)$ on $P\in \mathcal P$.
By Assumption \ref{ass_BetaP} (and the subsequent definition of $\nu$) and definitions (25) and (26) in the main text, we can write
\begin{align}
    \tilde r_i(\bm \beta) &= r_i(\bm \beta)- \bm c_i^T\hat {\bm \lambda}_r(\bm \beta)\nonumber\\
    &=\nu_i + \bm c_i^T\bm \theta - \bm c_i^T\hatED[\bm c \bm c^T]^{-1}\hatED[\bm c(\nu + \bm c^T\bm \theta)]\nonumber\\
    &= \nu_i - \bm c_i^T\hatED[\bm c\bm c^T]^{-1} \hatED[\bm c\nu].\label{eq_TildeRi}
\end{align}
Note that from the definition (25) of $\hat{\bm \lambda}_w$ in the main text,
\begin{align*}
    \frac{1}{\sqrt{n_0}}\sum_{i\in \mathcal D}\tilde w_i\bm c_i^T\hatED[\bm c\bm c^T]
    &=\frac{1}{\sqrt{n_0}}\sum_{i\in \mathcal D}\left(w(\bm z_i, \bm c_i)- \bm c_i^T\hat {\bm\lambda}_w\right)\bm c_i^T\hatED[\bm c\bm c^T]^{-1}\\
    &= \sqrt{n_0}\left(\hatED[w(\bm z, \bm c)\bm c^T]-\hat{\bm \lambda}_w^T\hatED[\bm c\bm c^T]\right)\hatED[\bm c\bm c^T]^{-1}\\
    &=\sqrt{n_0}\left(\hatED[w(\bm z, \bm c)\bm c^T]\hatED[\bm c\bm c^T]^{-1}-\hat{\bm \lambda}_w^T\right)\\
    &= 0.
\end{align*}
Hence, \eqref{eq_TildeRi} implies that
\begin{align}
    N(w, \bm \beta) &= \frac{1}{\sqrt{n_0}}\sum_{i \in \mathcal D} \tilde w_i \tilde r_i(\bm \beta)\nonumber\\ 
    &= \frac{1}{\sqrt{n_0}}\sum_{i \in \mathcal D}\left( w(\bm z_i, \bm c_i) - \bm c_i^T\hat {\bm \lambda}_w\right)\nu_i\nonumber\\
    &=\frac{1}{\sqrt{n_0}}\sum_{i\in \mathcal D}\nu_i \bar w_i+({\bm \lambda}_w-\hat{\bm \lambda}_w)^T\frac{1}{\sqrt{n_0}}\sum_{i\in \mathcal D}\bm c_i\nu_i,\label{eq_DecompNWeak}
\end{align}
with $\bm \lambda_w$ and $\bar w_i$ defined in \eqref{eq_DefLambdaW}. Then, Theorem \ref{thm_GenAsympWeak} follows using Assumptions \ref{ass_GenMomentsWeak}, \ref{ass_GenNormWeak}, \ref{ass_GenULLNWeak}, and \ref{ass_GenVarWeak} analogously to the proof of Theorem \ref{thm_AsNormGen} in Section \ref{sec_ProofAsNormGen}.

\subsection{Proof of Proposition \ref{prop_AsNormWeakCluster}}
We apply Theorem \ref{thm_GenAsympWeak} and need to verify Assumptions \ref{ass_GenMomentsWeak}-\ref{ass_GenVarWeak}.
\subsubsection{Verification of Assumption \ref{ass_GenMomentsWeak}}
Assumption \ref{ass_GenMomentsWeak} follows from Assumption \ref{ass_MomentsWeakCluster}.
\subsubsection{Verification of Assumption \ref{ass_GenNormWeak}}
Define
$$l_g(w) = \frac{1}{\bar \sigma_w \sqrt{n_0}}\sum_{i \in I_g}\nu_i\bar w_i$$
and note that for $|w|\leq 1$,
\begin{align*}
    \EP\left[|\nu_i\bar w_i|^{2+\eta}\right]\leq 2^{1 + \eta} \left(\EP[|\nu_i|^{2 + \eta}] + \|\bm \lambda_w\|_2^{2 + \eta} \EP[|\nu_i|^{2 + \eta}\|\bm c_i\|_2^{2 +  \eta}]\right)
\end{align*}
can be bounded independently of $i$, $P$ and $w$ because $\bm \lambda_w = \EDP[\bm c\bm c^T]^{-1}\EDP[\bm cw(\bm z, \bm c)]$ and by Assumption \ref{ass_MomentsWeakCluster}. Then, Assumption \ref{ass_GenNormWeak} follows analogously to Section \ref{sec_VerifyAssCLTGen} with $l_g(w)$ replacing $r_g(w)$.
\subsubsection{Verification of Assumption \ref{ass_GenULLNWeak}}
This is analogous to Section \ref{sec_VeriAssLLNGen} using Assumption \ref{ass_MomentsWeakCluster}.
\subsubsection{Verification of Assumption \ref{ass_GenVarWeak}}
The proof is similar to Section \ref{sec_VeriAssSigmaHatGen}, and we only sketch the main steps.

From the definitions \eqref{eq_SigmaWeakCluster} and \eqref{eq_HatSigmaWeakCluster}, it follows that we need to show
\begin{align}
    \frac{1}{n_0}\sum_{g = 1}^G \tilde s_g(w, \bm \beta) = \oPW(1),\label{eq_SgWeak1}\\
    \frac{1}{n_0}\sum_{g=1}^G\left(\tilde s_g(w,\bm \beta)^2-\EP[\bar s_g(w)^2]\right) = \oPW(1),\label{eq_SgWeak2}
\end{align}
with $\bm \beta = \bm \beta(P)$. We start with \eqref{eq_SgWeak1}, which is equal to
$$\frac{1}{n_0}\sum_{g = 1}^G\sum_{i\in I_g}\nu_i\bar w_i + (\bm \lambda_w - \hat{\bm \lambda}_w)^T\frac{1}{n_0}\sum_{g=1}^G\sum_{i \in I_g}\bm c_i\nu_i$$
by \eqref{eq_DecompNWeak}.
From this, \eqref{eq_SgWeak1} follows similarly to Section \ref{sec_VeriAssSigmaHatGen} with Assumption \ref{ass_MomentsWeakCluster}.

For, \eqref{eq_SgWeak2}, one can proceed similarly to Section \ref{sec_VeriAssSigmaHatGen} (and using Assumption \ref{ass_MomentsWeakCluster}) to see that it is enough to show
\begin{equation}
    \frac{1}{n_0}\sum_{g = 1}^G\left(\tilde s_g(w, \bm \beta)- \bar s_g(w)\right)^2 = \oPW(1).
\end{equation}
Plugging in the definitions, and \eqref{eq_TildeRi}, writing $\hat {\bm \lambda}_\nu \coloneqq\hatED[\bm c\bm c^T]^{-1}\hatED[\bm c\nu]$, and using the Cauchy--Schwarz inequality, we obtain
\begin{align*}
    &= \frac{1}{n_0}\sum_{g = 1}^G\left[\sum_{i\in I_g}\left(\tilde w_i \tilde r_i(\bm \beta)-\bar w_i \nu_i\right)\right]^2\\
    &\leq \frac{\bar n}{n_0}\sum_{g = 1}^G\sum_{i\in I_g}\left[\left(w(\bm z_i, \bm c_i)-\bm c_i^T\hat{\bm \lambda}_w \right)\left(\nu_i - \bm c_i^T \hat{\bm \lambda}_\nu\right) -\left(w(\bm z_i, \bm c_i)-\bm c_i^T{\bm \lambda}_w\right)\nu_i\right]^2\\
    & = \frac{\bar n}{n_0}\sum_{g = 1}^G\sum_{i\in I_g}\left[\left(\bm \lambda_w- \hat {\bm \lambda}_w\right)^T\bm c_i\nu_i - \left(w(\bm z_i, \bm c_i) -\bm c_i^T\hat{\bm \lambda}_w\right)\left( \bm c_i^T\hat {\bm \lambda}_\nu\right)\right]^2\\
    &\leq 3\underbrace{\|\bm \lambda_w - \hat{\bm \lambda}_w\|_2^2\frac{\bar n}{n_0}\sum_{g = 1}^G\sum_{i\in I_g}\|\bm c_i\|_2^2 \nu_i^2}_{(I)} + 3\underbrace{\|\hat{\bm \lambda}_\nu\|_2^2\|\hat {\bm \lambda}_w\|_2^2\frac{\bar n}{n_0}\sum_{g = 1}^G\sum_{i\in I_g}\|\bm c_i\|_2^4}_{(II)} + 3\underbrace{\|\hat{\bm \lambda}_\nu\|_2^2\frac{\bar n}{n_0}\sum_{g = 1}^G\sum_{i\in I_g}\|\bm c_i\|_2^2}_{(III)}.
\end{align*}
To show $(I) = \oPW(1)$, we can proceed similarly to before by showing that $\|\bm \lambda_w - \hat{\bm \lambda}_w\|_2^2= \oPW(1)$ and using Assumption \ref{ass_MomentsWeakCluster}. To show $(II) = \oPW(1)$, it suffices to note that $\|\hat {\bm \lambda}_\nu\|_2 = \oPW(1)$, $\sup_{|w| \leq 1}\|\hat{\bm \lambda}_w\|_2 = \OP(1)$, and $\frac{1}{n_0}\sum_{g = 1}^G\sum_{i\in I_g}\|\bm c_i\|_2^4 = \OP(1)$, which follows similarly to before from the assumptions of Proposition \ref{prop_AsNormWeakCluster}. Similarly, $(III)=\oPW(1)$.

\subsection{Proof of Theorem 4 in the Main Text}
This is completely analogous to the proof of Theorem 2 in the main text in Section \ref{sec_ProofPValNull}.

\subsection{Proof of Lemma \ref{lem_WnWStar}}\label{sec_ProofWnWStar}
For two functions $w, w'\in \mathcal W$, we have $|\rho_w-\rho_{w'}|=|\E_P[\epsilon(w({\bm z})-w'({\bm z}))]|\leq \EP[\epsilon^2]^{1/2}\EP[(w({\bm z})-w'({\bm z}))^2]^{1/2}=\EP[\epsilon^2]^{1/2}\|w-w'\|_{L_2}$. Hence, it follows from the assumption of Lemma \ref{lem_WnWStar} that $\rho_{\hat w_n}-\rho_{w^*} = o_P(1)$.
For \eqref{eq_SigmaWEst}, note that
    $$|\sigma_w^2-\sigma_{w'}^2| = \EP[(w({\bm z})+{\bm a}_w^T{\bm z})^2\epsilon^2]-\EP[(w'({\bm z})+{\bm a}_{w'}^T{\bm z})^2\epsilon^2]+\rho_{w'}^2 - \rho_w^2.$$
    Note that for random variables $v$ and $w$, it holds that $|\EP[v^2]-\EP[w^2]| = |\EP[(v+w)(v-w)]|\leq\EP[(v+w)^2]^{1/2}\EP[(v-w)^2]^{1/2}$. Hence, what remains to do is to obtain a bound for $\EP[(w({\bm z})+{\bm a}_w^T{\bm z} + w'({\bm z}) + {\bm a}_{w'}^T{\bm z})^2\epsilon^2]$ independent of $w$ and $w'$ and a bound for $\EP[(w({\bm z})-w'({\bm z})+{\bm a}_w^T {\bm z} - {\bm a}_{w'}^T{\bm z})^2\epsilon^2]$ in terms of $\|w-w'\|_{L_2}$.
    The first is a consequence of Lemma \ref{lem_BoundEp}. For the second, since $\EP[\epsilon^2|{\bm z}]\leq C$, it is enough to bound $\EP[(w({\bm z})-w'({\bm z}))^2]$ and $\EP[({\bm a}_w^T{\bm z}-{\bm a}_{w'}^T{\bm z})^2]$. The former is by definition equal to $\|w-w'\|_{L_2}^2$. The latter is equal to $({\bm a}_w-{\bm a}_{w'})^T\EP[{\bm z}{\bm z}^T]({\bm a}_w-{\bm a}_{w'})\leq \|\EP[{\bm z}{\bm z}^T]\|_{op}\|{\bm a}_w-{\bm a}_{w'}\|_2^2$ and 
    $$\|{\bm a}_w-{\bm a}_{w'}\|_2^2\leq\|\bm M\|_{op}^2\|\E_P[(w({\bm z})-w'({\bm z})){\bm x}]\|_2^2\leq \|\bm M\|_{op}^2\E[\|{\bm x}\|_2^2]\|w-w'\|_{L_2}^2.$$
    Hence, \eqref{eq_SigmaWEst} follows using Lemma \ref{lem_AwBwCons} and assertion 3 of Assumption 2 in the main text.

\subsection{Proof of Lemma 1 in the Main Text}
 If $\EP[\epsilon|{\bm z}]=0$, we can set ${\bm{\lambda}} = {\bm M}\EP[{\bm z}\EP[f({\bm x})+\eta|{\bm z}]]$. On the other hand, if $\EP[f({\bm x})+\eta|{\bm z}]=\EP[{\bm x}|{\bm z}]^T{\bm{\lambda}}$, then $\EP[\epsilon|{\bm z}]=0$, since ${\bm M}\EP[{\bm z}\EP[f({\bm x})+\eta|{\bm z}]] = {\bm M}\EP[{\bm z}{\bm x}^T]{\bm{\lambda}} = {\bm{\lambda}}.$

\subsection{Auxiliary Lemmas}
Here, we collect auxiliary Lemmas for the various proofs.

\begin{lemma}\label{lem_UnifOP1}
Let $(v_i)_{i\in \mathbb N}$ be a sequence of random variables with law determined by $P\in \mathcal P$. If $\sup_{P\in \mathcal P}\sup_{i\in  \mathbb N}\EP[|v_i|]<\infty$, then
$$\frac{1}{n}\sum_{i=1}^n |v_i|=\OP(1).$$
\end{lemma}
\begin{proof}
    Let $K>0$. By Markov's inequality,
    $$\sup_{P\in \mathcal P}\sup_{n \in  \mathbb N}\Prob_P\left(\frac{1}{n}\sum_{i=1}^n |v_i|\geq K\right)\leq \frac{1}{K}\sup_{P\in \mathcal P}\sup_{n\in \mathbb N}\frac{1}{n}\sum_{i=1}^n\EP[|v_i|]\leq \frac{\sup_{P\in \mathcal P}\sup_{i\in \mathbb N}\EP[|v_i|]}{K},$$
    which converges to $0$ as $K\to\infty$.
\end{proof}
The next lemma is a uniform version of the weak law of large numbers with non-identically distributed random variables.
\begin{lemma}\label{lem_ULLN}
Let $(v_i)_{i\in \mathbb N}$ be a sequence of independent (not necessarily identically distributed) random variables, vectors, or matrices with law determined by $P\in \mathcal P$. If there exists $\eta>0$ such that 
$\sup_{P\in \mathcal P}\sup_{i\in  \mathbb N}\EP[\|v_i\|^{1+\eta}]<\infty$, then
$$\frac{1}{n}\sum_{i=1}^n (v_i-\EP[v_i]) = \oP(1).$$
\end{lemma}
\begin{proof}
Without loss of generality, we consider univariate $v_i$ (otherwise, one can look at the individual components). Moreover, without loss of generality, let $\eta\in (0,1)$.
    Let $\delta >0$. Choose $P_n\in \mathcal P$ such that 
    \begin{equation}\label{eq_DefPn}
        \Prob_{P_n}\left(\left|\frac{1}{n}\sum_{i=1}^n (v_i-\E_{P_n}[v_i])\right|>\delta\right)\geq \sup_{P\in \mathcal P}\Prob_{P}\left(\left|\frac{1}{n}\sum_{i=1}^n (v_i-\EP[v_i])\right|>\delta\right) - \frac{1}{n}.
    \end{equation}

We apply Theorem 2.2.11. in \cite{DurrettProb} (weak law of large numbers for triangular arrays) with $X_{n, k}$ corresponding to $v_k$ under law $P_n$ and $b_n = n$. First, observe that
$$\sum_{i=1}^n \Prob_{P_n}(|v_i|>n)\leq \sum_{i=1}^n \E_{P_n}\left[\frac{|v_i|^{1+\eta}}{n^{1+\eta}}\right]\leq \frac{\sup_{P\in \mathcal P}\sup_{i\in \mathbb N}\EP[|v_i|^{1+\eta}]}{n^\eta},$$
which converges to $0$ as $n\to\infty$. Moreover, 
\begin{align*}
    \frac{1}{n^2}\sum_{i=1}^n\E_{P_n}\left[v_i^21\{|v_i|\leq n\}\right]&\leq \frac{1}{n^2}\sum_{i=1}^n\E_{P_n}\left[v_i^2\frac{n^{1-\eta}}{|v_i|^{1-\eta}}\right]\\
    &\leq \frac{1}{n^{1+\eta}}\sum_{i=1}^n\E_{P_n}\left[|v_i|^{1+\eta}\right]\\
    &\leq \frac{\sup_{P\in \mathcal P}\sup_{i\in \mathbb N}\EP[|v_i|^{1+\eta}]}{n^\eta},
\end{align*}
which also converges to $0$ as $n\to\infty$. Finally, 
\begin{align*}
    \frac{1}{n}\sum_{i=1}^n\left(\E_{P_n}[v_i]-\E_{P_n}\left[v_i1\{|v_i|\leq n\}\right]\right)&=\frac{1}{n}\sum_{i=1}^n\E_{P_n}\left[v_i1\{|v_i|> n\}\right]\\
    &\leq \frac{1}{n}\sum_{i=1}^n\frac{\E_{P_n}[|v_i|^{1+\eta}]}{n^\eta}\\
    &\leq \frac{\sup_{P\in \mathcal P}\sup_{i\in \mathbb N}\EP[|v_i|^{1+\eta}]}{n^\eta}
\end{align*}
converges to $0$ as $n\to\infty$. Applying Theorem 2.2.11. in \cite{DurrettProb}, it follows that
$$\lim_{n\to\infty}\Prob_{P_n}\left(\left|\sum_{i=1}^n (v_i-\E_{P_n}[v_i])/n\right|>\delta\right)=0.$$
From \eqref{eq_DefPn}, it follows that
$$\lim_{n\to\infty}\sup_{P\in \mathcal P}\Prob_{P}\left(\left|\frac{1}{n}\sum_{i=1}^n (v_i-\EP[v_i])\right|>\delta\right)=0.$$
Since $\delta>0$ was arbitrary, this concludes the proof.
\end{proof}

The following Lemma is a uniform version of Slutsky's theorem, see e.g., Lemma 20 in \cite{ShahPetersHardness} for a proof.
\begin{lemma}[Lemma 20 in \cite{ShahPetersHardness}]\label{lem_USlutsky}
Let $(v_n)_{n\in \mathbb N}$ and $(w_n)_{n\in \mathbb N}$ be two sequences of random variables with distributions determined by $P\in \mathcal P$. Assume that 
$$\lim_{n\to\infty}\sup_{P\in \mathcal P} \sup_{t\in \mathbb R}\left|\Prob_P\left(v_n\leq t\right)-\Phi(t)\right| = 0.$$
\begin{enumerate}
    \item If for all $\delta >0$, it holds that $\lim_{n\to\infty}\sup_{P\in \mathcal P}\Prob(|w_n|>\delta) = 0$, then
$$\lim_{n\to\infty}\sup_{P\in \mathcal P} \sup_{t\in \mathbb R}\left|\Prob_P\left(v_n+w_n\leq t\right)-\Phi(t)\right| = 0.$$
    \item If for all $\delta >0$, it holds that $\lim_{n\to\infty}\sup_{P\in \mathcal P}\Prob_P(|w_n-1|>\delta) = 0$, then
$$\lim_{n\to\infty}\sup_{P\in \mathcal P} \sup_{t\in \mathbb R}\left|\Prob_P\left(v_n/w_n\leq t\right)-\Phi(t)\right| = 0.$$
\end{enumerate}
\end{lemma}

\begin{lemma}\label{lem_ProdOrder}
Consider a sequence of random variables $w_n = \OP(1)$ and a sequence of random variables $(v_n)_{n\in \mathbb N}$ such that $v_n = \oP(1)$. Then, 
$$v_n w_n = \oP(1).$$
\end{lemma}
\begin{proof}
    Let $\epsilon, \delta>0$. Since $w_n = \OP(1)$, we can choose $M$ such that for all $n\in \mathbb N$ and $P\in \mathcal P$, $\Prob_P(|w_n|> M)<\delta/2$. Then, there exists $N\in \mathbb N$ such that for all $n\geq N$, $\sup_{P\in \mathcal P}\Prob_P(|v_n|>\epsilon/M)< \delta/2$. Hence, for all $n\geq N$ and all $P\in \mathcal P$, $\Prob_P(|v_n w_n|>\epsilon)<\delta$. Since $\epsilon,\delta>0$ were arbitrary, this concludes the proof.
\end{proof}

The following fact is used at several points in the proofs.
\begin{lemma}\label{lem_SumDiffSquares}
Let $a_1,\ldots, a_n$ and $b_1,\ldots, b_n\in \mathbb R$. Then,
    $$\frac{1}{n}\sum_{i=1}^n (a_i^2-b_i^2)\leq 2\sqrt{\frac{1}{n}\sum_{i=1}^nb_i^2}\sqrt{\frac{1}{n}\sum_{i=1}^n(a_i-b_i)^2} + \frac{1}{n}\sum_{i=1}^n (a_i-b_i)^2.$$
\end{lemma}
\begin{proof}
    This follows by writing $a_i^2 - b_i^2  = (a_i-b_i)^2 + 2(a_i-b_i) b_i$ and using the Cauchy--Schwarz inequality.
\end{proof}

\section{Additional Simulations}\label{sec_SuppSim}
In this section, we provide additional simulations. We first describe the detailed data-generating process and the implementation details. All the code is available on GitHub (\url{https://github.com/cyrillsch/RPIV_Application}).
\subsection{Data-Generating Process}\label{sec_DataGP}
Let $n$ be the sample size, let ${\bm z}\in \mathbb R^{n_{IV}}$ denote the instruments, and let ${\bm c}\in \mathbb R^{n_C}$ denote additional exogenous controls (possibly $n_C = 0$). We generate
$$
{\bm z}\sim \mathcal N({\bm 0},{\bm I}_{n_{IV}}),\qquad 
{\bm c}\sim \mathcal N({\bm 0},{\bm I}_{n_C}),
$$
independently, and then introduce correlation between instruments and controls by replacing the first $\min(n_{IV},n_C)$ components of $\bm z$ via
$$
z_k \leftarrow \frac{z_k + c_k}{\sqrt 2}, \, k=1,\ldots,\min(n_{IV},n_C)
$$
(where $z_k$ and $c_k$ are the $k$th components of $\bm z$ and $\bm c$, respectively).
Endogeneity is induced by a latent confounder $h\sim \mathcal N(0,1)$. Define
$$
\delta = -h + 0.3\,u_\delta,\quad \epsilon = h + 0.3\,u_\epsilon,
$$
where $u_\delta,u_\epsilon\sim\mathcal N(0,1)$ are independent of everything else. The endogenous regressor is generated as
$$
x = \pi \tanh\!\left(\frac{1}{\sqrt{n_{IV}}}\sum_{k=1}^{n_{IV}} z_k\right)
    + 0.3\,c_1\,1_{\{n_C>0\}}
    + \delta,
$$
where $\pi\ge 0$ controls instrument strength. For heteroskedastic errors, we scale
$$
\epsilon \leftarrow \epsilon\cdot |z_1|
$$
(this preserves $\E[\epsilon|\bm z]=0$), for homoskedastic errors, we leave $\epsilon$ unchanged.
Define the linear component
$$
\ell = -x + 0.5\sum_{j=1}^{n_C} c_j,
$$
(with the convention $\sum_{j=1}^{0}(\cdot)=0$), and the baseline outcome
$$
y = 2 + \ell + \epsilon.
$$
Under alternatives, we add a violation term $s_{viol}v$ with strength $s_{viol}\in \mathbb R$ and one of four violations
$$
v\in\left\{z_1^2,\; \sign(z_1),\; \ell^2,\; \sign(\ell)\right\},
$$
which we will refer to as \textit{``z squared''}, \textit{``sign(z)''}, \textit{``misspec. squared''}, and \textit{``misspec. sign''}, respectively. Thus,
$$
y \leftarrow y + s_{viol}\,v.
$$
Under $H_0$ we set $s_{viol}=0$. Note that the ``true'' coefficient on $x$ in the baseline linear component is $-1$.

\subsection{Implementation Details}\label{sec_ImplementationDetails}
We compute p-values for the following procedures:
\begin{description}
    \item[RP hom./RP het.:] The residual prediction test from Sections 2 and 3 in the main text with homoskedastic/heteroskedastic variance estimators (11) and (9) in the main text. Controls and intercept are included by adding ${\bm c}$ to both ${\bm x}$ and ${\bm z}$.
    \item[weak RP hom./weak RP het.:] The weak-identification robust extension from Section 4 in the main text with homoskedastic/heteroskedastic variance estimators given in \eqref{Def_WeakSigmaHatHomo} and in (28) in the main text. The p-value for overall specification is computed as $\max_{\beta_0} p_{val}(\beta_0)$ over a grid of candidate $\beta_0$ values (see below).
    \item[overid. J:] The classical overidentifying restrictions test if the number of instruments is larger than 1. If $n_{IV}=1$, we add $z^2$ as additional instrument, which corresponds to the approach by \citet{DieterleASimpleDiagnostic}.
    \item[smooth asymp./smooth boot.:] The \textit{smooth test} of \citet{DelgadoConsistentTestsOfConditionalMomentRestrictions} applied to the linear IV model with asymptotic and bootstrap p-values. The test uses a similar test statistic as our residual prediction test, but with kernel smoothing instead of machine learning, and does not use sample splitting. Instead, inference is based on a U-statistic. We treat the controls $\bm c$ by adding them to both $\bm x$ and $\bm z$.
    \item[ICM:] The integrated conditional moment (ICM) test of \citet{AntoineIdentificationRobustNonparametricInference}. The test is robust to weak identification and is also based on setting the parameter $\beta$ to some candidate $\beta_0$, leading to a p-value depending on $ \beta_0$. We compute the p-value for overall specification as $\max_{\beta_0} p_{val}(\beta_0)$ over a grid of candidate $\beta_0 $ values (see below). The test uses a test statistic based on the integrated conditional moment approach \citep{BierensConsistentModelSpecification}. We use a triangle density for the test statistic as in the simulation section of \citet{AntoineIdentificationRobustNonparametricInference}. The p-value is obtained by simulation using a conditional variance estimate based on kernel smoothing. Following the suggestion of the paper, we project out the controls $\bm c$ linearly before applying the test.
\end{description}
We are not aware of any R packages implementing the \textit{smooth} test and the \textit{ICM} test. Since there are no clear guidelines available regarding the kernel smoothing, we use a Gaussian kernel and the normal reference rule / Silverman's rule of thumb \citep{SilvermanDensityEstimation, WassermannAllOfNonparametricStatistics} to choose the bandwidth.

The residual prediction based methods are implemented in the R package \texttt{RPIV} with a random forest \cite{BreimanRandomForest} using the R package \texttt{ranger} \cite{RangerPackage}. The random forest is tuned via out-of-bag error and for \textit{weak RP hom./het.}, we tune at the two-stage least squares solution and only refit (but not retune) the random forest at the other candidate ${\beta_0}$ (see Remark 3 in the main text). We use $n_A=\min\!\left(\frac n2,\exp(1)\cdot \frac{n}{\log n}\right)$ observations as the auxiliary sample and calculate the test statistic on the $n_0=n-n_A$ observations.

For \textit{weak RP hom./het.} and ICM, we use an equally spaced grid of size $200$ between $-10$ and $10$ to calculate p-values (note that the endogenous regressor is univariate, so $\beta_0\in\mathbb R$ is scalar after partialling out).

\subsection{Simulations under $H_0$}\label{sec_SuppSimH0}
In this subsection, we assess size control under the null hypothesis, i.e., we set $s_{viol}=0$ so that
$$
\E[\epsilon|{\bm z}]=0
$$
holds by construction. We report rejection frequencies at level $\alpha=0.05$ over 1000 Monte Carlo repetitions and consider the following simulations:
\begin{description}
    \item[Varying sample size $n$:] We vary the sample size $n\in\{50,100,150,200,250,300,400,500\},$
consider $n_C\in\{0,2\}$ controls and $n_{IV}\in\{1,2\}$ instruments, and fix $\pi=1$. We run this under both homoskedastic and heteroskedastic errors.
    \item[Varying IV strength $\pi$:] We fix $n=300$, vary $\pi\in\{0,0.1,0.2,0.4,0.6,0.8,1\},$
and again consider $n_C\in\{0,2\}$, $n_{IV}\in\{1,2\}$, and both error designs.
    \item[Varying number of instruments $n_{IV}$:] We fix $n=300$ and $\pi=1$ and vary the number of instruments $n_{IV}\in\{5,10,15,20,25\},$
with $n_C\in\{0,2\}$ and both error designs.
\end{description}
\begin{figure}[t]
\centering
\includegraphics[width=0.95\textwidth]{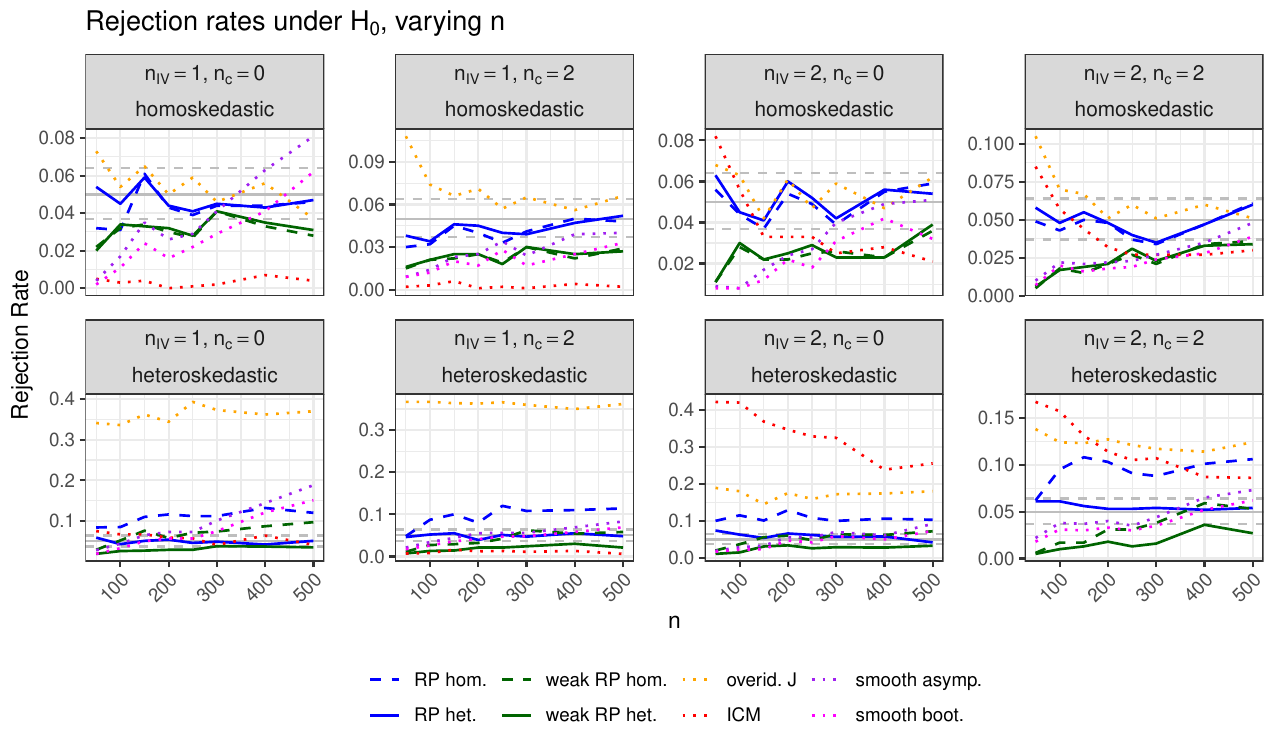}
\caption{Simulation results under $H_0$ with varying sample size $n$. The gray solid line marks $\alpha = 0.05$, the gray dashed lines show pointwise 95\% Monte Carlo bounds for rejection rates based on 1000 replications. Dashed curves coincide with the solid curves of the same color when not visible.}
\label{fig_SuppH0_vary_n}
\end{figure}
\begin{figure}[t]
\centering
\includegraphics[width=0.95\textwidth]{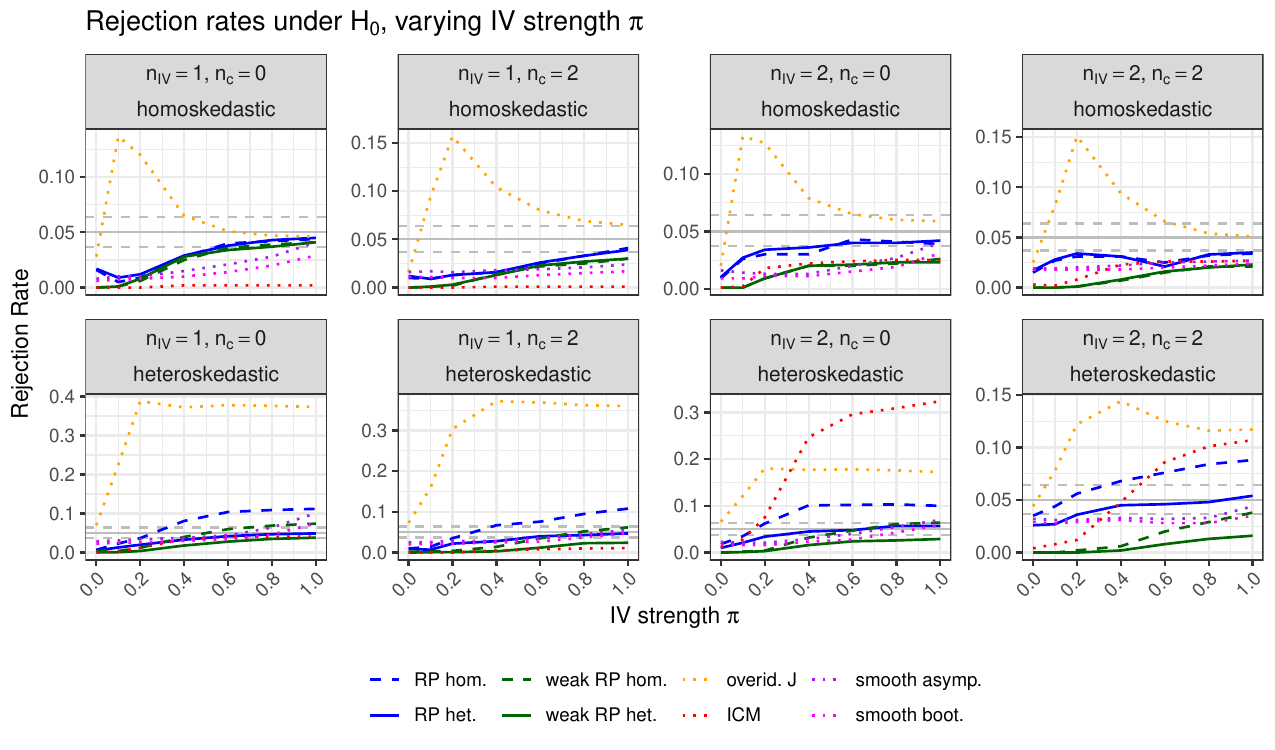}
\caption{Simulation results under $H_0$ with varying instrument strength $\pi$. The gray solid line marks $\alpha = 0.05$, the gray dashed lines show pointwise 95\% Monte Carlo bounds for rejection rates based on 1000 replications. Dashed curves coincide with the solid curves of the same color when not visible.}
\label{fig_SuppH0_vary_pi}
\end{figure}
\begin{figure}[t]
\centering
\includegraphics[width=0.9\textwidth]{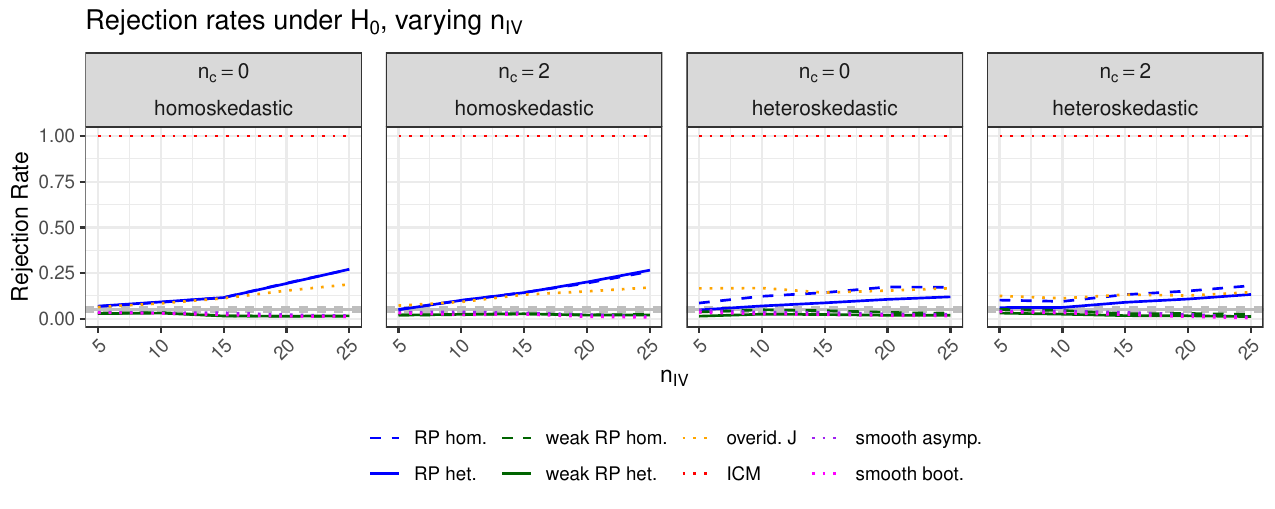}
\caption{Simulation results under $H_0$ with varying number of instruments $n_{IV}$. The gray solid line marks $\alpha = 0.05$, the gray dashed lines show pointwise 95\% Monte Carlo bounds for rejection rates based on 1000 replications. Dashed curves coincide with the solid curves of the same color when not visible.}
\label{fig_SuppH0_vary_niv}
\end{figure}

Figure \ref{fig_SuppH0_vary_n} shows rejection frequencies as a function of the sample size $n$. Across all configurations, the residual prediction based procedures control size well. Rejection rates are close to the nominal level $\alpha=0.05$ for \textit{RP hom./het.} and below the nominal level for \textit{weak RP hom./het.} and remain within the Monte Carlo confidence bands for moderate and large $n$. Regarding heteroskedasticity, the heteroskedasticity-robust variants consistently control size under both homoskedastic and heteroskedastic errors, whereas procedures relying on homoskedastic variance estimators may suffer from substantial overrejection in the heteroskedastic design.

Figure \ref{fig_SuppH0_vary_pi} reports rejection frequencies when varying the instrument strength $\pi$, including the extreme case $\pi=0$. The \textit{weak RP hom./het.} maintain size uniformly over $\pi$, but also \textit{RP hom./het.} does not overreject, even for small $\pi$, but appears to be overly conservative. Again, the \textit{RP hom.} variants may suffer from overrejection in the heteroskedastic design.

Figure \ref{fig_SuppH0_vary_niv} shows rejection frequencies as a function of the number of instruments $n_{IV}$. Here, we see that \textit{RP hom./het} overreject for large $n_{IV}$, whereas \textit{weak RP hom./het.} remain conservative even for large numbers of instruments.

Regarding the other methods, we see that \textit{overid. J} is neither robust to heteroskedastic errors, nor to weak instruments or many instruments. Moreover, we see that \textit{ICM} may suffer from overrejection in the heteroskedastic setting for some data-generating mechanisms and does not control size at all for $n_{IV}\geq 5$ (rejection rate 1 in all panels of Figure \ref{fig_SuppH0_vary_niv}). Also \textit{smooth asymp./boot.} may suffer from overrejection in some heteroskedastic settings, but the effect is less severe than for \textit{ICM}.

\subsection{Simulations under $H_A$}\label{sec_SuppSimHA}
In this subsection, we assess power under violations of well-specification.
We generate data as described in Section \ref{sec_DataGP} but set $s_{viol}\neq 0$, so that the model becomes misspecified due to the additional violation term $s_{viol}v$ with $v\in\{z_1^2,\sign(z_1),\ell^2,\sign(\ell)\}.$
We report rejection frequencies at level $\alpha=0.05$ over 1000 Monte Carlo repetitions and consider the following simulations:
\begin{description}
    \item[Varying violation strength $s_{viol}$:] 
    We fix $n=300$ and $\pi=1$, consider $n_C\in\{0,2\}$, $n_{IV}\in\{1,2\}$, and both homoskedastic and heteroskedastic errors. For each violation type, we vary $s_{viol}$ over a symmetric range. Concretely, we use $s_{viol}\in[-1,1]$ for \textit{z squared}, $s_{viol}\in[-3,3]$ for \textit{sign(z)}, $s_{viol}\in[-4,4]$ for \textit{misspec. squared}, and $s_{viol}\in[-8,8]$ for \textit{misspec. sign}.
    \item[Varying IV strength $\pi$:] 
    We fix $n=300$, vary $\pi\in\{0,0.1,0.2,0.4,0.6,0.8,1\}$, and again consider $n_C\in\{0,2\}$, $n_{IV}\in\{1,2\}$, and both error designs. We use fixed violation levels $s_{viol}\in\{1,3,4,8\}$ for \textit{z squared}, \textit{sign(z)}, \textit{misspec. squared}, and \textit{misspec. sign}, respectively.
    
    \item[Varying number of instruments $n_{IV}$] 
    We fix $n=300$, $\pi=1$, and the same moderate violation levels as above, and vary $n_{IV}\in\{5,10,15,20,25\}$, with $n_C\in\{0,2\}$ and both error designs.
\end{description}

Figures \ref{fig_SuppHA_gamma_Z2}--\ref{fig_SuppHA_gamma_msign} report rejection frequencies as a function of $s_{viol}$ for the four violation types. Figures \ref{fig_SuppHA_pi_Z2}--\ref{fig_SuppHA_pi_msign} report rejection frequencies as a function of $\pi$ for fixed (moderate) violation levels. Figures \ref{fig_SuppHA_niv_Z2}--\ref{fig_SuppHA_niv_msign} report rejection frequencies as a function of $n_{IV}$ for fixed (moderate) violation levels.
\begin{figure}[tbp]
\centering
\includegraphics[width=0.95\textwidth]{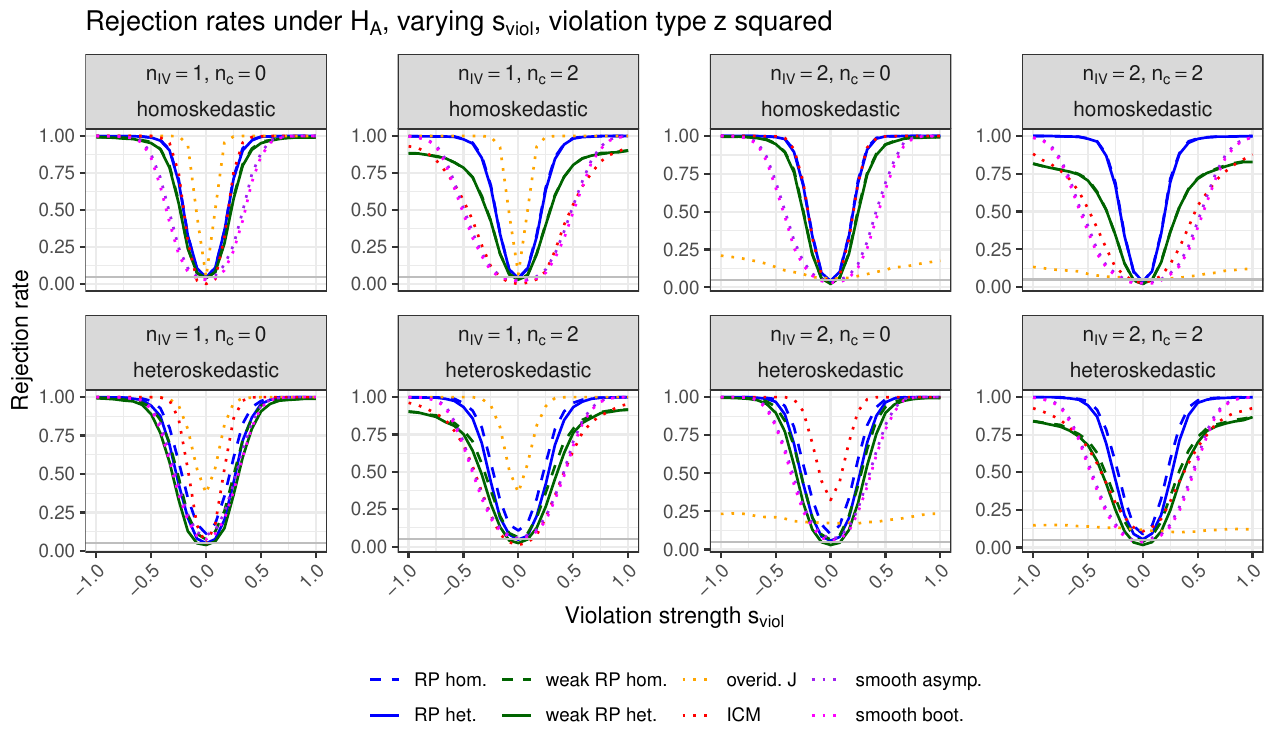}
\caption{Simulation results under $H_A$ with varying violation strength $s_{viol}$ for violation type \textit{z squared} ($v=z_1^2$). The gray solid line marks $\alpha = 0.05$. Panels correspond to combinations of $n_{IV}\in\{1,2\}$, $n_C\in\{0,2\}$, and homoskedastic vs.\ heteroskedastic errors. Dashed curves coincide with the solid curves of the same color when not visible.}
\label{fig_SuppHA_gamma_Z2}
\end{figure}
\begin{figure}[tbp]
\centering
\includegraphics[width=0.95\textwidth]{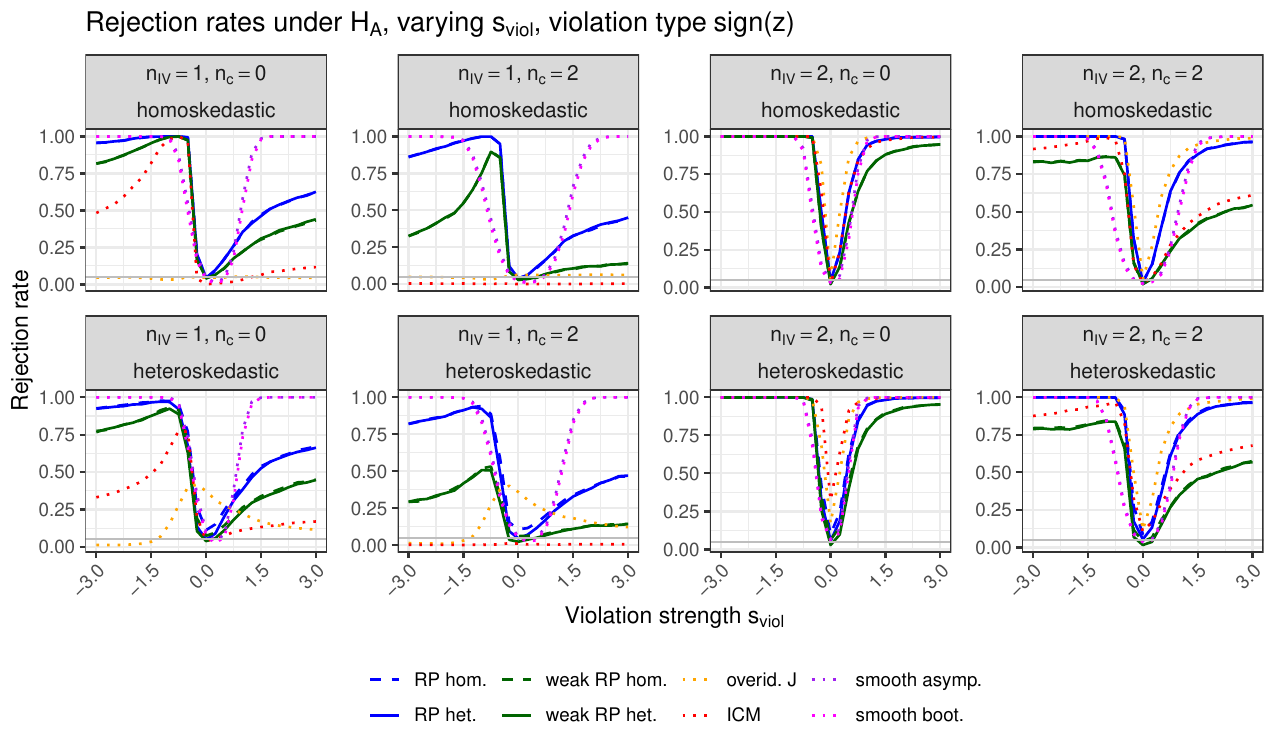}
\caption{Same as Figure \ref{fig_SuppHA_gamma_Z2} but with violation type \textit{sign(z)} ($v=\sign(z_1)$).}
\label{fig_SuppHA_gamma_sZ}
\end{figure}
\begin{figure}[tbp]
\centering
\includegraphics[width=0.95\textwidth]{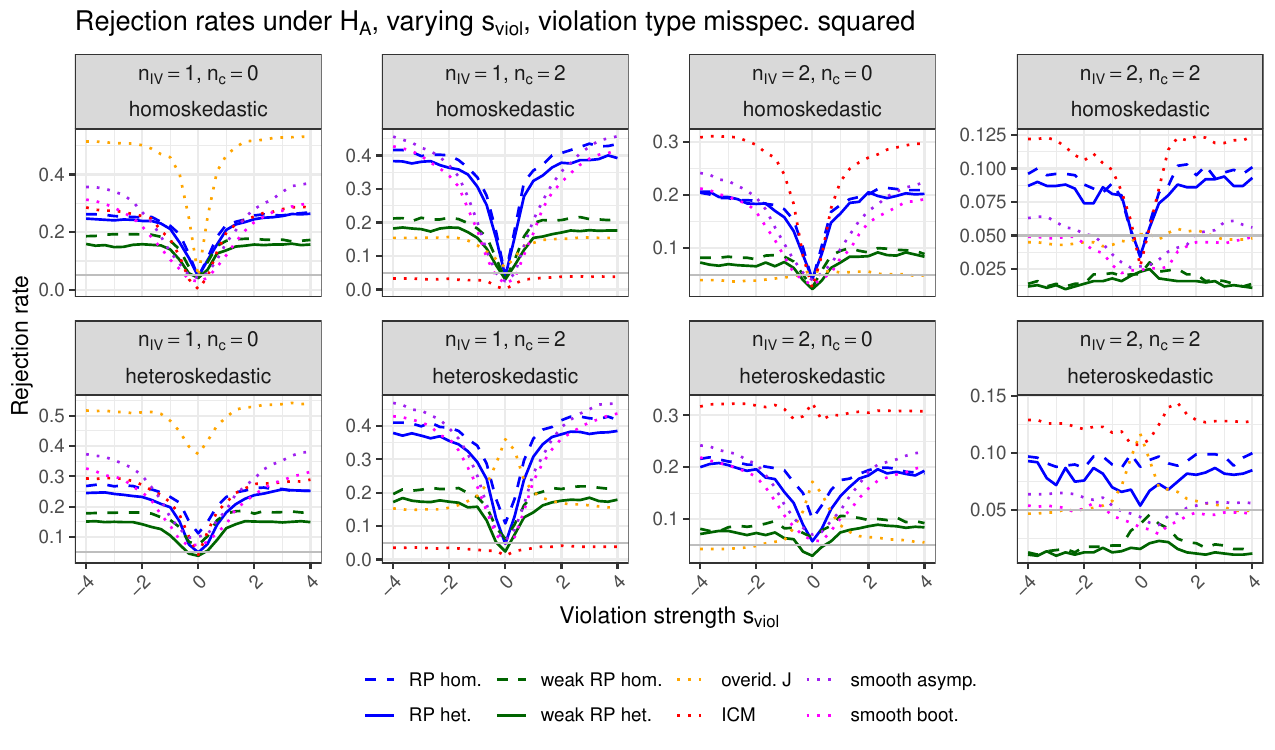}
\caption{Same as Figure \ref{fig_SuppHA_gamma_Z2} but with violation type \textit{misspec. squared} ($v=\ell^2$).}
\label{fig_SuppHA_gamma_msq}
\end{figure}
\begin{figure}[tbp]
\centering
\includegraphics[width=0.95\textwidth]{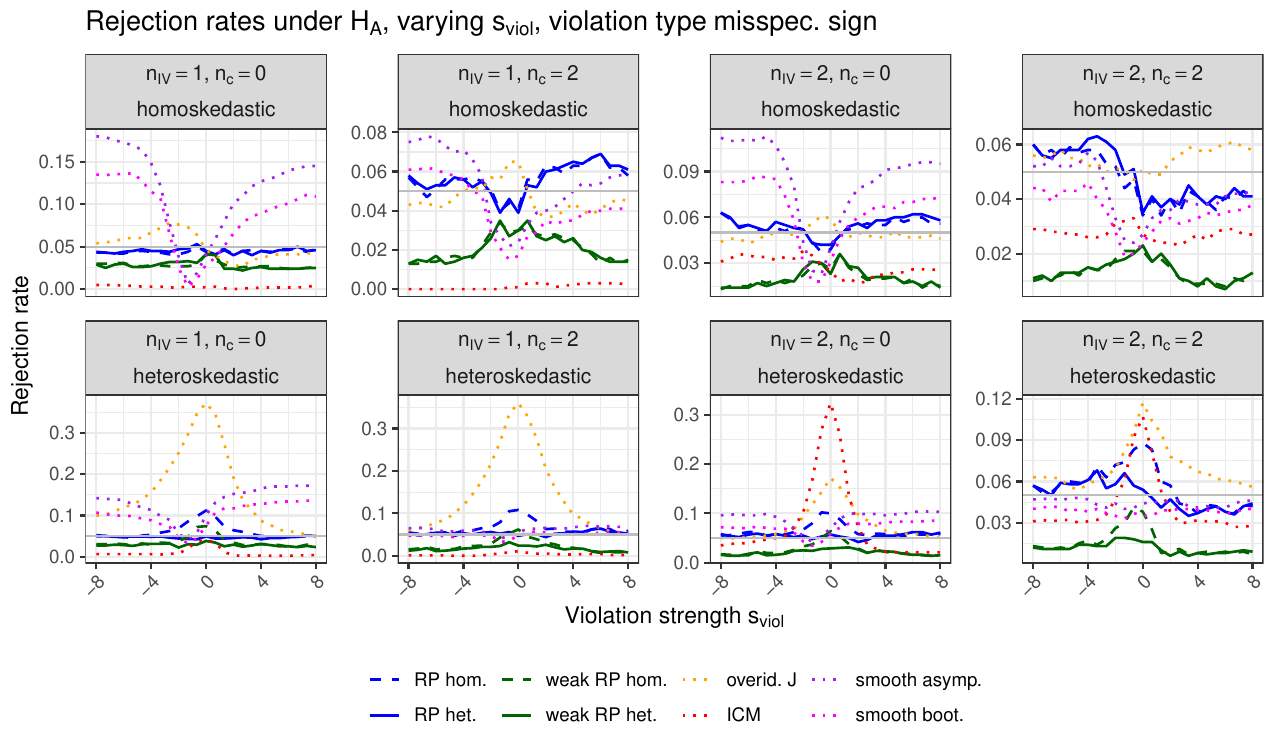}
\caption{Same as Figure \ref{fig_SuppHA_gamma_Z2} but with violation type \textit{misspec. sign} ($v=\sign(\ell)$).}
\label{fig_SuppHA_gamma_msign}
\end{figure}

\begin{figure}[tbp]
\centering
\includegraphics[width=0.95\textwidth]{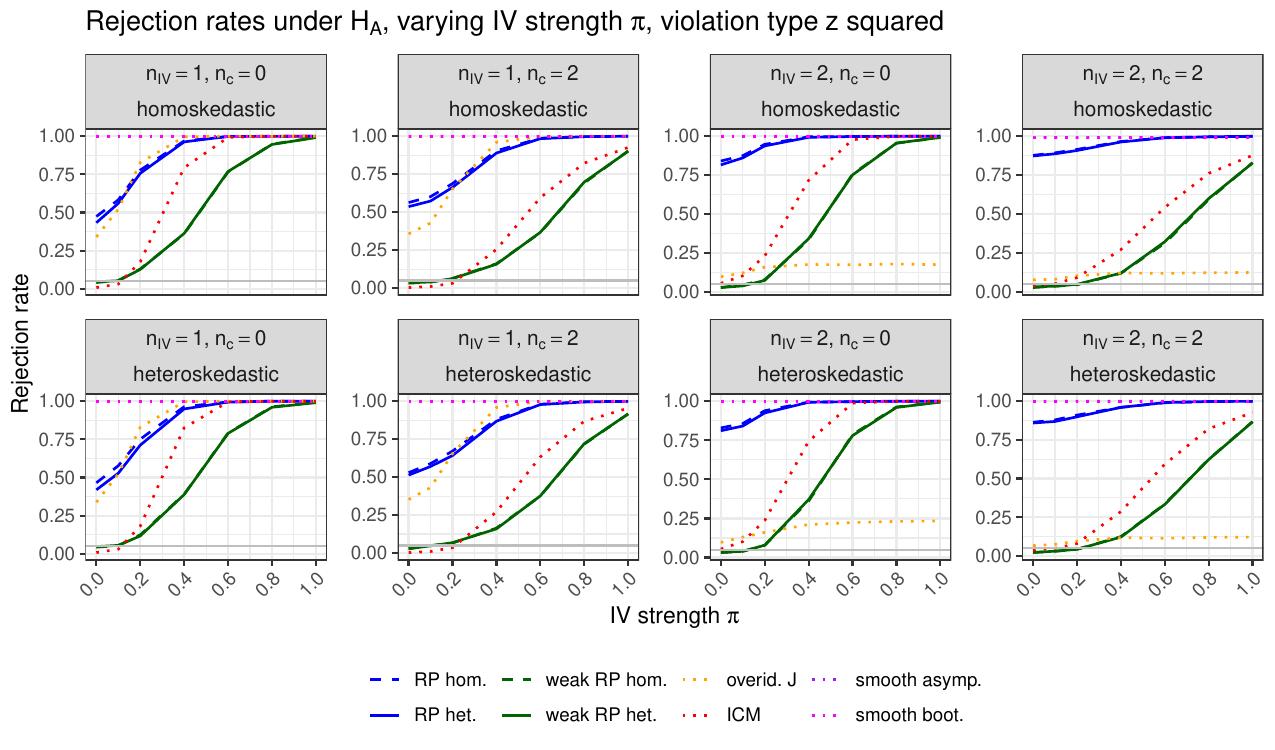}
\caption{Simulation results under $H_A$ with varying IV strength $\pi$ for violation type \textit{z squared} ($v=z_1^2$). The gray solid line marks $\alpha = 0.05$. Panels correspond to combinations of $n_{IV}\in\{1,2\}$, $n_C\in\{0,2\}$, and homoskedastic vs.\ heteroskedastic errors. Dashed curves coincide with the solid curves of the same color when not visible.}
\label{fig_SuppHA_pi_Z2}
\end{figure}
\begin{figure}[tbp]
\centering
\includegraphics[width=0.95\textwidth]{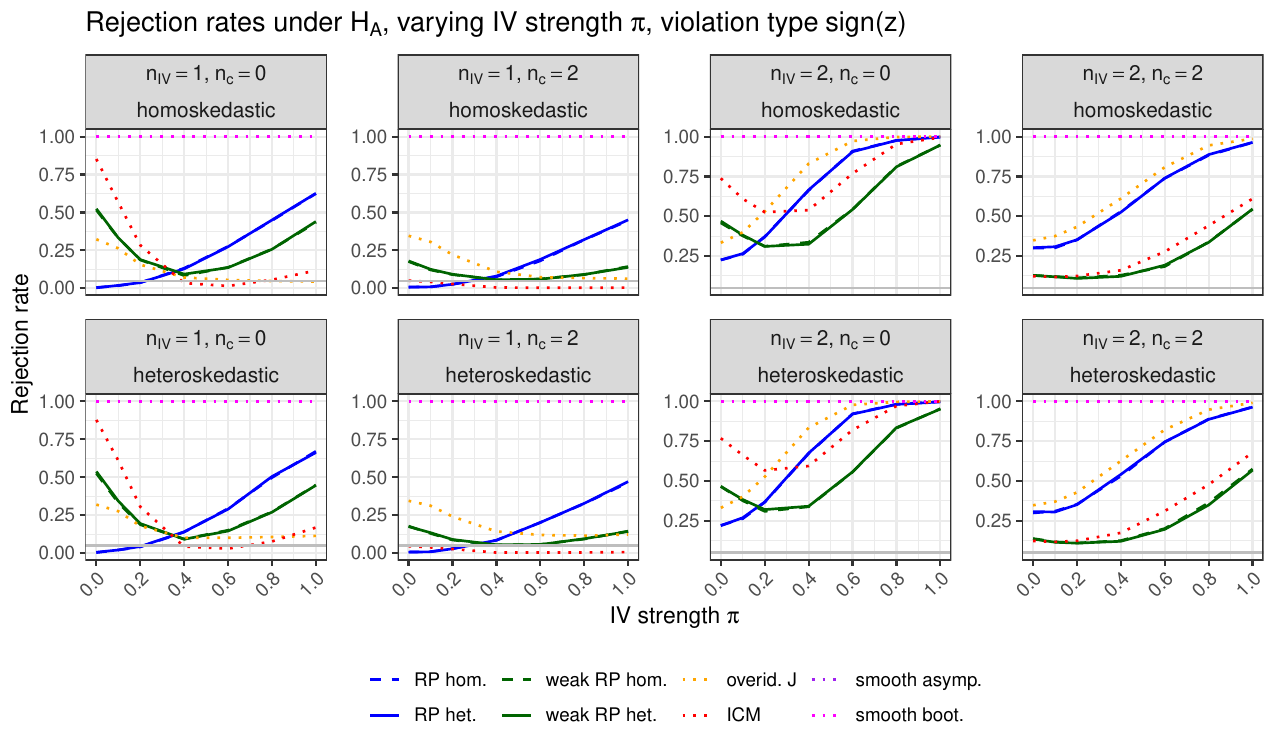}
\caption{Same as Figure \ref{fig_SuppHA_pi_Z2} but with violation type \textit{sign(z)} ($v=\sign(z_1)$).}
\label{fig_SuppHA_pi_sZ}
\end{figure}
\begin{figure}[tbp]
\centering
\includegraphics[width=0.95\textwidth]{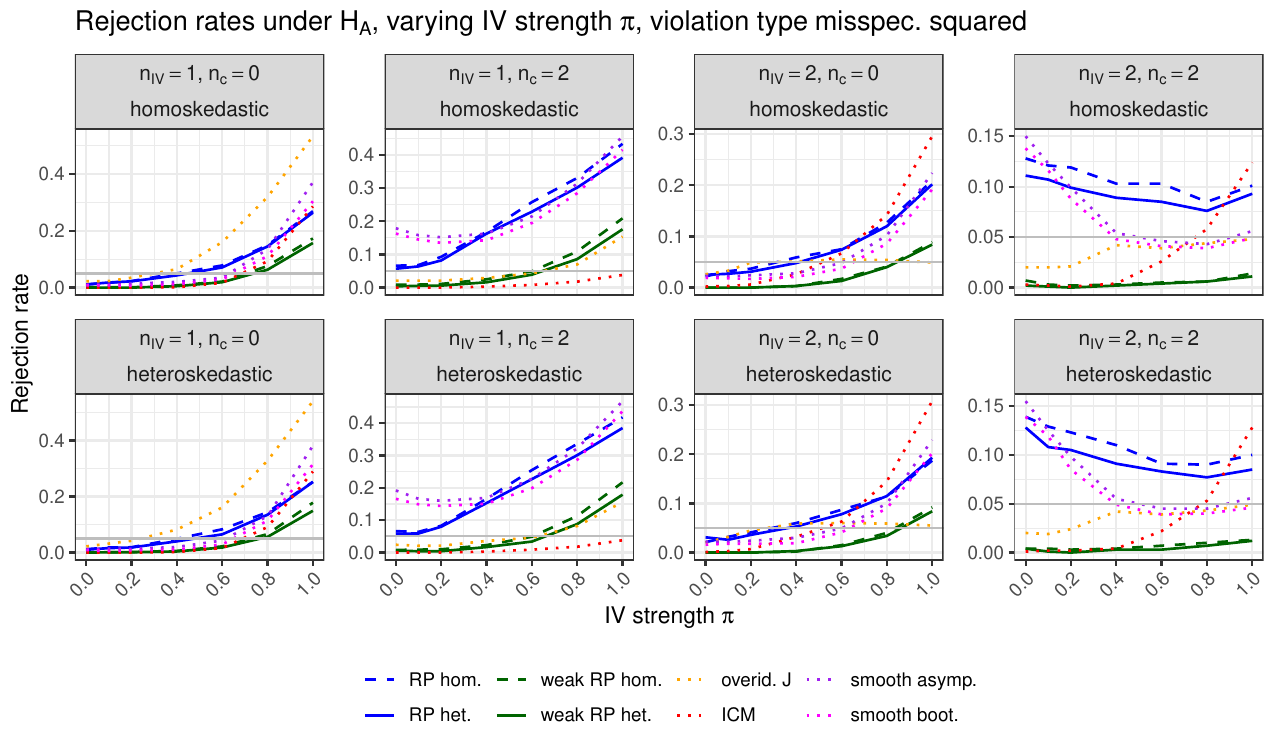}
\caption{Same as Figure \ref{fig_SuppHA_pi_Z2} but with violation type \textit{misspec. squared} ($v=\ell^2$).} 
\label{fig_SuppHA_pi_msq}
\end{figure}
\begin{figure}[tbp]
\centering
\includegraphics[width=0.95\textwidth]{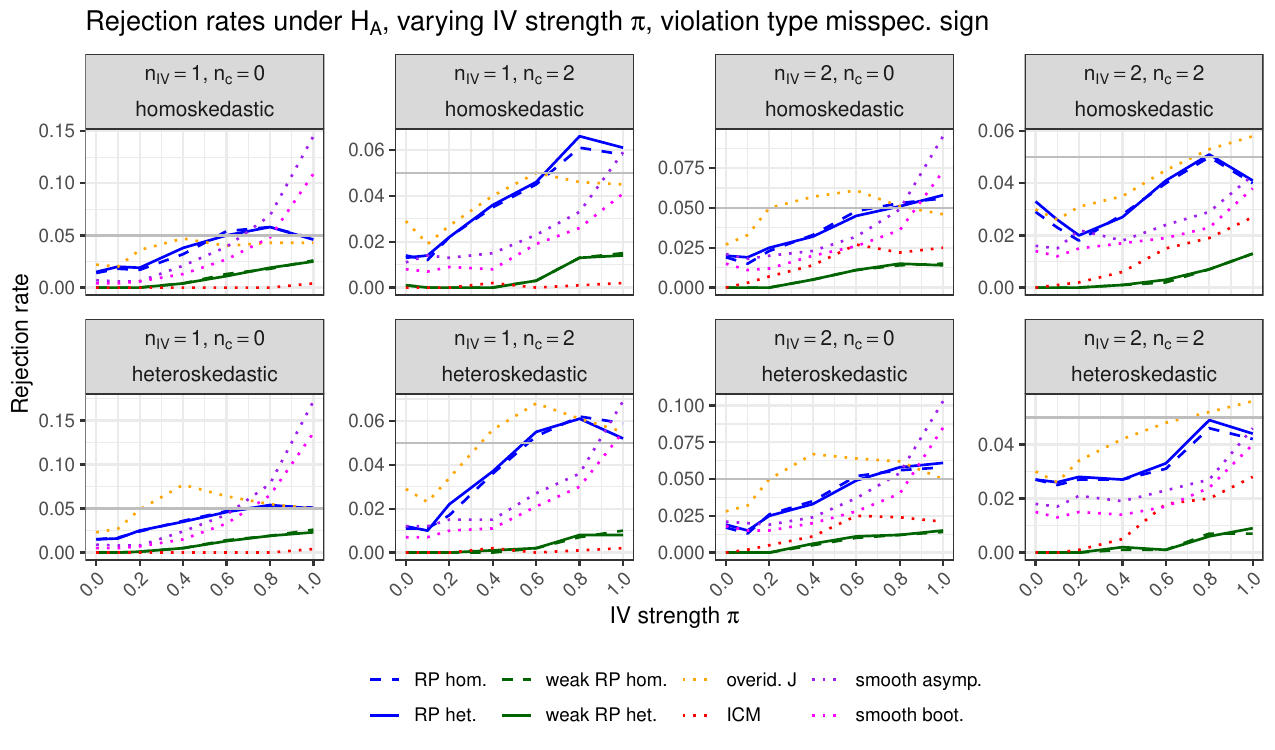}
\caption{Same as Figure \ref{fig_SuppHA_pi_Z2} but with violation type \textit{misspec. sign} ($v=\sign(\ell)$).}
\label{fig_SuppHA_pi_msign}
\end{figure}

\begin{figure}[tbp]
\centering
\includegraphics[width=0.9\textwidth]{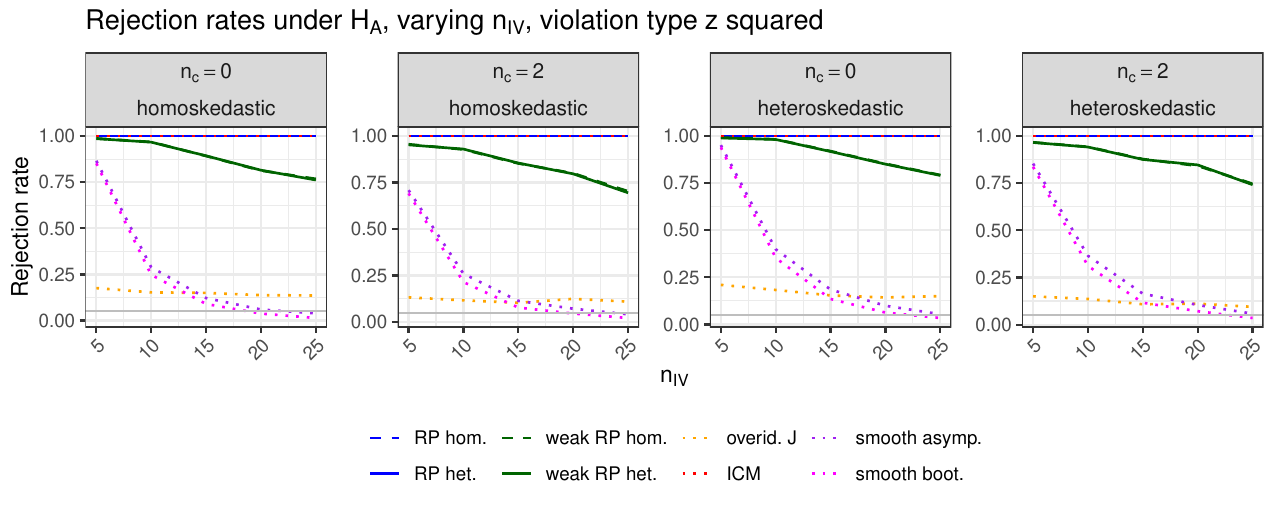}
\caption{Simulation results under $H_A$ with varying number of instruments $n_{IV}$ for violation type \textit{z squared} ($v=z_1^2$). The gray solid line marks $\alpha = 0.05$. Panels correspond to combinations of $n_C\in\{0,2\}$ and homoskedastic vs.\ heteroskedastic errors. Dashed curves coincide with the solid curves of the same color when not visible.}
\label{fig_SuppHA_niv_Z2}
\end{figure}
\begin{figure}[tbp]
\centering
\includegraphics[width=0.9\textwidth]{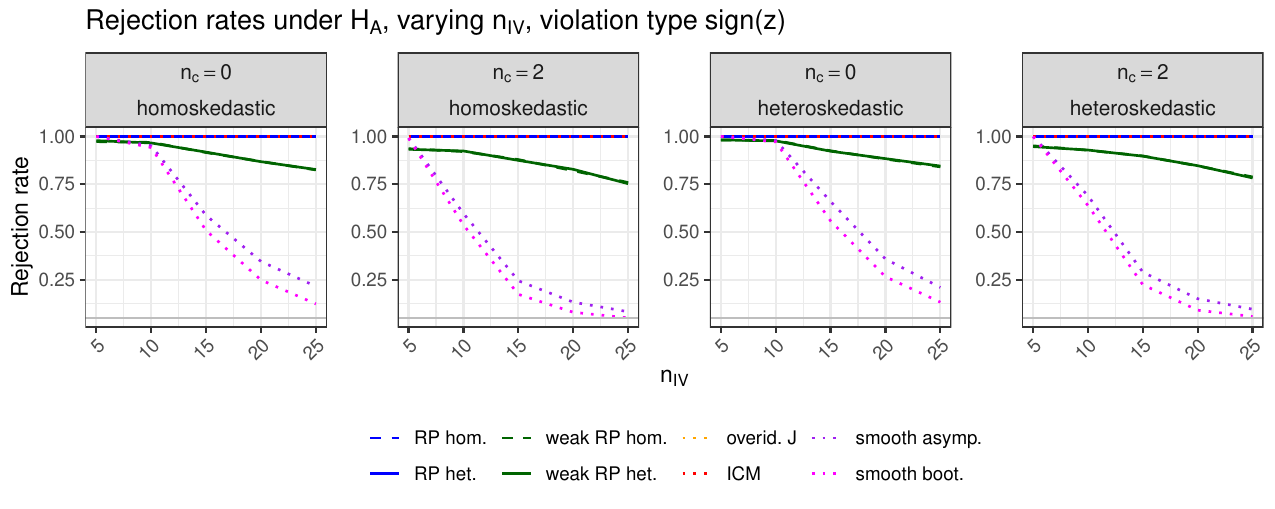}
\caption{Same as Figure \ref{fig_SuppHA_niv_Z2} but with violation type \textit{sign(z)} ($v=\sign(z_1)$).}
\label{fig_SuppHA_niv_sZ}
\end{figure}
\begin{figure}[tbp]
\centering
\includegraphics[width=0.9\textwidth]{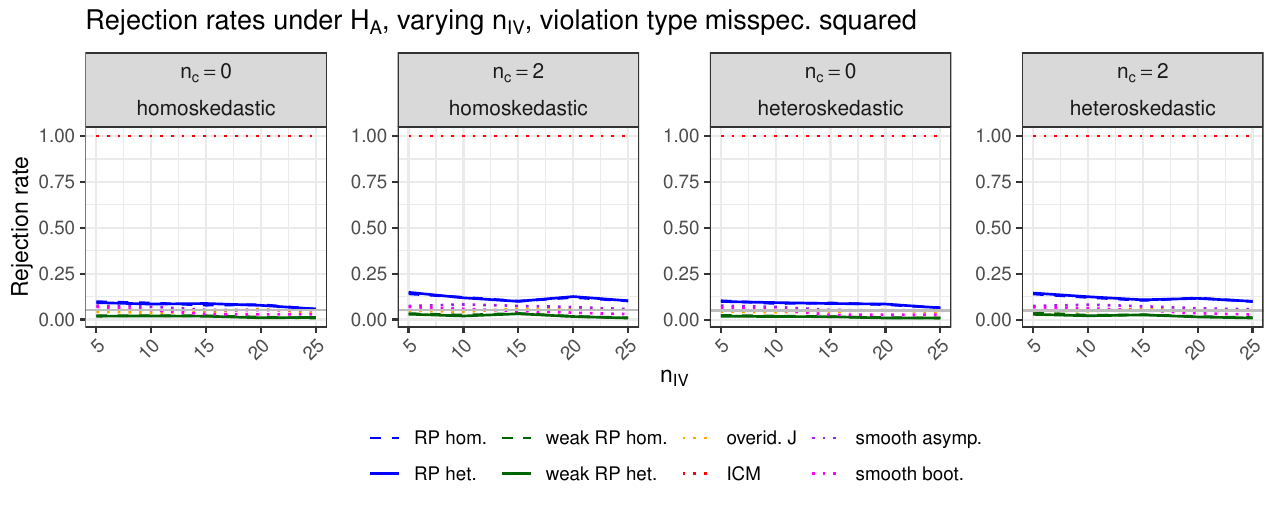}
\caption{Same as Figure \ref{fig_SuppHA_niv_Z2} but with violation type \textit{misspec. squared} ($v=\ell^2$).}
\label{fig_SuppHA_niv_msq}
\end{figure}
\begin{figure}[tbp]
\centering
\includegraphics[width=0.9\textwidth]{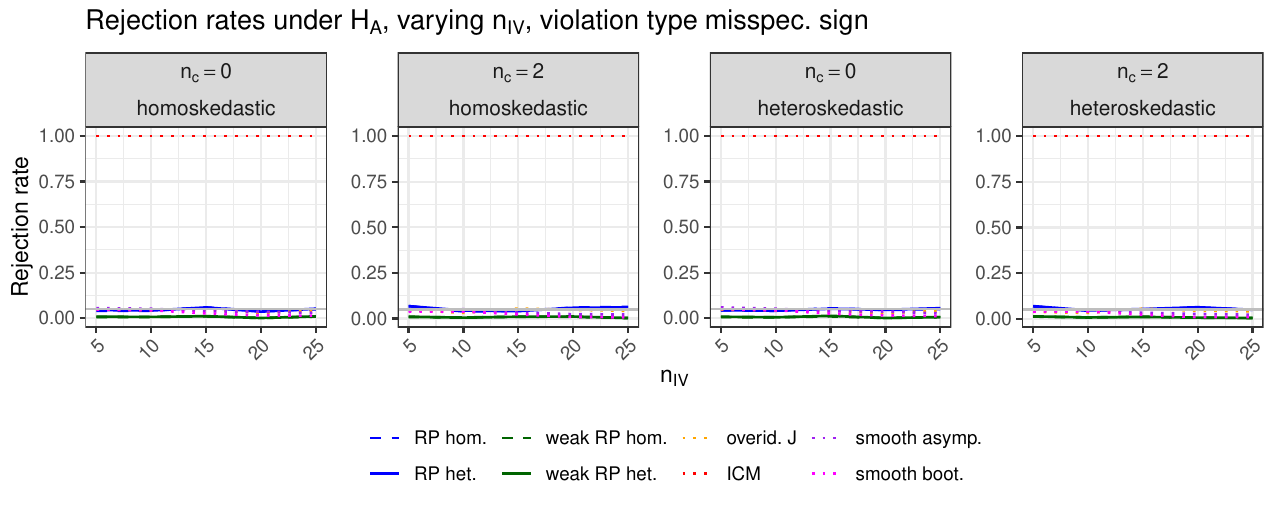}
\caption{Same as Figure \ref{fig_SuppHA_niv_Z2} but with violation type \textit{misspec. sign} ($v=\sign(\ell)$).}
\label{fig_SuppHA_niv_msign}
\end{figure}

In general, there is no clear ranking of the various methods in terms of power, but the residual prediction based methods are competitive against the other methods. Moreover, we see that in the settings considered, \textit{RP hom./het.} mostly has larger power than \textit{weak RP hom./het}, but for violation type \textit{sign(z)} and small IV strength $\pi$, \textit{weak RP hom./het.} actually has larger power for most settings (see Figure \ref{fig_SuppHA_pi_sZ}).

\subsection{Rejection Rate of $H_0(\beta_0)$ under $H_0$ as a function of $\beta_0$}\label{sec_SuppSimBeta}
In this subsection, we study the weak-identification robust procedures as tests of the point null hypothesis
$H_0(\beta_0):\E[y-x\beta_0-{\bm c}^T{\bm\theta}| {\bm z},{\bm c}]=0,$
as a function of the candidate value $\beta_0\in\mathbb R$. 
We generate data under $H_0$ as described in Section \ref{sec_DataGP}, i.e., with $s_{viol}=0$. 
Recall that the true coefficient on $x$ equals $\beta^\star=-1$.

For each configuration, we estimate the rejection frequency at level $\alpha=0.05$ of $H_0(\beta_0)$ over 1000 Monte Carlo repetitions for a grid with range $[-4, 2]$.
We vary $n\in\{100,300\}$, $n_C\in\{0,2\}$, $n_{IV}\in\{1,2,5\}$, $\pi\in\{0,0.5,1\}$, and consider both homoskedastic and heteroskedastic errors. We report results for \textit{weak RP hom.}, \textit{weak RP het.}, and \textit{ICM}.

\begin{figure}[tbp]
\centering
\includegraphics[width=0.8\textwidth]{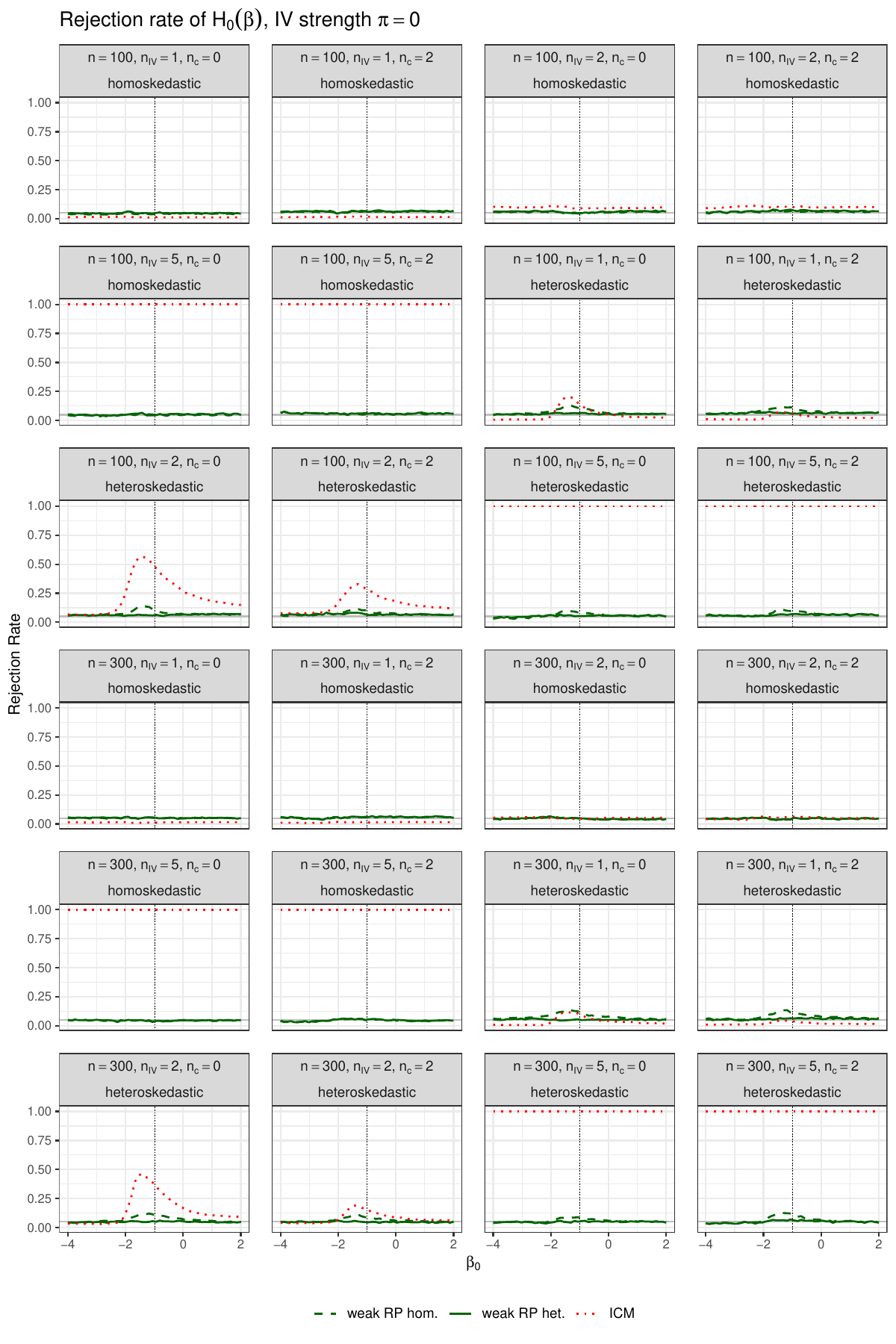}
\caption{Rejection frequency of the point null $H_0(\beta_0)$ under $H_0$ as a function of $\beta_0$, for IV strength $\pi=0$. The horizontal gray line marks $\alpha = 0.05$. The vertical gray line marks the true coefficient $\beta^\star=-1$. Panels correspond to combinations of $n\in\{100,300\}$, $n_{IV}\in\{1,2,5\}$, $n_C\in\{0,2\}$, and homoskedastic vs.\ heteroskedastic errors. Dashed curves coincide with the solid curves of the same color when not visible.}
\label{fig_SuppBeta_pi0}
\end{figure}

\begin{figure}[tbp]
\centering
\includegraphics[width=0.8\textwidth]{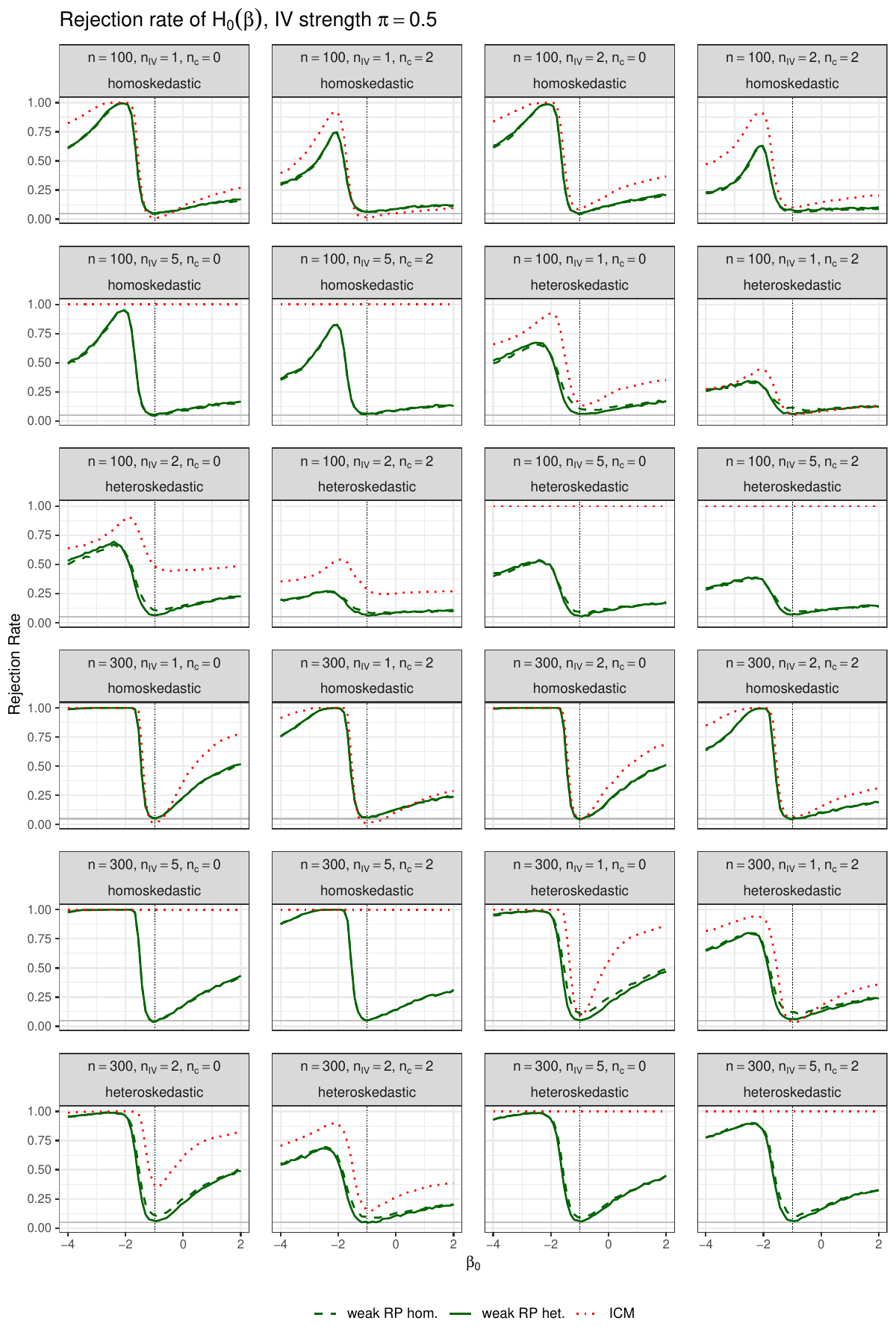}
\caption{Same as Figure \ref{fig_SuppBeta_pi0} but for IV strength $\pi=0.5$.}
\label{fig_SuppBeta_pi05}
\end{figure}

\begin{figure}[tbp]
\centering
\includegraphics[width=0.8\textwidth]{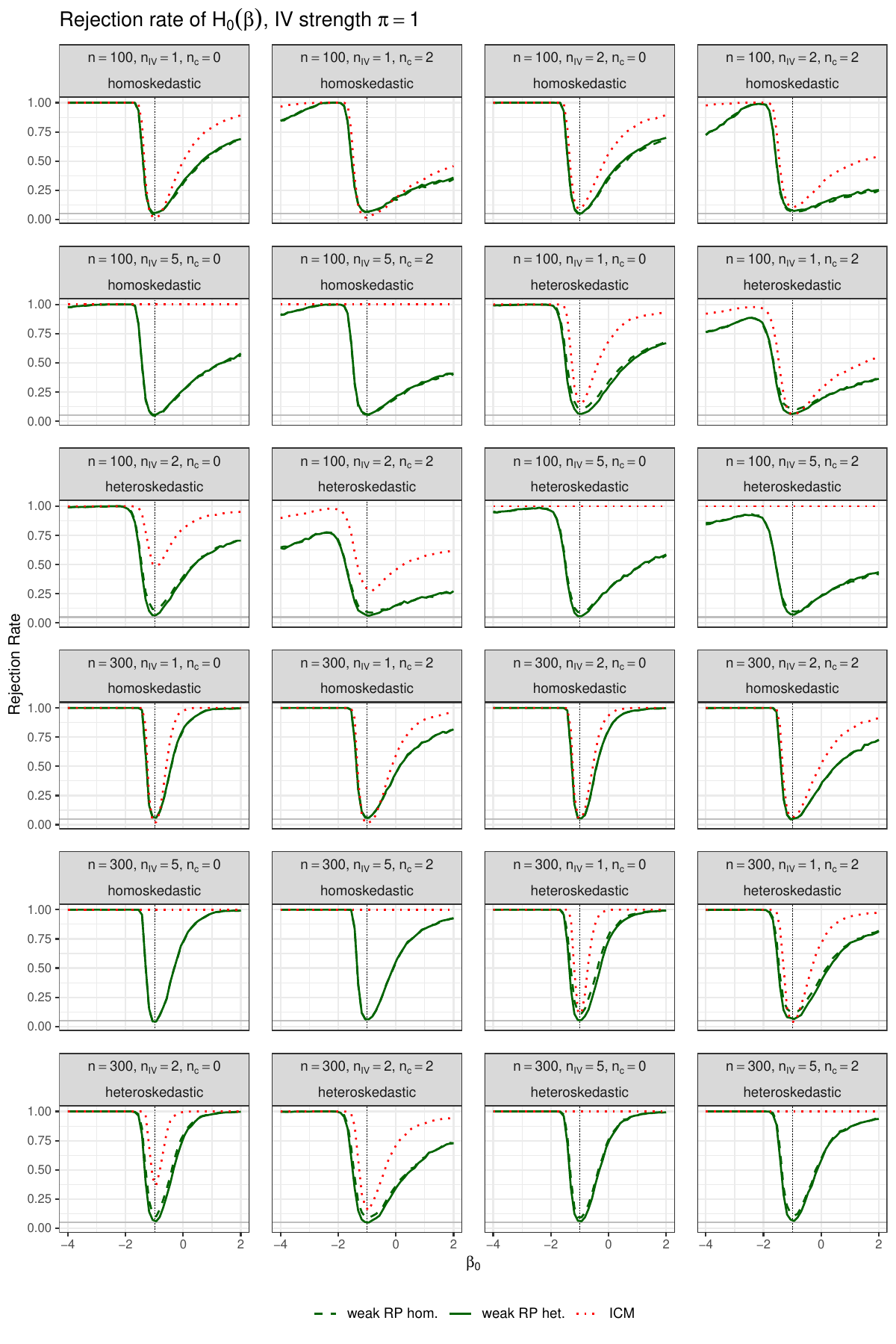}
\caption{Same as Figure \ref{fig_SuppBeta_pi0} but for IV strength $\pi=1$.}
\label{fig_SuppBeta_pi1}
\end{figure}

Figures \ref{fig_SuppBeta_pi0}--\ref{fig_SuppBeta_pi1} show the rejection frequency of the point null $H_0(\beta_0)$ as a function of $\beta_0$. 
By construction, rejection rates should be close to $\alpha$ when $\beta_0=\beta^\star=-1$, and should increase as $\beta_0$ moves away from $\beta^\star$.

When the IV strength is $\pi=0$ (Figure \ref{fig_SuppBeta_pi0}), rejection rates should remain close to $\alpha$ uniformly over $\beta_0$ by construction. While \textit{weak RP het.} behaves accordingly, \textit{weak RP hom.} exhibits overrejection at the true value $\beta_0=\beta^\star$ in some heteroskedastic settings. Moreover, \textit{ICM} shows substantial overrejection at $\beta_0=\beta^\star$ in several configurations, in particular when $n_{IV}=5$ or under heteroskedastic errors.

For $\pi\in\{0.5,1\}$ (Figures \ref{fig_SuppBeta_pi05} and \ref{fig_SuppBeta_pi1}), the rejection curves become more peaked around $\beta^\star=-1$, indicating increasing power to detect incorrect candidate values. In some configurations, \textit{ICM} exhibits slightly higher power than \textit{weak RP hom./het.} away from $\beta^\star$. However, the overrejection at the true value persists for \textit{ICM} in several settings, particularly for $n_{IV} = 5$ or heteroskedastic settings. In contrast, \textit{weak RP het.} maintains rejection rates close to the nominal level at $\beta_0=\beta^\star$ across the designs considered.

\section{Additional Empirical Illustration -- Protestantism and Literacy Rate}
Recall that for the empirical illustration on the dataset by \cite{CardCollege} in Section 5.4 in the main text, the approach by \cite{DieterleASimpleDiagnostic} was not applicable since the instrument was binary. As a comparison to their method, we now consider one of the datasets for which they found indications of model misspecification.
The data come from \cite{BeckerWasWeberWrong}, which consider (among many other questions) the effect of protestantism on literacy rate. As an instrument ${z}$, they use the distance to Wittenberg. There are 452 observations -- each corresponding to one county -- of the percentage of protestant population (endogenous regressor $x$), the literacy rate (outcome $y$), and diverse demographic controls ($\bm c$). We fit the same model as in the first row of Table 3.1 in \cite{DieterleASimpleDiagnostic}, which corresponds to the middle column of Table III in \cite{BeckerWasWeberWrong}. The data were obtained from the R package \texttt{ivdoctr} \citep{IVDOCTRPackage}. \cite{DieterleASimpleDiagnostic} noticed that the coefficients change drastically when fitting a model with instruments $(z, z^2)$ (where $z$ is the distance to Wittenberg) compared to a model with single instrument $z$, and the overidentifying J test based on $(z, z^2)$ rejects well-specification. Here, we additionally apply our residual prediction based tests and the additional tests as described in Section \ref{sec_ImplementationDetails}. As in Section 5.4 in the main text, we reduce the variability due to sample splitting by reporting doubled median p-values from 100 replications \citep{MeinshausenPValuesForHDRegression}. Figure \ref{fig_Weber} is the analogue of Figure 5 in the main text, plotting p-values (on log-scale) as a function of the candidate coefficient $\beta_0$. Note that \textit{RP hom./het.}, \textit{smooth asymp./boot.} and \textit{overid. J} appear as horizontal lines, whereas the weak-identification robust versions \textit{weak RP hom./het.} and \textit{ICM} vary with $\beta_0$.
\begin{figure}[t]
\centering
\includegraphics[width=0.8\textwidth]{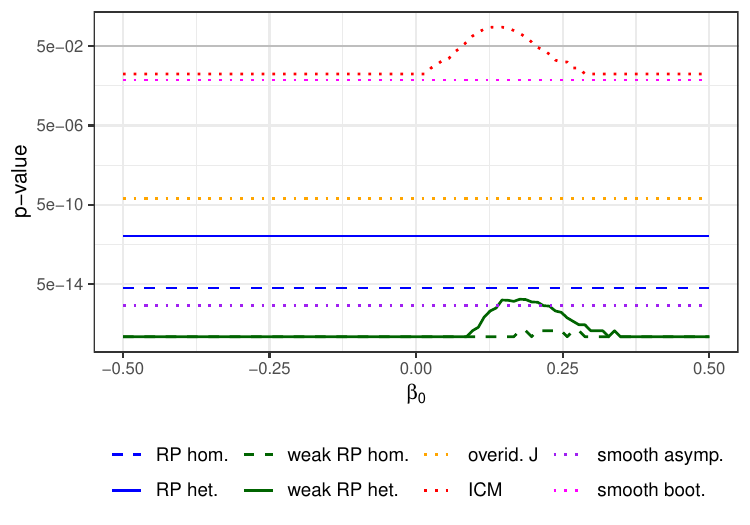}
\caption{P-values for the dataset from \cite{BeckerWasWeberWrong} as a function of $\beta_0$. When dashed lines are not visible, they are identical to the solid lines of the same color. The gray horizontal line indicates the significance level $\alpha = 0.05$.} 
\label{fig_Weber}
\end{figure}
We see that almost all of the tests reject well-specification with p-values far below the significance threshold $\alpha = 0.05$. Only the \textit{ICM} procedure produces a small range of $\beta_0$ such that the null hypothesis $H_0(\beta_0)$ is not rejected, hence not rejecting overall well-specification.

\section{Simulations with Clustered Data}\label{sec_SimCluster}
We now illustrate the theory from Section \ref{sec_GeneralTheory} by simulation. We consider independent clusters $I_1,\ldots,I_G$ of size $4$, so that $G = n/4$. 
We use exactly the same data-generating mechanism as in Section \ref{sec_DataGP}, with the only difference that we induce a clustered dependence structure in each $\mathcal N(0,1)$ distribution appearing in the structural assignments.
Specifically, whenever a standard normal variable $\xi_i \sim \mathcal N(0,1)$ appears in the construction in Section \ref{sec_DataGP}, we replace it by
$$
\xi_{ig} = \sqrt{s_{clust}}\,R_{g(i)} + \sqrt{1-s_{clust}}\,S_i,
$$
where $s_{clust}\in[0,1]$ controls the strength of the dependence, $g(i)\in\{1,\ldots,G\}$ denotes the cluster index of observation $i$, and $R_1,\ldots,R_G,S_1,\ldots,S_n$ are i.i.d.\ $\mathcal N(0,1)$. 
Thus, all observations within the same cluster share the component $R_{g(i)}$, while $S_i$ is unique to observation $i$.\footnote{By construction,
$\Var(\xi_{ig}) = s_{clust} + (1-s_{clust}) = 1$
so that the marginal distribution remains $\mathcal N(0,1)$ for all $s_{clust}$.}

The case $s_{clust}=0$ corresponds to i.i.d.\ data, whereas for $s_{clust}=1$ all observations within a cluster are identical.

We restrict attention to the case $n_{IV} = 1$, $n_C = 2$, and $\pi = 1$.

\subsection{Simulation under $H_0$}
For $n\in\{100,200,400,800\}$ and $s_{clust}\in\{0,0.1,\ldots,1\}$, we simulate 1000 datasets under $H_0$ (i.e., $s_{viol}=0$) and compute rejection rates at significance level $\alpha=0.05$ for \textit{RP hom.}, \textit{RP het.}, \textit{weak RP hom.}, \textit{weak RP het.}, and their cluster-robust counterparts \textit{RP clust.} and \textit{weak RP clust.}

In Figure \ref{fig_SuppCluster_H0}, we plot rejection rates against the strength of cluster dependence $s_{clust}$ for the different sample sizes. 
We see that for increasing $s_{clust}$, the rejection rates of the non-cluster-robust procedures rise above the nominal level $\alpha=0.05$, whereas the cluster-robust versions remain approximately constant around $\alpha=0.05$.

\begin{figure}[t]
\centering
\includegraphics[width=0.95\textwidth]{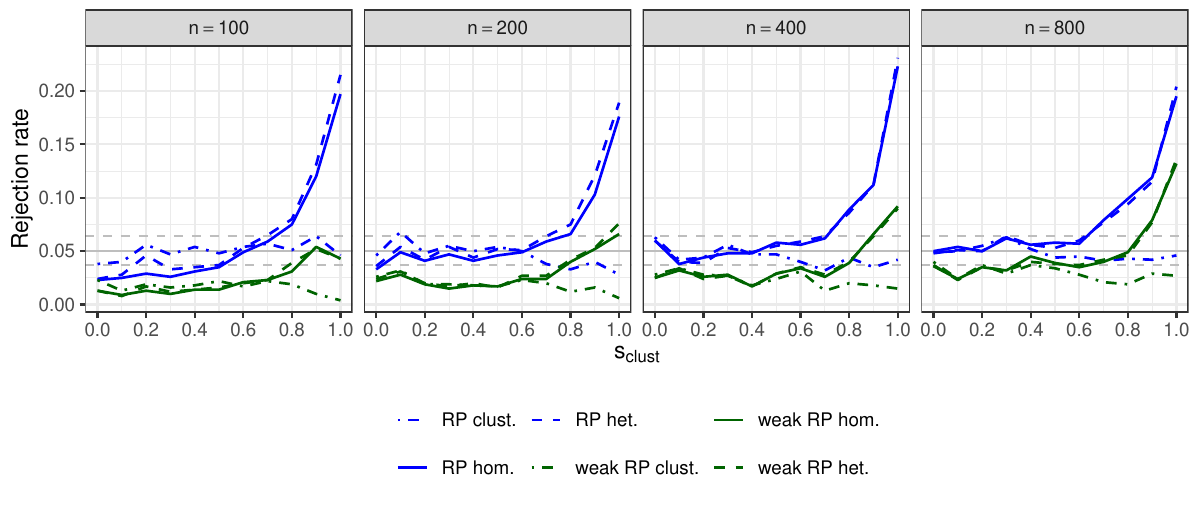}
\caption{Simulation results under $H_0$ with clustered dependence (cluster size 4), $n_{IV} = 1$ and $n_C = 2$ for $n\in\{100,200,400,800\}$ and varying strength $s_{clust}\in[0,1]$ of cluster dependence. The gray solid line marks $\alpha = 0.05$, the gray dashed lines show pointwise 95\% Monte Carlo bounds for rejection rates based on 1000 replications. Dashed curves coincide with the solid curves of the same color when not visible.}
\label{fig_SuppCluster_H0}
\end{figure}

\subsection{Simulation under $H_A$}
We next consider the same violations of well-specification as in Section \ref{sec_SuppSimHA}. 
We fix $n=300$ and simulate 1000 datasets for $s_{clust}\in\{0,0.1,\ldots,1\}$ with cluster size 4 and fixed violation strengths $s_{viol}\in\{1,3,4,8\}$ for \textit{z squared}, \textit{sign(z)}, \textit{misspec. squared}, and \textit{misspec. sign}, respectively.

Figure \ref{fig_SuppCluster_HA} reports rejection rates against $s_{clust}$ for the four violation types. When $s_{clust} = 0$, the clustered procedures exhibit approximately the same rejection rate as their non-cluster-robust counterparts; hence, there is no loss of efficiency when using the clustered procedures. On the other hand, for increasing $s_{clust}$, the rejection rates of the clustered procedures decrease more than the rejection rates of their non-cluster-robust counterparts. This is due to the fact that the cluster-robust variance estimator takes into account the decrease in the effective sample size as the within-cluster dependence $s_{clust}$ increases. This is not the case for the non-cluster-robust variance estimators.

\begin{figure}[t]
\centering
\includegraphics[width=0.95\textwidth]{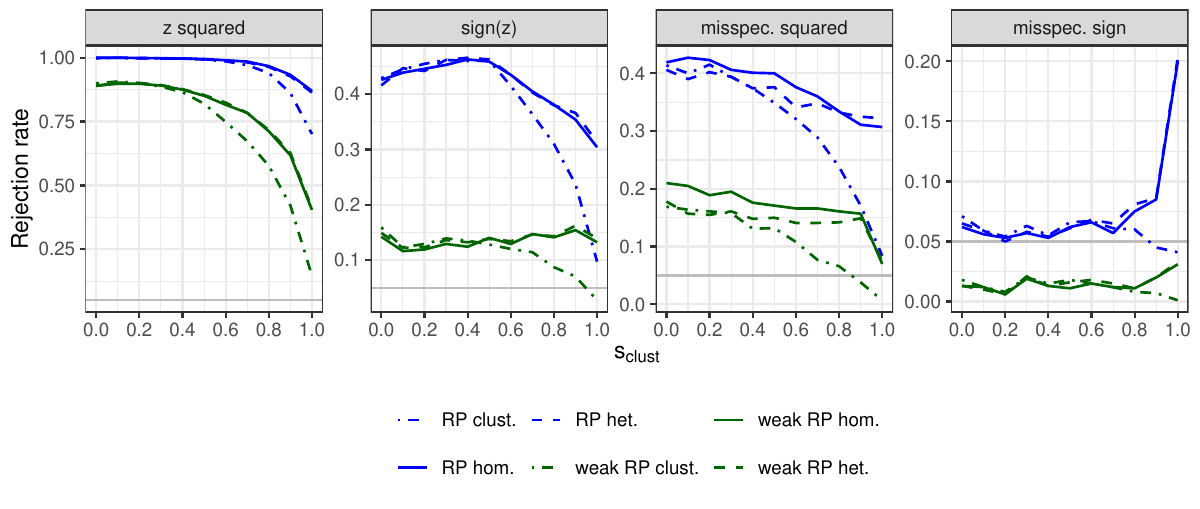}
\caption{Simulation results under $H_A$ with clustered dependence (cluster size 4), $n_{IV} = 1$ and $n_C = 2$, $n=300$, and varying strength $s_{clust}\in[0,1]$ of cluster dependence. Panels correspond to the four violation types. The gray solid line indicates the significance level $\alpha=0.05$. Dashed curves coincide with the solid curves of the same color when not visible.}
\label{fig_SuppCluster_HA}
\end{figure}

\section{Extension to Arbitrary Conditional Moment Restrictions}
In this section, we sketch how our approach can be generalized to test the well-specification of models defined by arbitrary conditional moment restrictions
\begin{equation}\label{eq_CMR}
    \E[g(\bm v, \bm \theta)|\bm z] = 0,
\end{equation}
where $\bm \theta\in\mathbb R^p$ is a finite-dimensional parameter, and $\bm v\in \mathbb R^l$ and $\bm z\in \mathbb R^d$ are random vectors that may have overlapping components. For simplicity, we assume $g$ to be univariate in the following, but extensions to vector-valued $g$ are straightforward. The linear instrumental variable model is a special case of \eqref{eq_CMR} by setting $\bm v = (y, \bm x)$, $\bm \theta = \bm \beta$, and $g(\bm v, \bm \theta) = y-\bm x^T\bm \beta$.

Such models are often estimated using the generalized method of moments (GMM) \citep[see, e.g.,][]{NeweyLargeSampleEstimation}, which we sketch in the following.
Let $h:\mathbb R^{d}\to \mathbb R^m$ ($m\geq p$), e.g., the identity $h(\bm z) = \bm z$. Then, the parameter $\theta$ solves
\begin{equation}\label{eq_GMM}
    \E[g(\bm v, \bm\theta)h(\bm z)]=0\,\, (\in\mathbb R^m).
\end{equation}
With $n_0$ i.i.d. samples $(\bm v_i, \bm z_i)_{i \in \mathcal D}$, let
\begin{equation}\label{eq_Defm}
    m_i(\bm \theta) = g(\bm v_i, \bm\theta) h(\bm z_i)\in \mathbb R^m, \quad \bar m_{\mathcal D}(\theta) = \frac{1}{n_0}\sum_{i\in \mathcal D} m_i(\theta)
\end{equation}
Then, $\bm \theta$ is estimated by
\begin{equation}\label{eq_GMMEstimator}
\hat{\bm \theta} = \arg\min_{\bm \theta} \bar m_{\mathcal D}(\bm \theta)^T{{\bm W}_n} \bar m_{\mathcal D}(\bm \theta),
\end{equation}
where ${{\bm W}_{\mathcal D}}\in \mathbb R^{m\times m}$ is a random matrix converging to some deterministic matrix $\bm W$. 

We can then use the same idea as for the linear instrumental variable model to construct a well-specification test, i.e., a test for
\begin{equation}\label{eq_H0GMM}
    H_0: \exists {\bm \theta}\in \mathbb R^p\, \text{s.t. } \E[g(\bm v, {\bm \theta})|\bm z] = 0.
\end{equation}

We can consider a test statistic of the form
$$T_n(w) = \frac{1}{\sqrt {n_0}} \sum_{i\in \mathcal D}g(\bm v_i, \hat {\bm \theta}) w(\bm z_i),$$
where $\hat {\bm \theta}$ is a GMM estimator of ${\bm \theta}$ and $w:\mathbb R^d\to [-1,1]$. We then learn a suitable weight function $\hat w$ on an independent auxiliary sample with machine learning, and a similar procedure to Procedure 1 in the main text can be applied to obtain asymptotically valid p-values.

In the following, we sketch how one may derive asymptotic normality of $T_n(w)$ (implicitly assuming suitable regularity conditions and -- for simplicity -- fixed $w$). By Taylor expansion,
\begin{align}
    T_n(w) &= \frac{1}{\sqrt{n_0}} \sum_{i\in \mathcal D}g(\bm v_i, \hat {\bm \theta}) w(\bm z_i)\nonumber\\
    &= \frac{1}{\sqrt{n_0}}\sum_{i \in \mathcal D}\left[g(\bm v_i, {\bm \theta}) + \frac{\partial g(\bm v_i, {\bm \theta})}{\partial \bm\theta^T}(\hat{\bm \theta} - {\bm \theta}) + o(\|\hat {\bm \theta} - {\bm \theta}\|_2)\right]w(\bm z_i).\label{eq_TaylorTn}
\end{align}
Moreover, one can show \citep[see, e.g.,][Section 3.3]{NeweyLargeSampleEstimation} 
\begin{equation}\label{eq_GMMLinear}
    \sqrt{n_0}(\hat {\bm \theta} - {\bm \theta}) = -({\bm D}^T{\bm W}{\bm D})^{-1}{\bm D}^T{\bm W}\frac{1}{\sqrt{n_0}}\sum_{i\in \mathcal D}h(\bm z_i)g(\bm v_i,\bm\theta) + o_P(1)
\end{equation}
with ${\bm D} = \EP\left[h(\bm z)\frac{\partial g(\bm v, {\bm \theta})}{\partial \bm\theta^T}\right]\in \mathbb R^{m\times p}$.
Plugging \eqref{eq_GMMLinear} into \eqref{eq_TaylorTn} and rearranging terms yields
\begin{align*}
     T_n(w) &= \frac{1}{\sqrt{n_0}}\sum_{i \in \mathcal D}g({\bm v}_i, {\bm \theta})w({\bm z}_i) + \frac{1}{\sqrt{n_0}}\sum_{i \in \mathcal D} w({\bm z}_i)\frac{\partial g(\bm v_i, {\bm \theta})}{\partial \bm\theta^T}(\hat{\bm \theta} - {\bm \theta}) + o_P(1)\\
     &=\frac{1}{\sqrt{n_0}}\sum_{i \in \mathcal D}g({\bm v}_i, {\bm \theta})w({\bm z}_i) \\
     & \qquad- \frac{1}{\sqrt{n_0}}\sum_{j \in \mathcal D} w(\bm z_j)\frac{\partial g(\bm v_j, {\bm \theta})}{\partial \bm\theta^T}({\bm D}^T{\bm W}{\bm D})^{-1}{\bm D}^T{\bm W}\frac{1}{n_0}\sum_{i \in \mathcal D}h(\bm z_i)g(\bm v_i,\bm\theta) + o_P(1)\\
     & = \frac{1}{\sqrt{n_0}}\sum_{i \in \mathcal D}g({\bm v}_i, {\bm \theta})w({\bm z}_i) \\
     & \qquad- \frac{1}{n_0}\sum_{j\in \mathcal D} w(\bm z_j)\frac{\partial g(\bm v_j, {\bm \theta})}{\partial \bm\theta^T}({\bm D}^T{\bm W}{\bm D})^{-1}{\bm D}^T{\bm W}\frac{1}{\sqrt{n_0}}\sum_{i\in \mathcal D}h(\bm z_i)g(\bm v_i,\bm\theta) + o_P(1)\\
     & =  \frac{1}{\sqrt{n_0}}\sum_{i \in \mathcal D}\left(g({\bm v}_i, {\bm \theta})w({\bm z}_i) - {\bm a}_w^T h(\bm z_i)g(\bm v_i,\bm\theta)\right) + o_P(1)\\
     & = \frac{1}{\sqrt{n_0}}\sum_{i \in \mathcal D}g({\bm v}_i, {\bm \theta})\left(w({\bm z}_i) - {\bm a}_w^T h({\bm z}_i)\right) + o_P(1)
\end{align*}
with 
$${\bm a}_w^T = \EP\left[w({\bm z})\frac{\partial g(\bm v, {\bm \theta})}{\partial \bm\theta^T}\right]({\bm D}^T{\bm W}{\bm D})^{-1}{\bm D}^T{\bm W}.$$
Under $H_0$, it holds that $\E[g({\bm v}_i,{\bm \theta})|{\bm z}_i]=0$ and hence by the central limit theorem, $T_n(w) \to\mathcal N(0, \sigma_w^2)$ with 
$$\sigma_w^2 = \EP[g({\bm v}_i, {\bm \theta})^2(w({\bm z}_i) - \bm a_w^T h({\bm z}_i))^2].$$
Then, $\sigma_w^2$ can be estimated by plugging in the corresponding sample quantities.

\end{document}